\documentclass[reqno]{amsart}

\usepackage{hyperref}
\usepackage{graphicx}
\usepackage{curves}
\usepackage{amssymb,amsmath,xcolor,color,framed,epstopdf,amsthm}
\usepackage[english]{babel}
\usepackage{enumerate}
\usepackage{mathrsfs}
\usepackage{empheq}
\usepackage{braket}
\hypersetup{colorlinks = true, linkcolor = {blue}}
\usepackage{cancel}
\usepackage{subcaption}
\usepackage{lipsum}
\usepackage{mathtools}
\usepackage{MnSymbol}
\usepackage{scalerel}
\usepackage[left=3cm, right=3cm, top=3cm, bottom=3cm]{geometry}
\usepackage{appendix}
\usepackage{amsmath,amssymb,extarrows} 
\usepackage{abstract}
\usepackage{tikz} 
\usetikzlibrary{shapes,snakes}
\usepackage{pgf}
\usepackage{pgflibraryarrows}
\usepackage{pgflibrarysnakes}
\usetikzlibrary{decorations.text}
\usepgfmodule{shapes}
\usetikzlibrary{decorations.pathmorphing}
\usetikzlibrary{decorations.markings}
\usetikzlibrary{patterns}
\usetikzlibrary{automata}
\usetikzlibrary{arrows.meta}
\usetikzlibrary{positioning}
\usepackage{pgfplots}
\usepackage{xcolor}
\usepackage{tikz}

\definecolor{softcyan}{RGB}{77, 192, 230}
\definecolor{softred}{RGB}{242, 115, 115}

\usetikzlibrary{backgrounds}
\tikzset{partial ellipse/.style args={#1:#2:#3}{insert path={+ (#1:#3) arc (#1:#2:#3)} }}
\tikzset{->-/.style={decoration={ markings, mark=at position #1 with {\arrow{>}}},postaction={decorate}}}
\tikzset{-<-/.style={decoration={ markings, mark=at position #1 with {\arrow{<}}},postaction={decorate}}}

\definecolor{green}{RGB}{0,153,0}

\newcommand*{\mailto}[1]{\href{mailto:#1}{\nolinkurl{#1}}}

\newtheorem{theorem}{Theorem}
\newtheorem{lemma}[theorem]{Lemma}

\newtheorem{remark}[theorem]{Remark}

\newtheorem{assumption}[theorem]{Assumption}
\newtheorem{problem}{RH problem}[section]
\numberwithin{theorem}{section}  
\newtheorem{proposition}[theorem]{Proposition}

\newcommand{\be}{\begin{equation}}
	\newcommand{\ee}{\end{equation}}
\newcommand{\bea}{\begin{eqnarray}}
	\newcommand{\eea}{\end{eqnarray}}

\def\XXint#1#2#3{{\setbox0=\hbox{$#1{#2#3}{\int}$}
		\vcenter{\hbox{$#2#3$}}\kern-.5\wd0}}

\def\1{\operatorname{Id}}

\def\Im{\operatorname{Im}}
\def\exp{\operatorname{exp}}

\numberwithin{equation}{section}

\let\oldtocsection=\tocsection
\let\oldtocsubsection=\tocsubsection
\let\oldtocsubsubsection=\tocsubsubsection
\renewcommand{\tocsection}[2]{\hspace{0em}\oldtocsection{#1}{#2}}
\renewcommand{\tocsubsection}[2]{\hspace{2em}\oldtocsubsection{#1}{#2}}
\renewcommand{\tocsubsubsection}[2]{\hspace{4em}\oldtocsubsubsection{#1}{#2}}

\begin{document}
	
	\begin{center}
		{\Large\bfseries Painlevé XXXIV Asymptotics for the Focusing mKdV Equation with Finite-Genus Background and Discrete Spectrum}\\[2em]
		
		{Ruihong Ma}\\[0.5ex]
		{\small\textit{School of Mathematical Sciences, Peking University, Beijing 100871, P.R. China}}\\
		{Engui Fan}\\[0.5ex]
		{\small\textit{School of Mathematical Sciences, Fudan University, Shanghai 200433, P.R. China}}\\[0.5ex]
		{\small\textit{Email: faneg@fudan.edu.cn}}
	\end{center}
	\begin{abstract}
We investigate the Cauchy problem for the focusing modified Korteweg--de Vries (mKdV) equation with finite-genus algebro-geometric quasi-periodic initial data. By applying the nonlinear steepest-descent method of Deift--Zhou to the associated Riemann--Hilbert (RH) problem, we derive the long-time asymptotics of the solution in the critical regime where complex stationary phase points coalesce with the endpoints of the finite-genus branch cuts. The collision is resolved via a local Painlev\'{e} XXXIV parametrix, and the discrete spectrum (breathers) is incorporated into the analysis. The resulting expansion is valid uniformly up to an error of order $\mathcal{O}(t^{-1/2})$. In this critical region, the leading-order term comprises the finite-genus algebro-geometric background together with breathers, whose parameters are slowly modulated by the background solution.

		{\bf Keywords:}    modified KdV equation,   algebro-geometric  background,  Riemann-Hilbert problem, Painlevé-XXXIV asymptotics.
		
		{\bf MSC:} 35Q55; 35P25; 35Q15; 35C20; 35G25.
	\end{abstract}
	
	\tableofcontents
		\section{Introduction}\label{sec1} 
We consider the Cauchy problem for the focusing modified Korteweg--de Vries equation (mKdV)
\begin{subequations}
	\begin{align}\label{eq:mkdv}
		&u_t(x,t)+u_{xxx}(x,t)+6u^2(x,t)u_x(x,t)=0,\quad t\ge0,\\[4pt]\label{eq:q_0}
		&u(x,0)=u_0(x),
	\end{align}
\end{subequations}
where the initial data $u_0(x)$ behaves asymptotically like a finite-genus algebro-geometric solution:
\begin{equation}\label{eq:a-alg-initi}
	u_0(x)\sim u^{(\mathrm{alg})}(x,0),\qquad x\to\pm\infty.
\end{equation}
Here
\begin{equation}\label{eq:q-alg}
	u^{(\mathrm{alg})}(x,t)=-\frac{i}{2}\operatorname{Im}\left(\sum_{j=0}^{n}E_j\right)\,e^{2i(p_0x+q_0t)}\,
	\frac{\Theta\bigl(\varphi(\infty)+d\bigr)\,\Theta\bigl(\varphi(\infty)-c(x,t;\phi)-d\bigr)}
	{\Theta\bigl(\varphi(\infty)-d\bigr)\,\Theta\bigl(\varphi(\infty)+c(x,t;\phi)+d\bigr)},
\end{equation}
with the vector $c(x,t;\phi)=(c_1,\dots,c_n)^{\mathrm{T}}\in\mathbb{C}^n$ defined by
\begin{equation}\label{eq:c-vector}
	c_j(x,t;\phi)=-\frac{x\,C_j^f+t\,C_j^g+\phi_j}{2\pi},\qquad j=1,\dots,n,
\end{equation}
and
\begin{equation}\label{eq:d-vector}
	d=\varphi(D)+K.
\end{equation}
Here $\Theta(z)$ is the Riemann theta function defined on the genus-$n$ Riemann surface $\mathcal{R}$ ($n\in\mathbb{N}_0$) associated with the hyperelliptic curve
\begin{equation}\label{eq:hyp-curve}
	w^2=P(z)=\prod_{j=0}^{n}\left(z-E_j\right)\left(z-\bar E_j\right),\quad E_j=a_j+ib_j,\; b_j>0,\; 
\end{equation}
where $z\in\mathbb{C}$ is a complex spectral parameter and $E_j,\bar E_j\in\mathbb{C}\setminus\mathbb{R}$. 
The Abel map $\varphi : \mathcal{R}\to\mathbb{C}^n$, the vector-valued Riemann constant $K\in\mathbb{C}^n$, and the pole divisor $D$ are defined on $\mathcal{R}$; their explicit constructions can be found in \cite{MaFan2026}. 
The constant vectors $C^f=(C_1^f,\dots,C_n^f)^{\mathrm{T}}$, $C^g=(C_1^g,\dots,C_n^g)^{\mathrm{T}}$, $\phi=(\phi_1,\dots,\phi_n)^{\mathrm{T}}\in\mathbb{C}^n$ and the constants $p_0,q_0\in\mathbb{R}$ are determined by the initial data.

To complete the formulation of the Cauchy problem, we impose the following boundary condition:
\begin{equation}\label{eq:finite-density}
	\int_{\mathbb{R}}\left|u(x,t)-u^{(\mathrm{alg})}(x,t)\right|\,\mathrm{d}x<\infty,\qquad t\geq0,
\end{equation}
which naturally extends the initial condition \eqref{eq:q_0} to $t>0$.

	The theory of algebro-geometric solutions for the modified Korteweg--de Vries (mKdV) equation originates from the seminal works of Novikov, Dubrovin, Its and Matveev in the mid-1970s on the periodic Cauchy problem for the KdV equation \cite{DN74,IM75,Dub75}. 
	By employing the spectral theory of periodic Schr\"odinger operators and the inverse spectral transform, they established that finite-gap potentials can be expressed explicitly via the Its--Matveev formula in terms of Riemann theta functions associated with hyperelliptic curves \cite{IM75,DMN76}. 
	This framework, later known as the Baker--Akhiezer formalism, was subsequently extended to the mKdV equation and other integrable systems such as the nonlinear Schr\"odinger, sine-Gordon and Toda lattice equations \cite{Kri77,BBEIM94,GH03}. 
	In particular, quasi-periodic and periodic solutions of the mKdV hierarchy are constructed by solving the Jacobi inversion problem on the underlying Riemann surface, and their real-valuedness, smoothness and isospectral deformations have been thoroughly investigated \cite{BBEIM94,GH03}. 
	More recently, algebro-geometric techniques have been further developed for higher-order matrix spectral problems and discrete integrable lattices, yielding theta-function representations for the coupled mKdV equation and related hierarchies \cite{CG90,GWC99,EGH17}.

	The rigorous analysis of the long-time behavior of the mKdV equation was pioneered by Deift and Zhou in their celebrated 1993 paper \cite{DZ93}, where they introduced the nonlinear steepest descent method for oscillatory RH problem. 
	This approach provides a systematic procedure to reduce the original RH problem to a sequence of localized model problems, whose solutions are expressed in terms of parabolic cylinder functions and Painlev\'e transcendents, thereby obtaining explicit asymptotic formulae with error estimates for Schwartz-class initial data \cite{DZ93,DIZ93}. 
	Subsequently, the method was refined and extended to weighted Sobolev spaces \cite{CL19}, to non-vanishing (finite-density) boundary conditions \cite{Fan23}, and to the soliton-resolution regime where the asymptotic profile comprises a superposition of solitons and dispersive radiation \cite{CL20,GR16}. 
	An important alternative development is the $\bar{\partial}$-steepest descent method introduced by McLaughlin and Miller \cite{MM08}, which relaxes the analyticity requirements on the scattering data and has been successfully applied to derive long-time asymptotics under minimal regularity assumptions \cite{DMM19}. 
	In addition to the zero-boundary-value problem, considerable attention has been devoted to the defocusing mKdV equation with nonzero boundary conditions, where the long-time behavior exhibits modulated oscillations in the left region, Painlev\'e-type transitions in the central region, and fast decay in the right region \cite{Fan23,Liu23}. Painlev\'{e} XXXIV asymptotics for the defocusing nonlinear {S}chr\"{o}dinger equation with a finite-genus algebro-geometric background \cite{FanLiYangZhang2026}.
	
The aim of the present work is to establish the long-time asymptotics for the Cauchy problem \eqref{eq:mkdv}--\eqref{eq:q_0} of the focusing modified Korteweg--de Vries (mKdV) equation by applying the nonlinear steepest-descent method. Our main result reads as follows.
\begin{theorem}\label{thm:main}
	Let $u^{(\mathrm{alg})}(x,t)$ be the genus-$n$ finite-genus algebro-geometric solution of the focusing mKdV equation \eqref{eq:mkdv}, given by \eqref{eq:q-alg}, with a nonempty discrete spectrum $\{\ell_1,\ell_2,\ell_3\}$ as in \eqref{eq:zero-decomposition}. Let $u(x,t)$ denote the solution of the Cauchy problem \eqref{eq:mkdv}--\eqref{eq:q_0} with initial data satisfying \eqref{eq:finite-density}. Then, as $t\to\infty$, within the transition regime
	\begin{equation}\label{eq:transition-regime}
		|\xi-\xi_{j_0}|\,t^{2/3}<C,\qquad \xi_{j_0}=-\frac{Q_{n+3}(E_{j_0})}{P_{n+1}(E_{j_0})},\quad j_0\in\mathcal{J}=\{0,1,2,\dots,n\},
	\end{equation}
	where $\xi=\frac{x}{t}$ and $C>0$ is an arbitrary fixed constant, the complex stationary phase points 
	$$\{z_{s_0}^{\mathrm{C}}(\xi),\bar{z}_{s_0}^{\mathrm{C}}(\xi),z_{s_0}^{\mathrm{C}}(\xi),\bar{z}_{s_0}^{\mathrm{C}}(\xi) \},\quad s_0\in\mathcal{S}=\{1,2,\dots,\tfrac{n+1}{4}\}$$
	defined in \eqref{eq:zero-jj} coalesce with the endpoints in
	\begin{equation}\label{eq:P2-set}
		\mathcal{P}_2=\{E_{j_0},\bar{E}_{j_0},E_{n-j_0},\bar{E}_{n-j_0}\},
	\end{equation}
	where $P_{n+1}(z)$ and $Q_{n+3}(z)$ are the polynomials associated with the background algebro-geometric structure via \eqref{eq:f(z)} and \eqref{eq:g(z)}. In this regime, $u(x,t)$ admits the asymptotics
	\begin{align}\label{eq:main-asymp}
		u(x,t)
		&=u^{(\mathrm{sol})}(x,t;\ell_3)+2i e^{2i(xp_0+tq_0)}\delta^{2}(\infty)e^{2ig(\infty)} \notag\\
		&\quad\times\Biggl(u^{(\mathrm{alg})}(x,t)+t^{-1/3}\sum_{p_0\in\mathcal{P}_2}i\widetilde H_{p_0}(p_0)
		\dfrac{a(\omega(p_0))}{|c_{j_0,2}(\xi, p_0)|^{2/3}}\Biggr)+\mathcal{O}\bigl(t^{-1/2}\bigr),
	\end{align}
	where $\widetilde H_{p_0}(p_0)$ is defined in \eqref{eq:tildehhhh} for $p_0\in\mathcal{P}_2$; $u^{(\mathrm{sol})}(x,t;\ell_3)$ is given by \eqref{eq:breather-recon}; $\delta(\infty)$ and $g(\infty)$ are as in \eqref{eq:delta-infty} and \eqref{eq:g-infty}; and
	\begin{equation}\label{eq:a-omega}
		a(\omega(p_0))=\int_{-\infty}^{\omega(p_0)}\left(u_{\mathrm{P}}(\zeta)+\frac{\zeta}{2}\right)\,d\zeta,
	\end{equation}
	with $\omega(p_0)$ and $c_{j_0,2}(\xi,p_0)$ defined in \eqref{eq:pmeggaaa}. Here, $u_{\mathrm{P}}(s)$ denotes the unique solution of the Painlev\'e-XXXIV equation
	\begin{equation}\label{eq:P34}
		u_{\mathrm{P}}''(s)=4u_{\mathrm{P}}(s)^2+2s\,u_{\mathrm{P}}(s)+\frac{u_{\mathrm{P}}'(s)^2-1/4}{2u_{\mathrm{P}}(s)},
	\end{equation}
	satisfying the boundary conditions
	\begin{equation}\label{eq:u-asymp-real}
		u_{\mathrm{P}}(\omega)=
		\begin{cases}
			\dfrac{1}{\sqrt{4\omega}}+\mathcal{O}\bigl(\omega^{-2}\bigr), & \omega\to+\infty, \\[10pt]
			-\dfrac{\omega}{2}+\mathcal{O}\bigl(\omega^{-2}\bigr), & \omega\to-\infty,
		\end{cases}
		\qquad\text{for } p_0=E_{j_0},E_{n-j_0},
	\end{equation}
	and
	\begin{equation}\label{eq:u-asymp-bar}
		u_{\mathrm{P}}(\omega)=
		\begin{cases}
			-\dfrac{1}{\sqrt{4\omega}}+\mathcal{O}\bigl(\omega^{-2}\bigr), & \omega\to+\infty, \\[10pt]
			-\dfrac{\omega}{2}+\mathcal{O}\bigl(\omega^{-2}\bigr), & \omega\to-\infty,
		\end{cases}
		\qquad\text{for } p_0=\bar E_{j_0},\bar E_{n-j_0}.
	\end{equation}
\end{theorem}
\subsection{Notation}

We introduce some notations that will be used throughout this paper.

\begin{enumerate}
	\item The three Pauli matrices are defined as usual by
	\begin{equation*}
		\sigma_1=\begin{pmatrix}
			0 & 1\\
			1 & 0
		\end{pmatrix},\qquad
		\sigma_2=\begin{pmatrix}
			0 & -i\\
			i & 0
		\end{pmatrix},\qquad
		\sigma_3=\begin{pmatrix}
			1 & 0\\
			0 & -1
		\end{pmatrix}.
	\end{equation*}
	
	\item The operator $\hat{\sigma}_3$ acts on a $2\times2$ matrix $A$ via the commutator
	\begin{equation*}
		e^{\hat{\sigma}_3}A=e^{\sigma_3}Ae^{-\sigma_3}.
	\end{equation*}
	In particular, $\hat{\sigma}_3A=[\sigma_3,A]=\sigma_3A-A\sigma_3$.
	
	\item The Cauchy operator associated with an oriented contour $\Sigma$ is defined by
	\begin{equation}\label{eq:Cauchy-operator}
		(C_{\Sigma}f)(z)=\frac{1}{2\pi i}\int_{\Sigma}\frac{f(s)}{s-z}\,\mathrm{d}s,\qquad z\in\mathbb{C}\setminus\Sigma.
	\end{equation}
\end{enumerate}
	
	\subsection{Summary of results}
The remainder of this paper is organized as follows. In Section~\ref{sec-2}, we provide a rigorous definition of the algebro-geometric solution $u^{(\mathrm{alg})}(x,t)$ and formulate the RH problem~\ref{RH2-2} characterizing the Cauchy problem \eqref{eq:mkdv}--\eqref{eq:q_0} via the spectral analysis of its Lax pair. In Section~\ref{secc3}, we employ the nonlinear steepest descent method to investigate the Painlev\'{e}-type asymptotics of $u(x,t)$ in the transition regime where four complex stationary phase points coalesce with the endpoints of the branch cuts of the algebro-geometric background.
	\section{The Riemann-Hilbert Formulism}\label{sec-2}
\subsection{The algebro-geometric background solution}

	The mKdV equation \eqref{eq:mkdv} is the compatibility condition for the pair of systems:
	\begin{align}\label{eq:lax-x}
		\Psi_x &= -i z \sigma_3 \Psi + U(x,t) \Psi,\\
		\label{eq:lax-t}
		\Psi_t &= -4iz^3\sigma_3\Psi + V(x,t)\Psi,
	\end{align}
	where $\Psi=\Psi(x,t,z)$ is a $2\times2$ matrix-valued function, $z \in \mathbb{C}$ is a spectral parameter and
	\begin{align}
		U(x,t)=\begin{pmatrix}
			0 & u \\
			-u & 0
		\end{pmatrix},\quad 
		V(x, t)=\begin{pmatrix}
			2 i z u^2 & 4 z^2 u+2 i z u_x-u_{xx}-2 u^3 \\
			-4 z^2 u+2 i z u_x+u_{xx}+2 u^3 & -2 i z u^2
		\end{pmatrix}.
	\end{align}
	
	Next, we present the construction of the planar Baker-Akhiezer (BA) function for the focusing mKdV equation, corresponding to the finite-genus solutions of the mKdV equation \eqref{eq:mkdv}. The BA function is described as the unimodular solution (i.e., $\det\Psi\equiv1$) of a matrix RH problem in the complex plane with piecewise constant jumps across a set of arcs.
	
In view of the central role of the RH problem in the inverse scattering method, it is natural to characterize the background solutions in terms of appropriate RH problems.
Let $\{E_j,\bar E_j\}_{j=0}^n$ be a set of points in the complex plane, and let $\mathcal{R}$ be the Riemann surface of genus $n$ defined by the hyperelliptic curve
\begin{equation}\label{eq:hyp-curve}
	w^2=P(z)=\prod_{j=0}^n\left(z-E_j\right)\left(z-\bar E_j\right),\quad E_j=a_j+ib_j,\; b_j>0. 
\end{equation}
Here $z\in\mathbb{C}$ is the spectral parameter and $E_j,\bar E_j\notin\mathbb{R}$. In this paper, we consider the odd genus case where the branch cuts are taken as $\Gamma_j=[\bar E_j,E_j]$ oriented upwards, with the canonical homology base $\{a_j,b_j\}_{j=1}^n$ shown in Figure~\ref{fig:contour}. These cuts are symmetric with respect to the imaginary axis in the sense that
\begin{equation}\label{eq:symmetry}
	E_{n-j}=-\bar E_j,\qquad j=0,1,\dots,n.
\end{equation}
		\begin{figure}[htp]
		\centering
		\begin{tikzpicture}[scale=0.85, >=Stealth, font=\footnotesize]
			% 水平实轴
%			\draw[very thick, ->] (-9,0) -- (9,0);
			
			% 左端点 E_{j_0} 处 (-4,1.5)
			\draw[very thick] (-2,0.75) -- (-2,1.5);
			\draw[very thick, ->] (-2,-1.5)--(-2,0.75);
		
						\draw[very thick] (2,0.75) -- (2,1.5);
			\draw[very thick, ->] (2,-1.5)--(2,0.75);

					\draw[very thick] (-6,0.75) -- (-6,1.5);
		\draw[very thick, ->] (-6,-1.5)--(-6,0.75);

	\draw[very thick] (6,0.75) -- (6,1.5);
\draw[very thick, ->] (6,-1.5)--(6,0.75);

	\draw[very thick] (-2,1.5) to[out=-30,in=30] (-2,-1.5);

	\draw[very thick] (2,1.5) to[out=-30,in=30] (2,-1.5);
		
				\draw[very thick] (6,1.5) to[out=-30,in=30] (6,-1.5);

					  \draw[very thick, 
					postaction={decorate}, 
					decoration={markings, 
						mark=at position 0.5 with {\arrow[very thick]{>}}
					}
					] (-2,-1.5) to[out=150,in=-150] (-2,1.5);
					
					% 端点标记（可选）
					\fill (-2,1.5) circle (1.5pt);
		
					\fill (-2,-1.5) circle (1.5pt) ;

						  \draw[very thick, 
					postaction={decorate}, 
					decoration={markings, 
						mark=at position 0.5 with {\arrow[very thick]{>}}
					}
					] (2,-1.5) to[out=150,in=-150] (2,1.5);
					
					% 端点标记（可选）
					\fill (2,1.5) circle (1.5pt);
					\fill (2,-1.5) circle (1.5pt);

						  \draw[very thick, 
					postaction={decorate}, 
					decoration={markings, 
						mark=at position 0.5 with {\arrow[very thick]{>}}
					}
					] (6,-1.5) to[out=150,in=-150] (6,1.5);
					
					% 端点标记（可选）
					\fill (6,1.5) circle (1.5pt) ;
					\fill (6,-1.5) circle (1.5pt);

						\fill (-6,1.5) circle (1.5pt) node[right] {$E_0$};
					\fill (-6,-1.5) circle (1.5pt) node[right] {$\bar E_0$};

						  \draw[very thick, 
					postaction={decorate}, 
					decoration={markings, 
						mark=at position 0.5 with {\arrow[very thick]{>}}
					}
					] (-6,0) to[out=50,in=130] (2,0);

					\draw[very thick, 
					postaction={decorate}, 
					decoration={markings, 
						mark=at position 0.5 with {\arrow[very thick]{>}}
					}
					] (-6,0) to[out=30,in=150] (-2,0);

					\draw[very thick, 
					postaction={decorate}, 
					decoration={markings, 
						mark=at position 0.5 with {\arrow[very thick]{>}}
					}
					] (-6,0) to[out=50,in=130] (6,0);
					
						\draw[very thick, dashed,
					postaction={decorate}, 
					] (-6,0) to[out=-50,in=-130] (6,0);
					
						\draw[very thick, dashed,
					postaction={decorate}, 
					] (-6,0) to[out=-30,in=-150] (-2,0);
					
						\draw[very thick, dashed,
					postaction={decorate}, 
					] (-6,0) to[out=-40,in=-140] (2,0);
					
			\node at (-1,0) {$a_1$};
				\node at (3,0) {$a_2$};
					\node at (7,0) {$a_3$};
				
					\node at (-4,1) {$b_1$};
				
					\node at (-2,2.1) {$b_2$};
					
						\node at (0,3) {$b_3$};
		\end{tikzpicture}
	\caption{The standard canonical homology base $\{a_j,b_j\}_{j=1}^{3}$ for genus $n=3$.}
		\label{fig:contour}
	\end{figure}
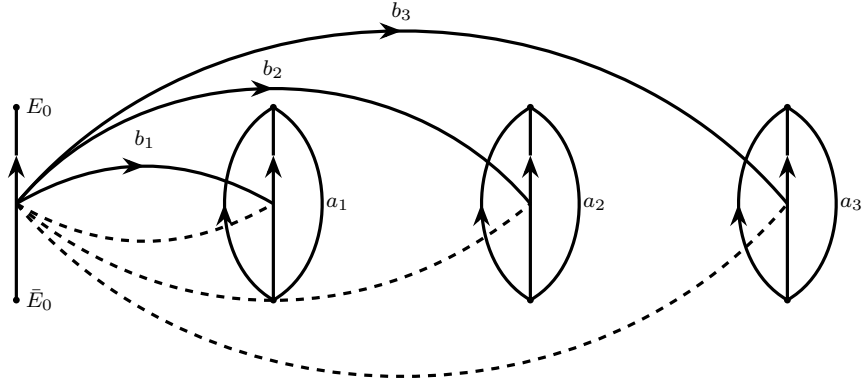		
	
Let us introduce the background eigenfunction $\Psi^{(\mathrm{alg})}(z)$, which takes the form
\[
\Psi^{(\mathrm{alg})}(z)=e^{\left(i p_0 x+i q_0 t\right) \sigma_3} N^{(\mathrm{alg})}(z)e^{-(i p(z) x+i q(z) t) \sigma_3},
\]
where $p_0, q_0\in\mathbb{R}$, and $p(z), q(z), N^{(\mathrm{alg})}(z)$ are determined as follows. The functions $p(z)$ and $q(z)$ are defined by
\begin{align}\label{eq:f(z)}
	p(z)&=\int_{\bar{E}_0}^z\frac{P_{n+1}(s)}{w(s)}\,ds=\int_{\bar{E}_0}^z \frac{s^{n+1}+{p}_n s^n+{p}_{n-1} s^{n-1}+\cdots+{p}_0}{w(s)} \,d s, 
	\\\label{eq:g(z)}
	q(z)&=\int_{\bar{E}_0}^z \frac{Q_{n+3}(s)}{w(s)}\,ds=\int_{\bar{E}_0}^z \frac{12 s^{n+3}+{q}_{n+2} s^{n+2}+{q}_{n+1} s^{n+1}+\cdots+{q}_0}{w(s)} \,d s,
\end{align}
and $N^{(\mathrm{alg})}(z)$ satisfies the following RH problem.
\begin{problem}\label{prob:1}
	Construct functions $p(z), q(z)$, meromorphic on $\mathbb{C}\setminus\Gamma=\cup_{j\in\mathcal{J}}\Gamma_j$, such that:
	\begin{enumerate}
		\item $p(z)$ and $q(z)$ are analytic in $\mathbb{C}\setminus\Gamma$ and satisfy
		\begin{equation}\label{eq:asymp-pq}
			p(z)=z+p_0+O(1/z),\quad q(z)=4z^3+q_0+O(1/z),\qquad z\to\infty;
		\end{equation}
		\item For $z\in\Gamma_j$, the boundary values satisfy the additive jump conditions
		\begin{equation}\label{eq:additive-jump}
			p_{+}(z)+p_{-}(z)=C_j^p, \quad
			q_{+}(z)+q_{-}(z)=C_j^q,\qquad C_0^p=C_0^q=0;
		\end{equation}
		\item $p(z)$ and $q(z)$ obey the Schwarz symmetry
		\begin{equation}\label{eq:symmetry-pq}
			\overline{p(\bar z)}=p(z),\quad \overline{q(\bar z)}=q(z).
		\end{equation}     
	\end{enumerate}
\end{problem}

\begin{lemma}\label{lem:pq-properties}
	The functions $p(z)$ and $q(z)$ defined in Problem~\ref{prob:1} have the following properties:
	\begin{enumerate}
		\item 
		The coefficients ${p}_n,\dots,{p}_0$ of $P_{n+1}$ and ${q}_{n+2},\dots,{q}_0$ of $Q_{n+3}$ are uniquely determined by the $a$-period normalization conditions
		\begin{equation}\label{eq:a-period}
			\oint_{a_j}\frac{P_{n+1}(s)}{w(s)}\,ds=0,\qquad 
			\oint_{a_j}\frac{Q_{n+3}(s)}{w(s)}\,ds=0,\qquad j=1,\dots,n.
		\end{equation}
		Furthermore, comparing the asymptotic expansion of the integrands at infinity with \eqref{eq:asymp-pq} yields the relations
		\begin{align}\label{eq:q-coeffs}
		&	{q}_{n+2}=6\bar E_1,\quad 	{p}_n=-\frac{\bar E_1}{2},\\
		&	{q}_{n+1}=-12\Bigl(\frac{4\bar E_2-\bar E_1^2}{8}\Bigr)+{q}_{n+2}\frac{\bar E_1}{2},\\
		&	{q}_{n}=-12c_{-1}-{q}_{n+2}\Bigl(\frac{4\bar E_2-\bar E_1^2}{8}\Bigr)+{q}_{n+1}\frac{\bar E_1}{2},
		\end{align}
		where
		\begin{align}
		&	\bar E_1=\sum_{j=0}^{n}(E_j+\bar{E}_j), \quad 
			c_{-1}=-\frac{\bar E_1^3-4\bar E_1\bar E_2+8\bar E_3}{16},\\
		&	\bar E_2=\sum_{j=0}^{n}E_j\bar{E}_j+\sum_{0\leq j<<k\leq n}(E_j+\bar{E}_j)(E_k+\bar{E}_k),\\
		&	\bar E_3=\sum_{0\leq j<k<l\leq n}(E_j+\bar{E}_j)(E_k+\bar{E}_k)(E_l+\bar{E}_l)
			+\sum_{j=0}^{n}\sum_{\substack{k=0\\k\ne j}}^{n}E_j\bar{E}_j(E_k+\bar{E}_k).
		\end{align}
		In the standard construction of finite-gap solutions, $P_{n+1}(z)$ and $Q_{n+3}(z)$ are required to have real coefficients so that the resulting solution $u(x,t)$ is real-valued. 
		
		\item The endpoint values are real: $\operatorname{Im}p(E_j)=\operatorname{Im}p(\bar E_j)=0$ and $\operatorname{Im}q(E_j)=\operatorname{Im}q(\bar E_j)=0$. Consequently, for $z\in\Gamma$,
		\begin{equation}\label{eq:im-jump}
			\operatorname{Im}p_+(z)=-\operatorname{Im}p_-(z)>0,\qquad 
			\operatorname{Im}q_+(z)=-\operatorname{Im}q_-(z).
		\end{equation}
		
		\item By the Schwarz reflection symmetry \eqref{eq:symmetry-pq}, $p(z)$ and $q(z)$ are single-valued and analytic in $\mathbb{C}\setminus\Gamma$.
	\end{enumerate}
\end{lemma}

\begin{proof}
	(i) The reality of $\operatorname{Im}p(E_j)$ follows from the fact that $E_j$ and $\bar E_j$ are the endpoints of the vertical branch cut $\Gamma_j$ and from the Schwarz symmetry \eqref{eq:symmetry-pq}. Differentiating the additive jump \eqref{eq:additive-jump} gives $dp_+(z)=-dp_-(z)$ for $z\in\Gamma_j$, which implies that $\operatorname{Im}p_+(z)+\operatorname{Im}p_-(z)$ is locally constant; since it vanishes at the endpoints, we obtain $\operatorname{Im}p_+(z)=-\operatorname{Im}p_-(z)$. The strict positivity on $\Gamma_j$ follows from the harmonic measure (or maximum principle) in the upper half-plane. 
	
	(ii) The single-valuedness is a consequence of \eqref{eq:symmetry-pq}: for any closed loop encircling a pair of conjugate branch cuts, the symmetry guarantees that the periods cancel, leaving no monodromy.
\end{proof}
	The background solution $N^{(\mathrm{alg})}(z)$ satisfying following RH problem:
	\begin{problem}Find a $2\times2$ matrix-valued function $N^{(\mathrm{alg})}(z)$, analytic in 
		$\mathbb{C}\setminus\Gamma$, such that:
		\begin{enumerate}
			\item $N^{(\mathrm{alg})}(z)=I+\mathcal{O}(z^{-1}),$\quad $|z|\to\infty$.
			\item 	For each $z\in\Gamma$, the boundary values $N^{(\mathrm{alg})}_\pm(z)$ satisfy the jump relation
			$$N^{(\mathrm{alg})}_+(z)=N^{(\mathrm{alg})}_-(z)V_j^{alg}(z),\quad z\in\Gamma_j$$
			where
			$$
			V^{alg}_j(z)=\begin{pmatrix}
				0 & i e^{-2i\pi c_j} \\
				i e^{2i\pi c_j} & 0
			\end{pmatrix}
			$$
			where $c=\left\{c_j \mid c_0=0, c_j=\frac{xC_j^p+tC_j^q+\phi_j}{2 \pi}\,\,{\rm for} \,\,1\leq j\leq n\right\}$
			and $\phi_j$ is real constant.
			\item The function $N^{(\mathrm{alg})}(z)$ has singularities at the endpoints of $\Gamma_j$ of order at most $|z-E_j|^{-1/4}$ or $|z-\bar{E}_j|^{-1/4}$.
		\end{enumerate}
	\end{problem}
	Finally, we define $N^{(\mathrm{alg})}(z)$ by 
	$$
	\begin{aligned}
		N^{(\mathrm{alg})}(z)= & \frac{1}{2}\left(\begin{array}{cc}
			\frac{1}{\Lambda_{11}(\infty)} & 0 \\
			0 & \frac{1}{\Lambda_{22}(\infty)}
		\end{array}\right)  \times\left(\begin{array}{cc}
			\left(\nu(z)+\nu^{-1}(z)\right) \Lambda_{11}(z) & \left(\nu(z)-\nu^{-1}(z)\right) \Lambda_{12}(z) \\
			\left(\nu(z)-\nu^{-1}(z)\right) \Lambda_{21}(z) & \left(\nu(z)+\nu^{-1}(z)\right) \Lambda_{22}(z)
		\end{array}\right),
	\end{aligned}
	$$
	where
	\begin{equation}\label{eq:kappa}
		\nu(z)=\prod_{k=0}^n \sqrt[4]{\frac{z-E_j}{z-\bar{E}_j}}, \quad z \in \mathbb{C} \setminus\Gamma,
	\end{equation}
	and for  $s=1,2$, we have
	$$
	\Lambda_{s1}(z)=\frac{\Theta\left(\varphi(z)+c+d_s\right)}{\Theta\left(\varphi(z)+d_s\right)}, \quad \Lambda_{s2}(z)=\frac{\Theta\left(-\varphi(z)+c+d_s\right)}{\Theta\left(-\varphi(z)+d_s\right)}, \quad z \in \mathbb{C} \setminus \Gamma.
	$$ Then the finite-genus algebro-geometric background solution is given by \eqref{eq:q-alg}.
\subsection{The basic Riemann--Hilbert problem}
	
	Let $\Psi^{\pm}(x,t,z)$ be the unique solutions of the Volterra integral equations
	\begin{align}\label{eq:volterra}
		\Psi^{\pm}(x,t,z) &= \Psi^{\mathrm{alg}}(x,t,z) \\\nonumber
		&\quad +\int_{\pm\infty}^{x} e^{ip_0(x-y)\sigma_3}\,N^{(\mathrm{alg})}(z)\,
		e^{-ip(z)(x-y)\sigma_3}\,N^{(\mathrm{alg})}(z)^{-1}\,
		U\bigl(u(y,t)-u^{(\mathrm{alg})}(y,t)\bigr)\,\Psi^{\pm}(y,t,z)\,\mathrm{d}y,
	\end{align}
	where $U(\cdot)$ is the potential matrix defined in \eqref{eq:lax-x}. Then the Jost solutions of the Lax pair \eqref{eq:lax-x}--\eqref{eq:lax-t} are recovered via
	\begin{equation}\label{eq:jost-recovery}
		M^{\pm}(x,t,z)=e^{-i(p_0x+q_0t)\sigma_3}\,\Psi^{\pm}(x,t,z)\,e^{i(p(z)x+q(z)t)\sigma_3},
	\end{equation}
	which satisfy the normalization condition
	\begin{equation}\label{eq:jost-normalization}
		M^{\pm}(x,t,z)\to I,\qquad z\to\infty,
	\end{equation}
	and constitute fundamental matrix solutions to \eqref{eq:lax-x}--\eqref{eq:lax-t}.
	
	Under the symmetry \eqref{eq:symmetry-pq} and the normalization \eqref{eq:a-period}, the zero level set of $\operatorname{Im}p(z)$,
	\begin{equation}\label{eq:Sigma0}
		\Sigma^0=\bigl\{z\in\mathbb{C}:\operatorname{Im}p(z)=0\bigr\},
	\end{equation}
	decomposes as
	\begin{equation}\label{eq:Sigma0-decomp}
		\Sigma^0=\mathbb{R}\cup\Gamma\cup\Sigma^{\mathrm{add}},
	\end{equation}
	where $\Sigma^{\mathrm{add}}$ consists of finitely many analytic arcs shown as black lines in Figures \ref{ttt} and \ref{tttt}. Since $\Sigma^0$ is a closed set, $\mathbb{C}\setminus\Sigma^0$ is open and decomposes uniquely into at most countably many disjoint open connected components:
	\[
	\mathbb{C}\setminus\Sigma^0=\bigsqcup_{k\in\Lambda}\Omega_k,
	\]
	where $\Lambda$ is a countable index set and each $\Omega_k$ is a maximal connected open subset of $\mathbb{C}\setminus\Sigma^0$. For any $k\in\Lambda$, $\operatorname{Im}p(z)$ never vanishes on $\Omega_k$ and maintains a definite sign; that is, either
	\[
	\operatorname{Im}p(z)>0,\quad\forall\,z\in\Omega_k,
	\qquad\text{or}\qquad
	\operatorname{Im}p(z)<0,\quad\forall\,z\in\Omega_k.
	\]
	Define the index subsets
	\[
	\Lambda^{\pm}:=\Bigl\{k\in\Lambda:\pm\operatorname{Im}p(z)>0,\;\forall\,z\in\Omega_k\Bigr\},
	\]
	and the open sets
	\[
	\Omega^{\pm}:=\bigcup_{k\in\Lambda^{\pm}}\Omega_k
	=\Bigl\{z\in\mathbb{C}\setminus\Gamma:\pm\operatorname{Im}p(z)>0\Bigr\}.
	\]
	Then $\Omega^{+}\cap\Omega^{-}=\varnothing$ and $\Omega^{+}\cup\Omega^{-}=\mathbb{C}\setminus\Sigma^0$. Note that $\Omega^{\pm}$ need not be the upper/lower half-planes; they are unions of several disconnected regions partitioned by $\Sigma^0$.

 Furthermore, the following symmetry relations hold:
\begin{equation}\label{eq:jost-symmetry}
	M^{\pm}(z)=\sigma_2\,\overline{M^{\pm}(\bar z)}\,\sigma_2,\qquad 
	M^{\pm}(z)=\sigma_1\,M^{\pm}(-z)\,\sigma_1.
\end{equation}
Thus, there exists a scattering matrix $S(z)$, independent of $(x,t)$, such that
\begin{equation}\label{eq:scattering}
	\Psi^{+}(x,t,z)=\Psi^{-}(x,t,z)\,S(z),\qquad z\in\Sigma^0,
\end{equation}
or equivalently,
\begin{equation}\label{eq:scattering-equiv}
	M^{+}(x,t,z)=M^{-}(x,t,z)\,e^{-i((p(z)-p_0)x+(q(z)-q_0)t)\hat\sigma_3}\,S(z),\qquad z\in\Sigma^0,
\end{equation}
 The scattering matrix takes the form
\begin{equation}\label{eq:S-matrix}
	S(z)=\begin{pmatrix}
		s_{11}(z) & s_{12}(z)\\[2pt]
		s_{21}(z) & s_{22}(z)
	\end{pmatrix},\qquad \det S(z)=1.
\end{equation}
Its entries can be expressed in terms of the Jost solutions as
\begin{equation}\label{eq:s-entries}
	\begin{aligned}
	&	s_{11}(z)=\det\bigl(M_1^{+}(z),M_2^{-}(z)\bigr),\quad
		s_{12}(z)=-e^{-2i((p(z)-p_0)x+(q(z)-q_0)t)}\det\bigl(M_1^{+}(z),M_1^{-}(z)\bigr),\\
		&
		s_{22}(z)=\det\bigl(M_1^{-}(z),M_2^{+}(z)\bigr),\quad
		s_{21}(z)=-e^{2i((p(z)-p_0)x+(q(z)-q_0)t)}\det\bigl(M_2^{-}(z),M_2^{+}(z)\bigr),
	\end{aligned}
\end{equation}
where $M_i$ denotes the $i$-th column of the matrix $M$. The reflection coefficient is defined by
\begin{equation}\label{eq:reflection}
	r(z)=-\frac{s_{21}(z)}{s_{11}(z)}.
\end{equation}
The Jost solutions satisfy the normalization condition
\[
\begin{pmatrix} M_{1}^{\mp}(z) & M_{2}^{\pm}(z) \end{pmatrix}\to I\quad\text{as }z\to\infty\text{ in }\mathbb{C}^{\pm},
\]
\begin{figure}[htp]
	\subfloat{\label{a}
		\begin{minipage}[t]{0.45\linewidth}
			\centering
			\includegraphics[width=2.5in]{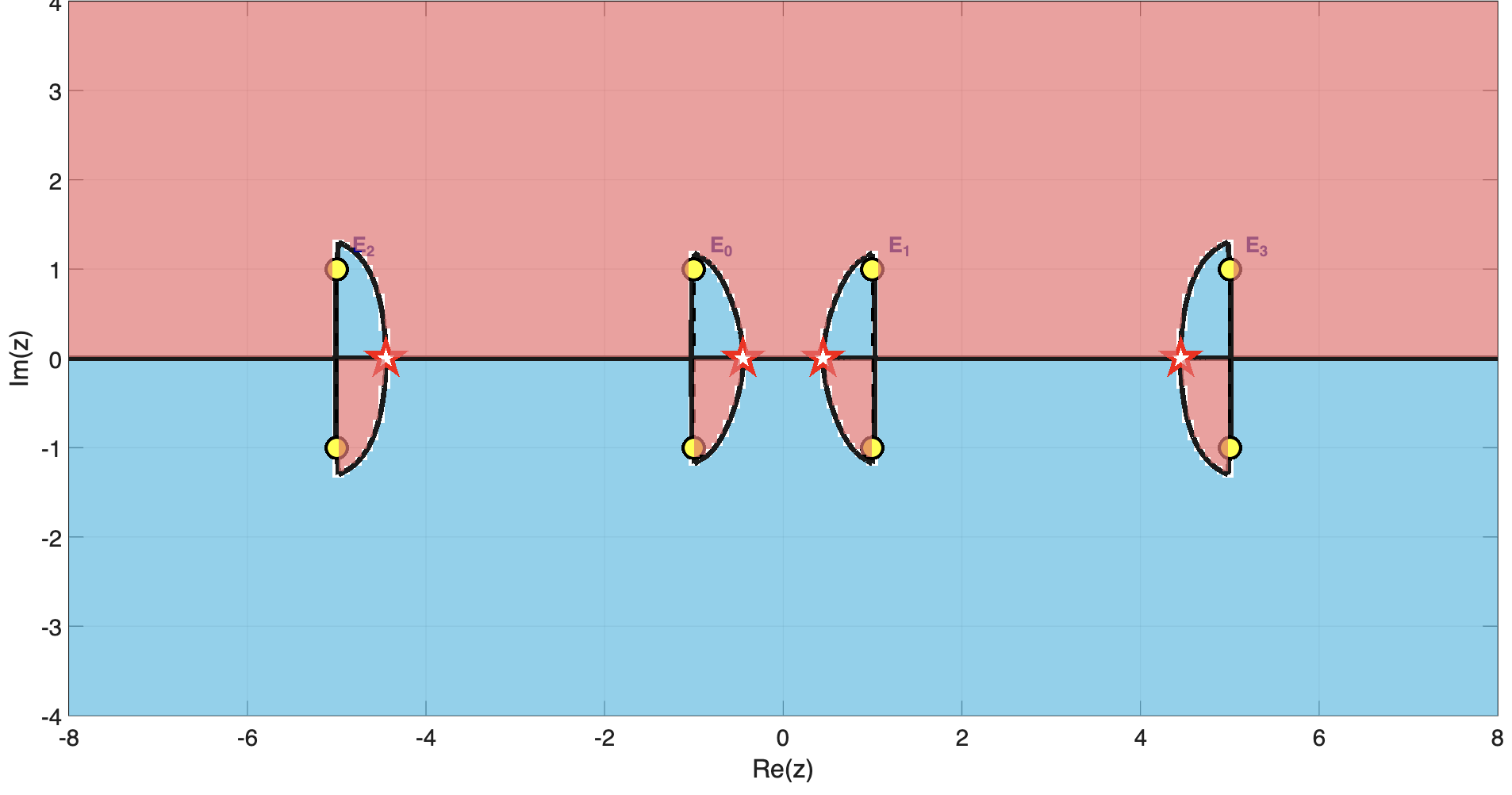}
		\end{minipage}
	}\quad
	\subfloat{\label{b}  % 已修正：原来是 c
		\begin{minipage}[t]{0.45\linewidth}
			\centering
			\includegraphics[width=2.5in]{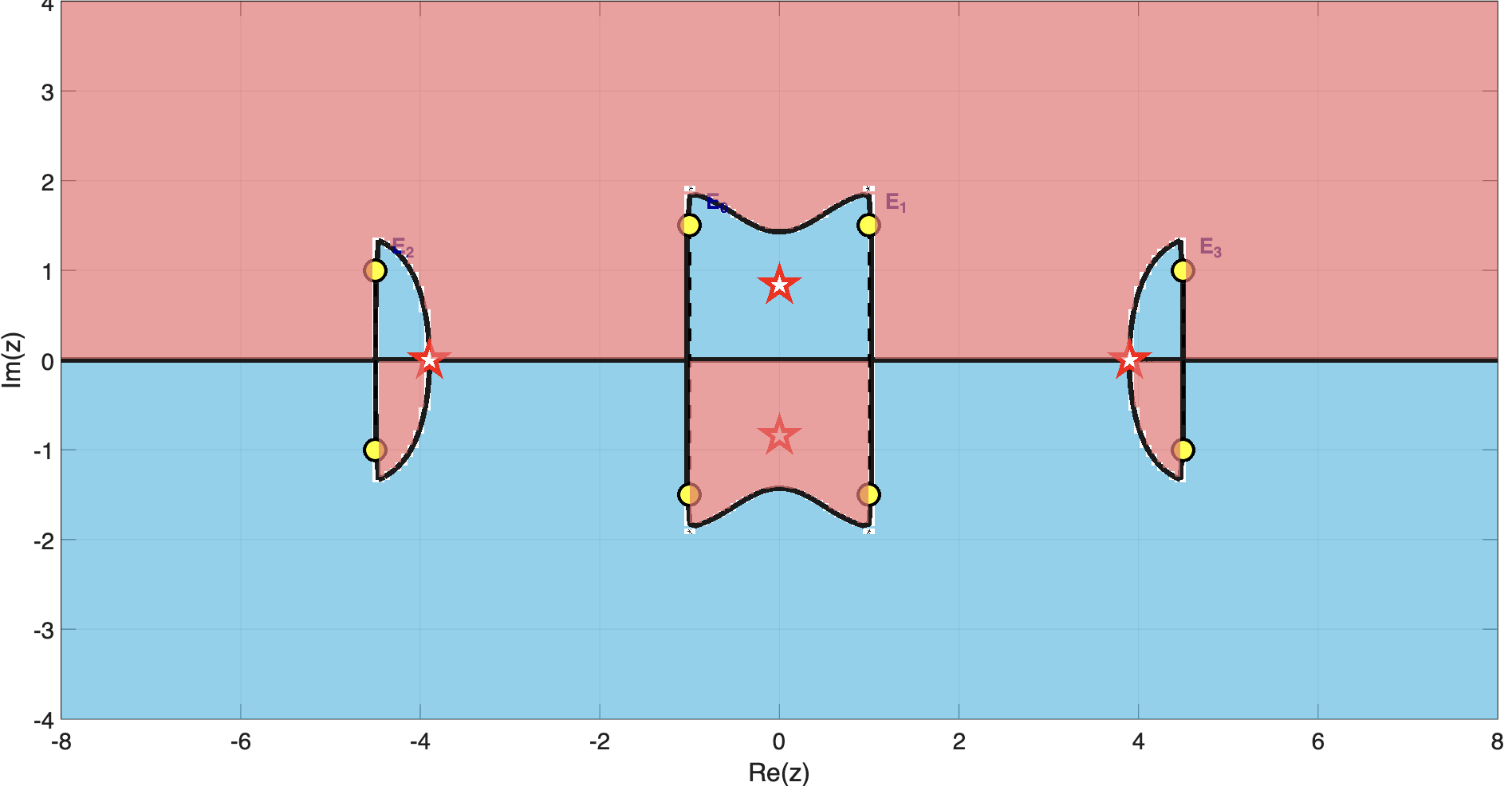}
		\end{minipage}
	}
	\caption{The set $\Sigma^0$ and the regions $\Omega^\pm$ in the complex plane. 
		The region $\Omega^+$ where $\operatorname{Im}p(z)>0$ is shaded in {red}, while the region $\Omega^-$ where $\operatorname{Im}p(z)<0$ is shaded in {blue}.}
	\label{ttt}
\end{figure}

\begin{figure}[htp]
	\subfloat{\label{a}
		\begin{minipage}[t]{0.45\linewidth}
			\centering
			\includegraphics[width=2.5in]{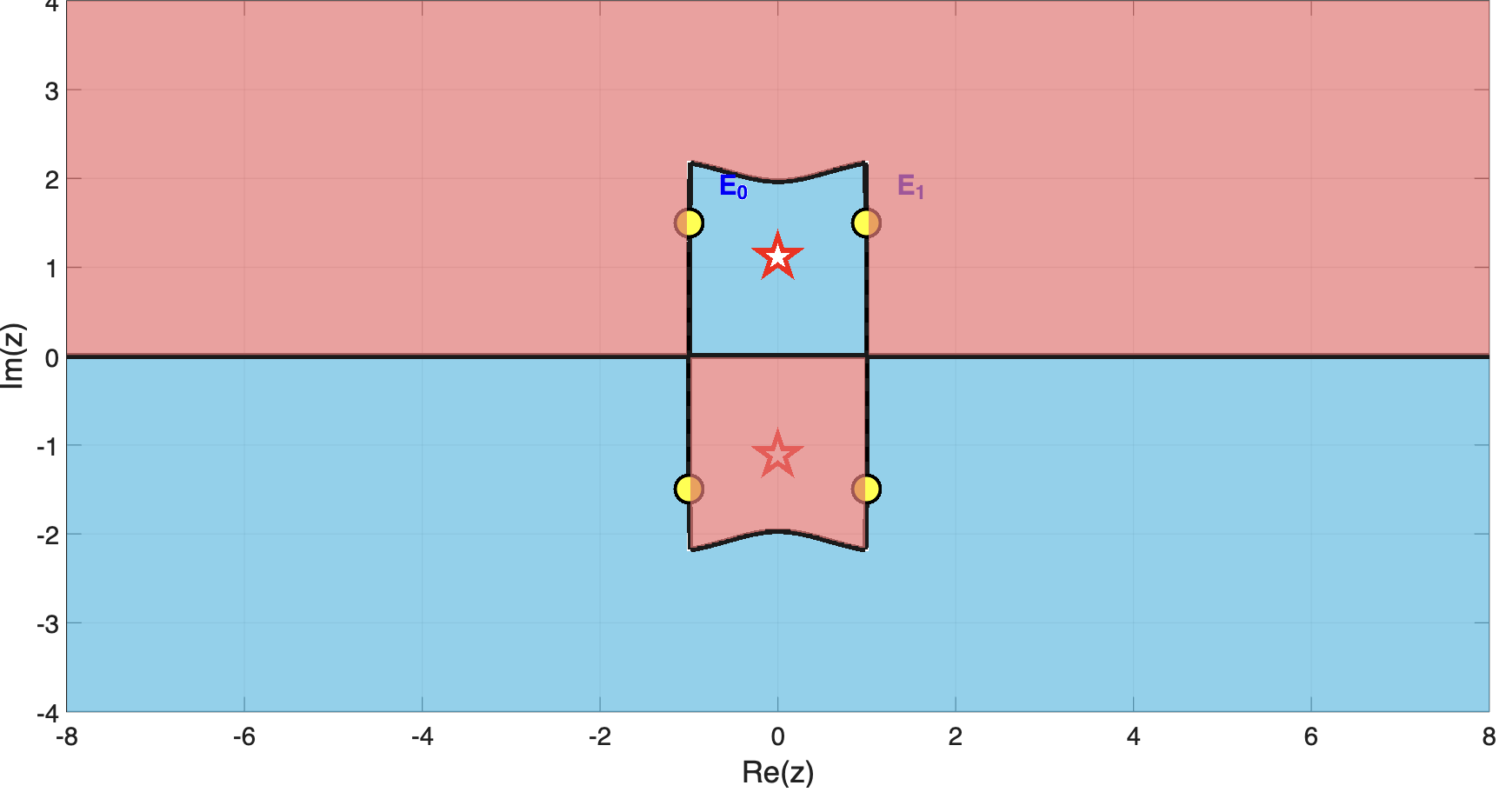}
		\end{minipage}
	}\quad
	\subfloat{\label{c}\begin{minipage}[t]{0.45\linewidth}
			\centering
			\includegraphics[width=2.5in]{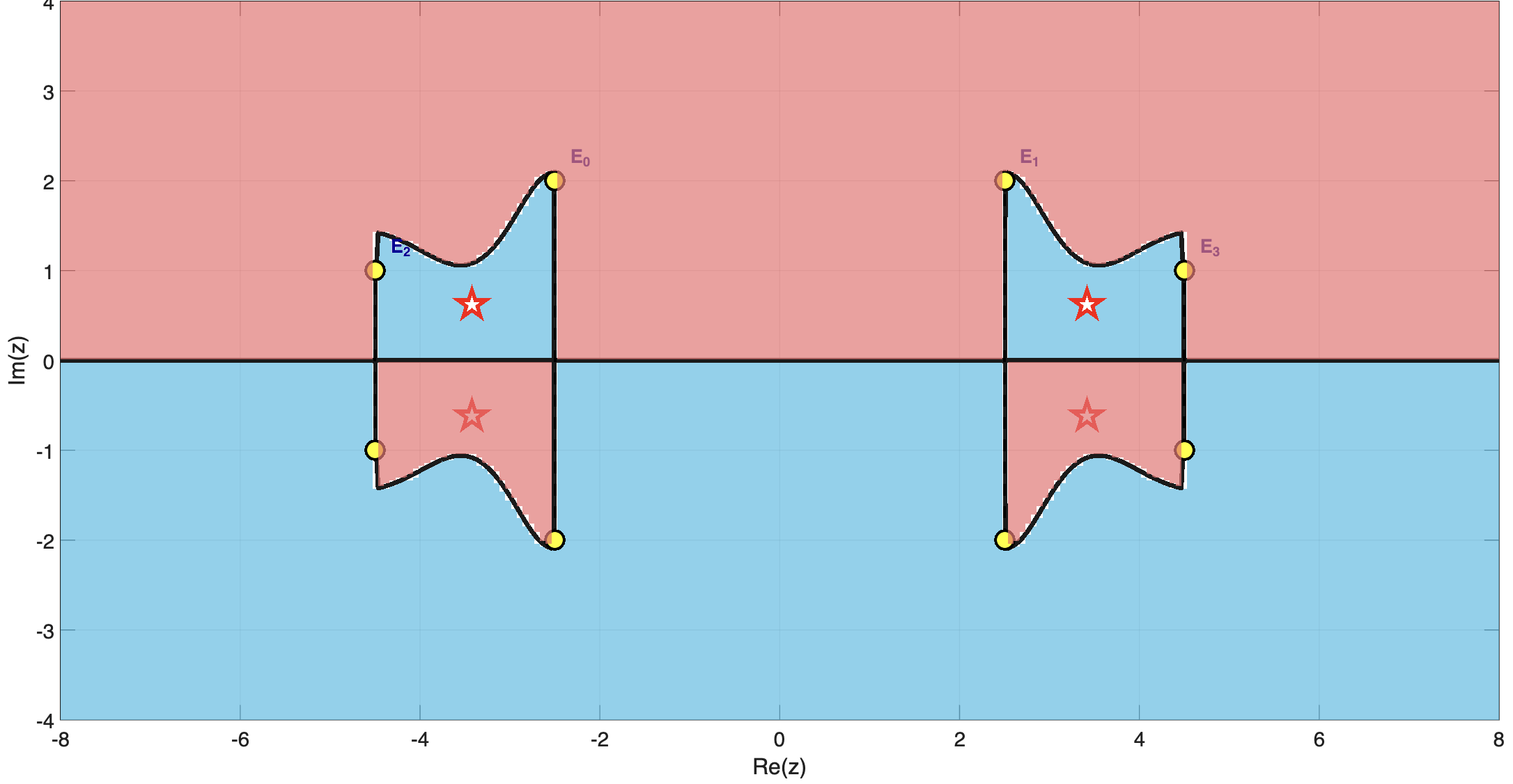}
	\end{minipage}}
	\caption{The probable set $\Sigma^0$ and the corresponding $\Omega^{ \pm}$ are shown for genus one (left) and genus three (right), with $\Omega^{+}$ in red and $\Omega^{-}$ in blue.}
	\label{tttt}
\end{figure}	
\begin{proposition}\label{prop:analyticity}
	The Jost solutions and scattering data possess the following properties:
	\begin{itemize}
		\item[(i)] \textit{Analyticity and Schwarz symmetry.}
		For any fixed $(x,t)\in\mathbb{R}^2$, as functions of $z$,
		\[
		M_1^{-}(z),\;M_2^{+}(z)\ \text{are analytic in }\ \Omega^{+},
		\qquad
		M_1^{+}(z),\;M_2^{-}(z)\ \text{are analytic in }\ \Omega^{-}.
		\]
		For $z\in\Sigma^0$, by the Schwarz symmetry,
		\begin{equation}\label{eq:schwarz-real}
			\overline{s_{22}(\bar{z})}=s_{11}(z),\qquad s_{21}(z)=-\overline{s_{12}(\bar{z})}.
		\end{equation}
		The transmission coefficient $s_{22}(z)$ is analytic in $\Omega^{+}$, while $s_{11}(z)$ is analytic in $\Omega^{-}$, with normalization
		\begin{equation}\label{eq:transmission-normalization}
			s_{11}(z),\, s_{22}(z)\to1\quad\text{as}\quad z\to\infty.
		\end{equation}
		Moreover,
		\begin{equation}\label{eq:reflection-decay}
			r(z),\, s_{12}(z),\, s_{21}(z)=\mathcal{O}(z^{-1}),\qquad z\to\infty.
		\end{equation}
		
		\item[(ii)] \textit{Jump relations on $\Gamma$.}
		For $z\in\Gamma_j$,
		\begin{equation}\label{eq:jump-Gamma-symmetry}
			\overline{s_{22}(\bar{z})}=s_{11}(z),\quad 
			s_{12}(z)=-e^{-2i\phi_j}\overline{s_{21}(\bar{z})},\quad 
			\overline{r_+(\bar{z})}=-e^{-2i\phi_j}r_-(z).
		\end{equation}
		Furthermore, the Jost solutions satisfy
		\begin{align}\label{eq:jump-jost-plus}
			M^+_{1+}(z)&=ie^{2i\phi_j}M^+_{2-}(z), & M^+_{2+}(z)&=ie^{-2i\phi_j}M^+_{1-}(z),\\
			\label{eq:jump-jost-minus}
			M^-_{1+}(z)&=ie^{2i\phi_j}M^-_{2-}(z), & M^-_{2+}(z)&=ie^{-2i\phi_j}M^-_{1-}(z).
		\end{align}
		
		\item[(iii)] \textit{Endpoint behavior.}
		For the generic case, let $\beta\in\{E_j,\bar{E}_j\}_{j=0}^{n}$. Then, as $z\to\beta$,
		\begin{equation}\label{eq:endpoint-asymp}
			s_{11}(z),\, s_{12}(z)=\mathcal{O}\bigl((z-\beta)^{-1/2}\bigr),\qquad
			r(z)=e^{i\phi_j}+\mathcal{O}\bigl((z-\beta)^{1/2}\bigr).
		\end{equation}
	\end{itemize}
\end{proposition}
\begin{proof}
	(i) Analyticity in $\Omega^{\pm}$ follows from the Volterra equations \eqref{eq:volterra} with holomorphic kernels. The Schwarz symmetry \eqref{eq:schwarz-real} on $\Sigma^0$ is inherited from the Lax pair \eqref{eq:lax-x}--\eqref{eq:lax-t} under conjugation. The normalization \eqref{eq:transmission-normalization} and decay \eqref{eq:reflection-decay} follow from \eqref{eq:jost-normalization} and the large-$z$ Volterra expansion.
	
	(ii) On $\Gamma_j$, the background eigenfunction $\Psi^{(\mathrm{alg})}(z)$ satisfies the jump conditions of RH Problem~\ref{prob:1}. Substituting into the scattering relation \eqref{eq:scattering} and using \eqref{eq:jost-recovery} yields \eqref{eq:jump-jost-plus}--\eqref{eq:jump-jost-minus}. The symmetries \eqref{eq:jump-Gamma-symmetry} follow by analytic continuation of \eqref{eq:schwarz-real} to $\Gamma$ and the anti-diagonal structure of the background jump matrix.
	
	(iii) Near $\beta\in\{E_j,\bar E_j\}$, the curve \eqref{eq:hyp-curve} gives $w(z)\sim(z-\beta)^{1/2}$, whence the Abelian integrals $p(z),q(z)$ in \eqref{eq:f(z)}--\eqref{eq:g(z)} inherit the square-root branch. This singularity propagates through the Volterra kernels in \eqref{eq:volterra}, yielding \eqref{eq:endpoint-asymp}. The phase $e^{i\phi_j}$ in $r(z)$ arises from $|r(z)|=1$ on $\Gamma$.
\end{proof}
\begin{assumption}\label{ass:discrete-spectrum}
	We consider initial data $u_0(x)$ satisfying \eqref{eq:norm-condition} such that the analytic scattering coefficient $s_{11}(z)$ admits the zero set
	\begin{equation}\label{eq:zero-decomposition}
		\mathcal{Z}=\mathcal{Z}_1\cup\mathcal{Z}_2\cup\mathcal{Z}_3,
	\end{equation}
	where 
	\begin{equation}\label{eq:Zj-def}
		\mathcal{Z}_j=\left\{\ell_j,\bar{\ell}_j,-\ell_j,-\bar{\ell}_j\right\},\qquad j=1,2,3,
	\end{equation}
	and each $\ell_j=\xi_j+i\eta_j$ lies in the first quadrant with $\xi_j>0$, $\eta_j>0$.
	The zeros are distributed according to the sign of $\operatorname{Im}p(z)$ as follows:
	\begin{itemize}
		\item $\mathcal{Z}_1$ consists of simple zeros with $\ell_1$ located in the region $\{\operatorname{Im}p(z)<0\}$,
		\item $\mathcal{Z}_2$ consists of simple zeros with $\ell_2$ located in the region $\{\operatorname{Im}p(z)>0\}$,
		\item $\mathcal{Z}_3$ consists of simple zeros with $\ell_3$ located on the zero level set $\{\operatorname{Im}p(z)=0\}$.
	\end{itemize}
	Each family $\mathcal{Z}_j$ gives rise to a distinct type of breather solution, as illustrated in Figure~\ref{ttt6tttttttt}.
\end{assumption}

\begin{remark}\label{rem:norming-constant}
	For any zero $\ell_j\in\mathcal{Z}$, the relation $s_{22}(\ell_j)=0$ is equivalent to $\overline{s_{11}(\bar{\ell}_j)}=0$ by the Schwarz symmetry \eqref{eq:schwarz-real}. Consequently, the first column of $M^-(z)$ and the second column of $M^+(z)$, i.e., $M_1^-(z)$ and $M_2^+(z)$, become linearly dependent at $z=\ell_j$; specifically, there exists a norming constant $c_j\in\mathbb{C}\setminus\{0\}$ such that
	\begin{equation}\label{eq:norming-constant}
		M^-_1(\ell_j)=c_j\,M^+_2(\ell_j).
	\end{equation}
	Moreover, the corresponding bound-state solution decays exponentially to the background $u^{(\mathrm{alg})}(x,t)$ as $x\to\pm\infty$.
\end{remark}

Define the sectionally analytic $2\times2$ matrix
\begin{equation}\label{eq:RH-con}
	N(z)=N(x,t;z):=
	\begin{cases}
		\left(M_1^{-}(z)\;\;\dfrac{M_2^{+}(z)}{\overline{a(\bar{z})}}\right),
		& z\in\Omega^{+},\\[2.5ex]
		\left(\dfrac{M_1^{+}(z)}{a(z)}\;\;M_2^{-}(z)\right),
		& z\in\Omega^{-}.
	\end{cases}
\end{equation}
Then the matrix $N(z)$ defined by \eqref{eq:RH-con} satisfies the following 
RH problem.
\begin{problem}\label{RH2-2}
	Find a $2\times2$ matrix-valued function $N(z)$, analytic in 
	$\mathbb{C}\setminus\Sigma^{0}$, such that:
	\begin{enumerate}
		\item $N(z)=I+\mathcal{O}(z^{-1}),\qquad |z|\to\infty$.
		
		\item For each $z\in\mathbb{R}\cup\Gamma$, the boundary values 
		$N_{\pm}(z)$ satisfy
		\begin{equation}\label{eq:jump}
			N_{+}(z)=N_{-}(z)J(z),
		\end{equation}
		where
		\begin{equation}\label{eq:jummp-2}
			J(z)=
			\begin{cases}
				\begin{pmatrix}
					1 & \overline{r(\bar{z})}\,e^{-2it\theta(z)}\\[4pt]
					r(z)\,e^{2it\theta(z)} & 1+|r(z)|^{2}
				\end{pmatrix},
				& z\in \mathbb{R}\cup \Sigma^{\mathrm{add}},\\[6pt]
				\begin{pmatrix}
					0 & ie^{-2\pi ic_j}\\
					ie^{2\pi ic_j} & 0
				\end{pmatrix},
				& z\in\Gamma_j,
			\end{cases}
		\end{equation}
		with
		\begin{equation}\label{eq:theta}
			\theta(z;\xi):=\xi\bigl(p(z)-p_0\bigr)+\bigl(q(z)-q_0\bigr),\qquad \xi=x/t.
		\end{equation}
		
		\item The following residue conditions for $N(z)$ hold:
	\begin{subequations}\label{eq:res-condition}
		\begin{align}
			\underset{z=\ell_j}{\operatorname{Res}}\,N(z)
			&=\lim_{z\to\ell_j}N(z)
			\begin{pmatrix}
				0 & -e^{-2it\theta(z)}c_j\\[2pt]
				0 & 0
			\end{pmatrix},\\[4pt]
			\underset{z=\bar{\ell}_j}{\operatorname{Res}}\,N(z)
			&=\lim_{z\to\bar{\ell}_j}N(z)
			\begin{pmatrix}
				0 & 0\\[2pt]
				e^{2it\theta(z)}\bar{c}_j & 0
			\end{pmatrix},\\[4pt]
			\underset{z=-\ell_j}{\operatorname{Res}}\,N(z)
			&=\lim_{z\to-\ell_j}N(z)
			\begin{pmatrix}
				0 & 0\\[2pt]
				-e^{2it\theta(z)}c_j & 0
			\end{pmatrix},\\[4pt]
			\underset{z=-\bar{\ell}_j}{\operatorname{Res}}\,N(z)
			&=\lim_{z\to-\bar{\ell}_j}N(z)
			\begin{pmatrix}
				0 & e^{-2it\theta(z)}\bar{c}_j\\[2pt]
				0 & 0
			\end{pmatrix}.
		\end{align}
	\end{subequations}
		\item For each endpoint $\beta\in\{E_j,\bar{E}_j\}_{j=0}^{n}$,
		\begin{equation}\label{eq:delta-endpoint}
			\delta(z)=\mathcal{O}\bigl((z-\beta)^{1/2}\bigr),\qquad z\to \beta.
		\end{equation}
	\end{enumerate}
\end{problem}
Then the soliton solution $q(x,t)$ is determined by
\begin{equation}\label{eq:2.28}
	q(x,t)=2i e^{2i(xp_0+tq_0)}\lim_{|z|\to\infty}(N(z))_{12}.
\end{equation}
\begin{remark}
	As $z\to\infty$, the phase function admits the asymptotic expansion
	\begin{equation}
		\theta(z)=4z^3+\xi z+\theta_c(\infty,\xi)+\mathcal{O}(z^{-1}),
	\end{equation}
	where $\theta_c(\infty,\xi)$ is a real-valued function of $\xi$. Consequently, the zero level set $\{z\in\mathbb{C}:\operatorname{Im}\theta(z)=0\}$ comprises three unbounded branches: the real axis, together with two symmetric curves. These curves are asymptotic to the vertical lines $\operatorname{Re}z=\pm\sqrt{-\xi/12}$ when $\xi<0$, and to the horizontal lines $\operatorname{Im}z=\pm\sqrt{-\xi/12}$ when $\xi>0$.
\end{remark}
\subsection{Stationary phase analysis and factorizations of the jump matrices}
The long-time asymptotic behavior is governed by the stationary phase points of the phase function $\theta(z;\xi)$, which are defined as the roots of
\[
\partial_z\theta(z;\xi)=\xi\,p'(z)+q'(z)=0.
\]
Equivalently, these roots satisfy the algebraic equation
\[
h(z;\xi)\equiv\xi\,P_{n+1}(z)+Q_{n+3}(z)=0,
\]
where $P_{n+1}$ and $Q_{n+3}$ correspond to the derivatives of $p(z)$ and $q(z)$, respectively. For each fixed $\xi\in\mathbb{R}$, this equation of degree $n+3$ possesses exactly $n+3$ roots (counting multiplicities). Generically, these consist of $r$ real stationary phase points $\{z_l^{\mathrm{R}}(\xi)\}_{l=1}^{r}$ and $\frac{n+3-r}{4}$ quadruples of symmetric complex points $\{z_s^{\mathrm{C}}(\xi),\bar{z}_s^{\mathrm{C}}(\xi),-z_s^{\mathrm{C}}(\xi),-\bar{z}_s^{\mathrm{C}}(\xi)\}_{s=1}^{(n+3-r)/4}$. Consequently, $h(z;\xi)$ factorizes as
\begin{equation}\label{eq:zero-j}
	h(z;\xi)=\prod_{l=1}^{r}\bigl(z-z_l^{\mathrm{R}}(\xi)\bigr)\prod_{s=1}^{\frac{n+3-r}{4}}\bigl(z-z_s^{\mathrm{C}}(\xi)\bigr)\bigl(z-\bar{z}_s^{\mathrm{C}}(\xi)\bigr)\bigl(z+z_s^{\mathrm{C}}(\xi)\bigr)\bigl(z+\bar{z}_s^{\mathrm{C}}(\xi)\bigr),\quad r\leq n+3.
\end{equation}
As the parameter $\xi$ varies, the roots move continuously in the complex plane. A pair of simple real roots may collide and subsequently leave the real axis as a complex conjugate pair; conversely, a complex conjugate pair may coalesce on the real axis and become real.

In order to deform the jump contours onto the steepest-descent paths and to separate the oscillatory and non-oscillatory components of the RH problem \ref{RH2-2}, we first perform suitable algebraic factorizations of the jump matrix $J(z)$. On the contour $\Gamma$, the jump matrix admits the following two equivalent factorizations:
\begin{align*}
	J(z)&= \begin{pmatrix}
		1 &	\overline{r(\bar{z})}\, e^{-2 i t \theta_-(z)}  \\[4pt]
		0 & 1
	\end{pmatrix}
	\begin{pmatrix}
		0 & i\, e^{-i\bigl(t(\theta_{+}(z)+\theta_{-}(z))\bigr)-i\phi_j} \\[4pt]
		i\, e^{i\bigl(t(\theta_{+}(z)+\theta_{-}(z))\bigr)+i\phi_j} & 0
	\end{pmatrix} 
	\begin{pmatrix}
		1 & 0 \\[4pt]
		r(z)\, e^{2 i t \theta_+(z)}	& 1
	\end{pmatrix}\\[6pt]
	&= \begin{pmatrix}
		1 & 0 \\[4pt]
		\dfrac{r(z)\, e^{2 i t \theta_-(z)}}{1+|r(z)|^2} & 1
	\end{pmatrix}
	\begin{pmatrix}
		0 & i\, e^{-i\bigl(t(\theta_{+}(z)+\theta_{-}(z))\bigr)-i\phi_j} \\[4pt]
		i\, e^{i\bigl(t(\theta_{+}(z)+\theta_{-}(z))\bigr)+i\phi_j} & 0
	\end{pmatrix} 
	\begin{pmatrix}
		1 & \dfrac{\overline{r(\bar{z})}\, e^{-2 i t \theta_+(z)}}{1+|r(z)|^2} \\[8pt]
		0 & 1
	\end{pmatrix}.
\end{align*}
On the $\mathbb{R}\cup \Sigma^{\mathrm{add}}$, the jump matrix can be factorized as
\begin{equation}\label{eq:factor-real-UL}
	J(z)=J_L(z)J_U(z)=
	\begin{pmatrix}
		1 & 0 \\[4pt]
		r(z)\, e^{2 i t \theta(z)}	& 1
	\end{pmatrix}
	\begin{pmatrix}
		1 &	\overline{r(\bar{z})}\, e^{-2 i t \theta(z)}  \\[4pt]
		0 & 1
	\end{pmatrix},
\end{equation}
or alternatively as
\begin{equation}\label{eq:factor-real-LDU}
	\begin{aligned}
		J(z) &=\widetilde J_U(z)\,J_D(z)\,\widetilde J_L(z) \\[4pt]
		&=
		\begin{pmatrix}
			1 & \dfrac{\overline{r(\bar{z})}\, e^{-2 i t \theta(z)}}{1+|r(z)|^2} \\[8pt]
			0 & 1
		\end{pmatrix}
		\begin{pmatrix}
			\dfrac{1}{1+|r(z)|^2} & 0 \\[8pt]
			0 & 1+|r(z)|^2
		\end{pmatrix}
		\begin{pmatrix}
			1 & 0 \\[4pt]
			\dfrac{r(z)\, e^{2 i t \theta(z)}}{1+|r(z)|^2} & 1
		\end{pmatrix}.
	\end{aligned}
\end{equation}
\section{Painlev\'e asymptotics in the transition region}
\label{secc3}
In this section, we carry out the Painlev\'e asymptotic analysis of RH problem~\ref{RH2-2} for $N(z)$ to derive the long-time asymptotic behavior of $q(x,t)$ in the transition region
\begin{equation}\label{eq:painleve-region-j}
	|\xi-\xi_{j_0}|\,t^{2/3}<C,\qquad \xi_{j_0}=-\frac{Q_{n+3}(E_{j_0})}{P_{n+1}(E_{j_0})},\quad j_0\in\mathcal{J},
\end{equation}
where $C>0$ is an arbitrary fixed constant. This region corresponds to the case in which, for some $s_0\in\mathcal{S}$ with $\mathcal{S}=\{1,2,\dots,\frac{n+1}{4}\}$, the four complex stationary phase points $z^{\mathrm{C}}_{s_0}$, $\bar{z}^{\mathrm{C}}_{s_0}$, $-z^{\mathrm{C}}_{s_0}$, $-\bar{z}^{\mathrm{C}}_{s_0}$ coalesce with the endpoints $E_{j_0}$, $\bar{E}_{j_0}$, $E_{n-j_0}$, $\bar{E}_{n-j_0}$ of the two branch cuts $\Gamma_{j_0}$ and $\Gamma_{n-j_0}$, as shown in Figures~\ref{tttttt} and~\ref{ttttttt}.

Specifically, we consider the special case in which the stationary phase points are distributed according to
\begin{equation}\label{eq:zero-jj}
	h(z;\xi)=\prod_{l=1}^{2}\bigl(z-z_l^{\mathrm{R}}(\xi)\bigr)\prod_{s=1}^{\frac{n+1}{4}}\bigl(z-z_s^{\mathrm{C}}(\xi)\bigr)\bigl(z-\bar{z}_s^{\mathrm{C}}(\xi)\bigr)\bigl(z+z_s^{\mathrm{C}}(\xi)\bigr)\bigl(z+\bar{z}_s^{\mathrm{C}}(\xi)\bigr),
\end{equation}
as depicted in panel~(C) of Figure~\ref{ttttttt}. Here the real stationary phase point $z_1^{\mathrm{R}}\in\mathbb{R}^-$ satisfies the signature condition $\operatorname{Im}\theta(z+i\epsilon)>0$ for $z\in(\kappa^{\mathrm{R}}_1,B_{j_0+1})$. By virtue of the symmetry relations \eqref{eq:symmetry-pq} for the phase function $\theta(z)$ defined in \eqref{eq:theta}, we deduce that $\operatorname{Im}\theta(z+i\epsilon)>0$ for $z\in(B_{n-j_0-1},\kappa^{\mathrm{R}}_2)$.
\begin{figure}[htp]
	\subfloat[pre-collision]{\label{agf}
		\begin{minipage}[t]{0.45\linewidth}
			\centering
			\includegraphics[width=2.5in]{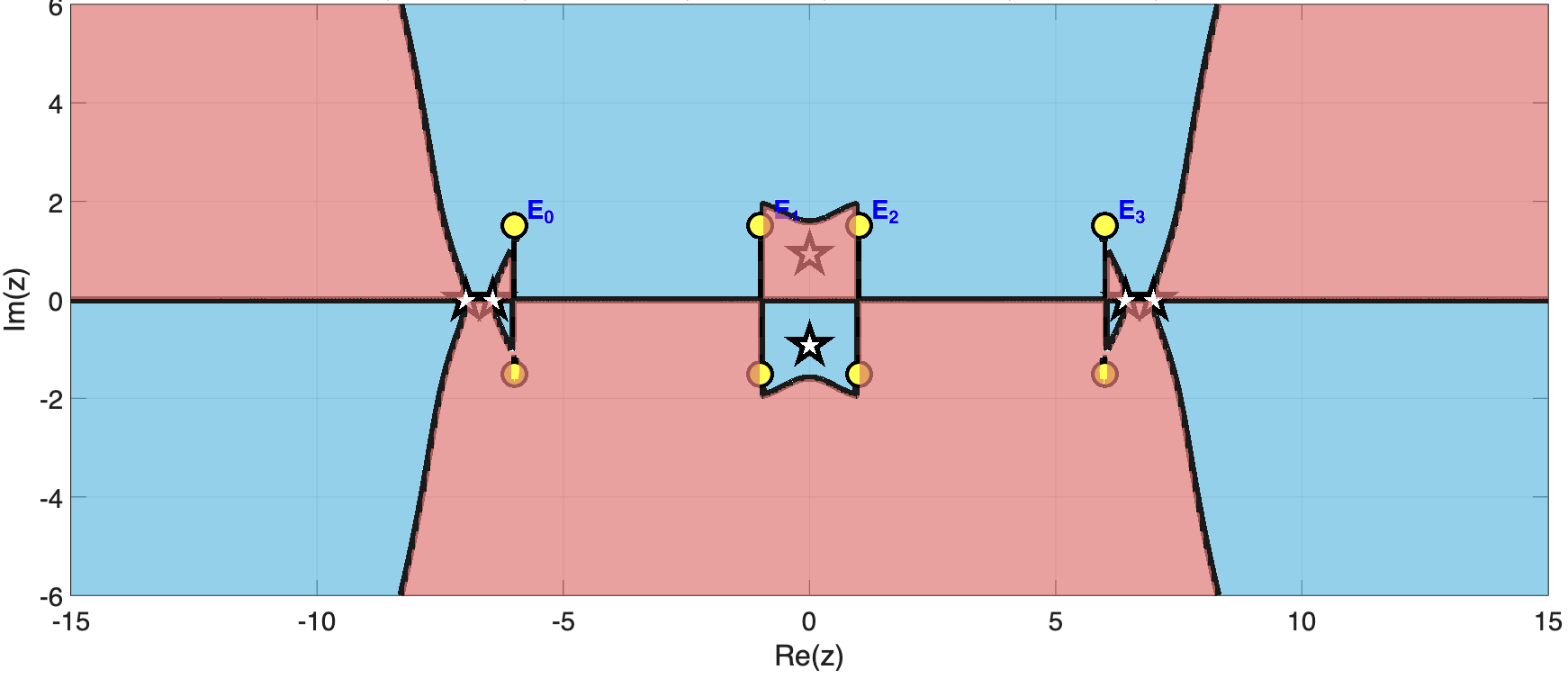}
		\end{minipage}
		}\quad
	\subfloat[post-collision ]{\label{c}\begin{minipage}[t]{0.45\linewidth}
			\centering
			\includegraphics[width=2.5in]{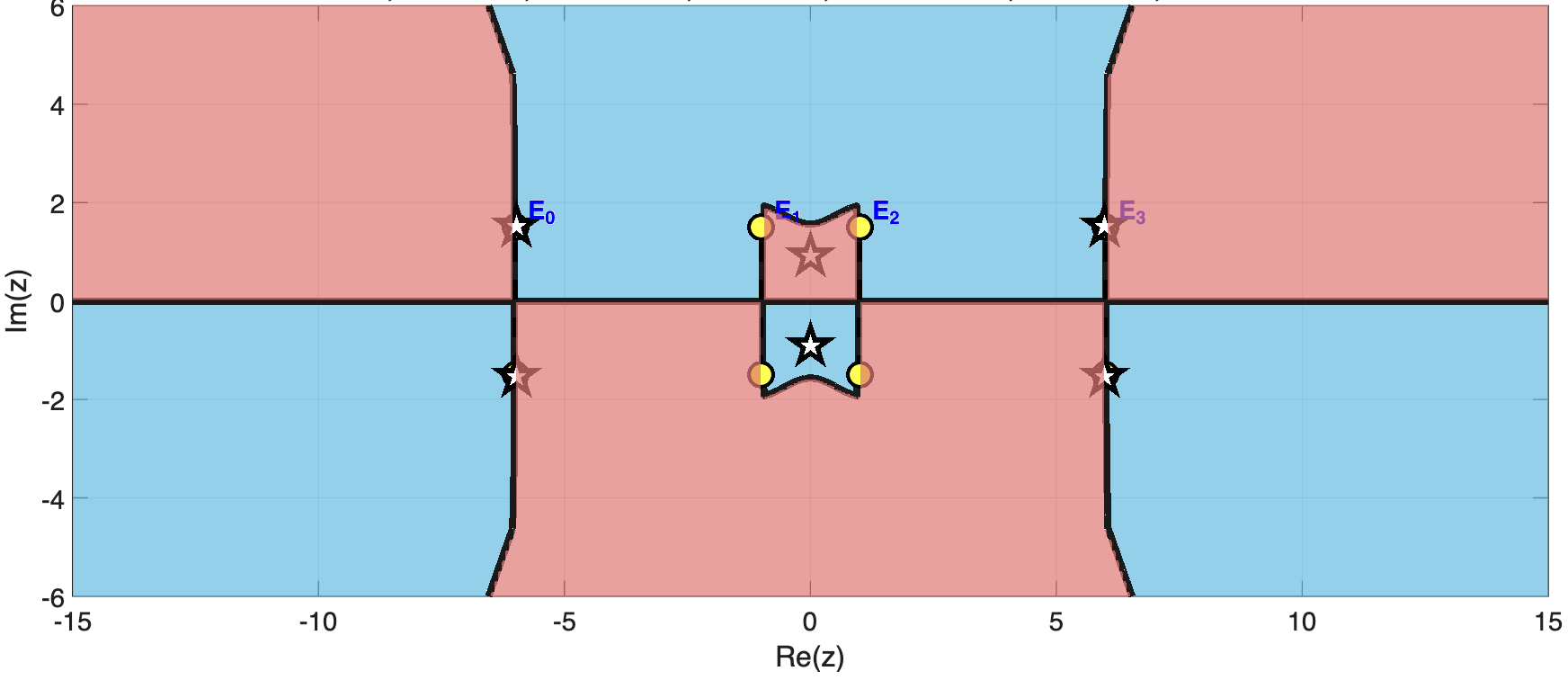}
	\end{minipage}}\\
	\subfloat[pre-collision ]{\label{agf}
	\begin{minipage}[t]{0.45\linewidth}
		\centering
		\includegraphics[width=2.5in]{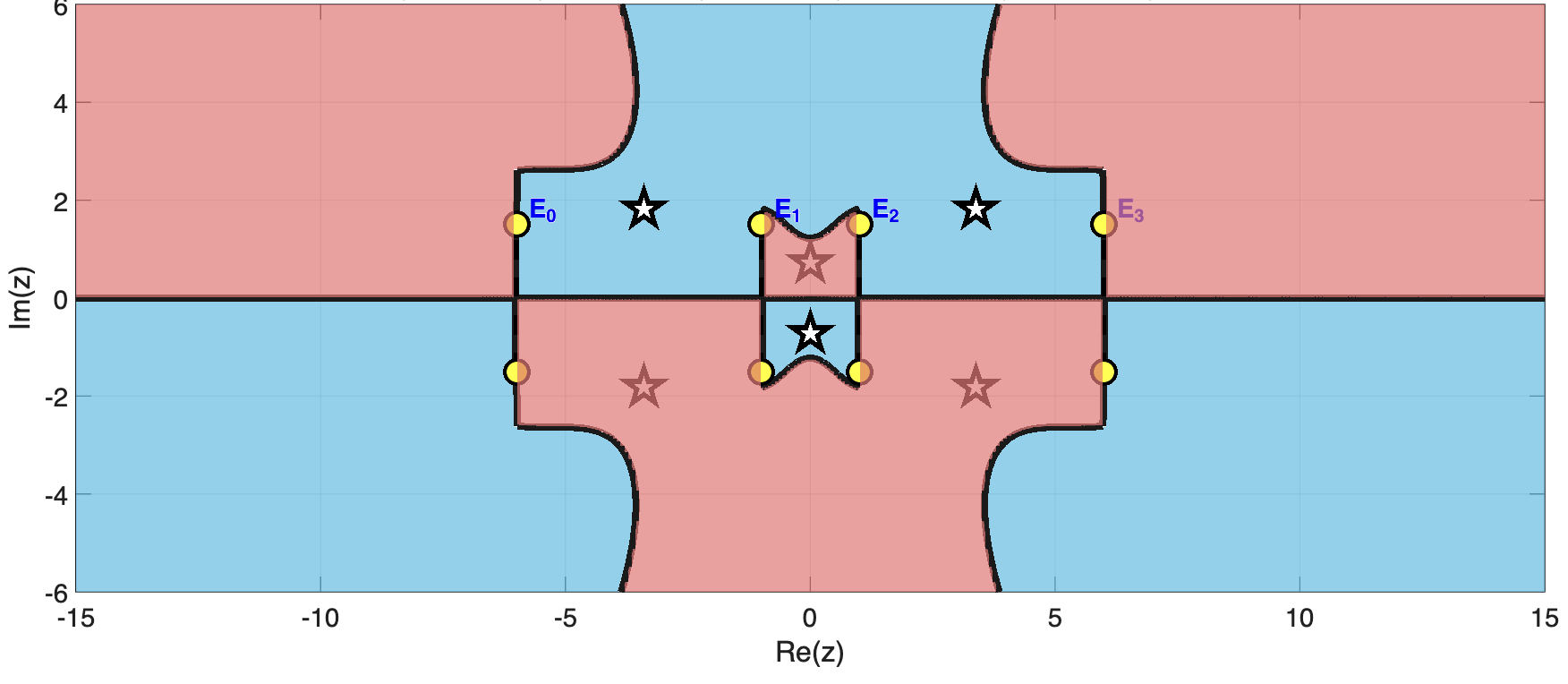}
	\end{minipage}
		}\quad
\subfloat[post-collision ]{\label{c}\begin{minipage}[t]{0.45\linewidth}
		\centering
		\includegraphics[width=2.5in]{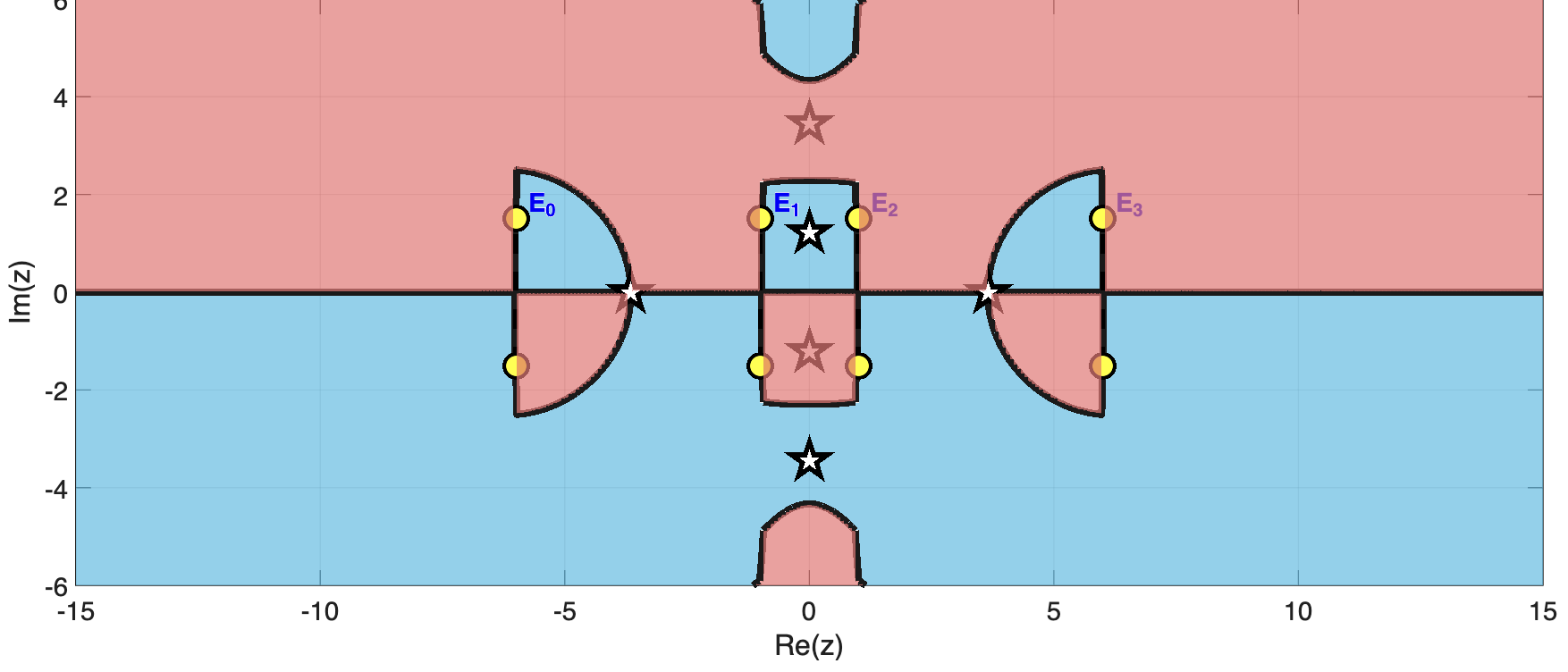}
\end{minipage}}
	\caption{One possible scenario of collision between complex stationary phase points and the endpoint (odd genus 3)}
\label{tttttt}
\end{figure}

	\begin{figure}[htp]
		\subfloat[pre-collision with $\Gamma_0$]{\label{agf}
			\begin{minipage}[t]{0.45\linewidth}
				\centering
				\includegraphics[width=2.5in]{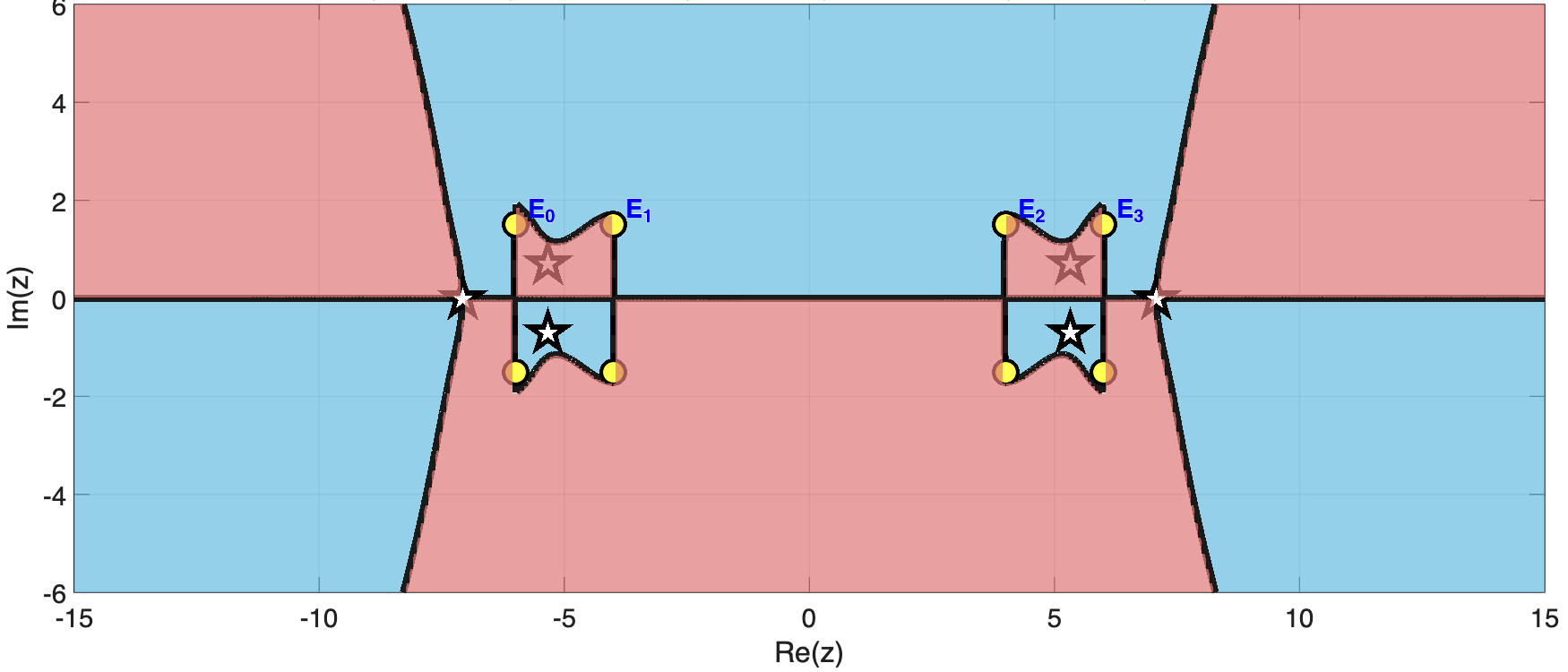}
			\end{minipage}
		}\quad
		\subfloat[coalescence with $\Gamma_0$]{\label{c}\begin{minipage}[t]{0.45\linewidth}
				\centering
				\includegraphics[width=2.5in]{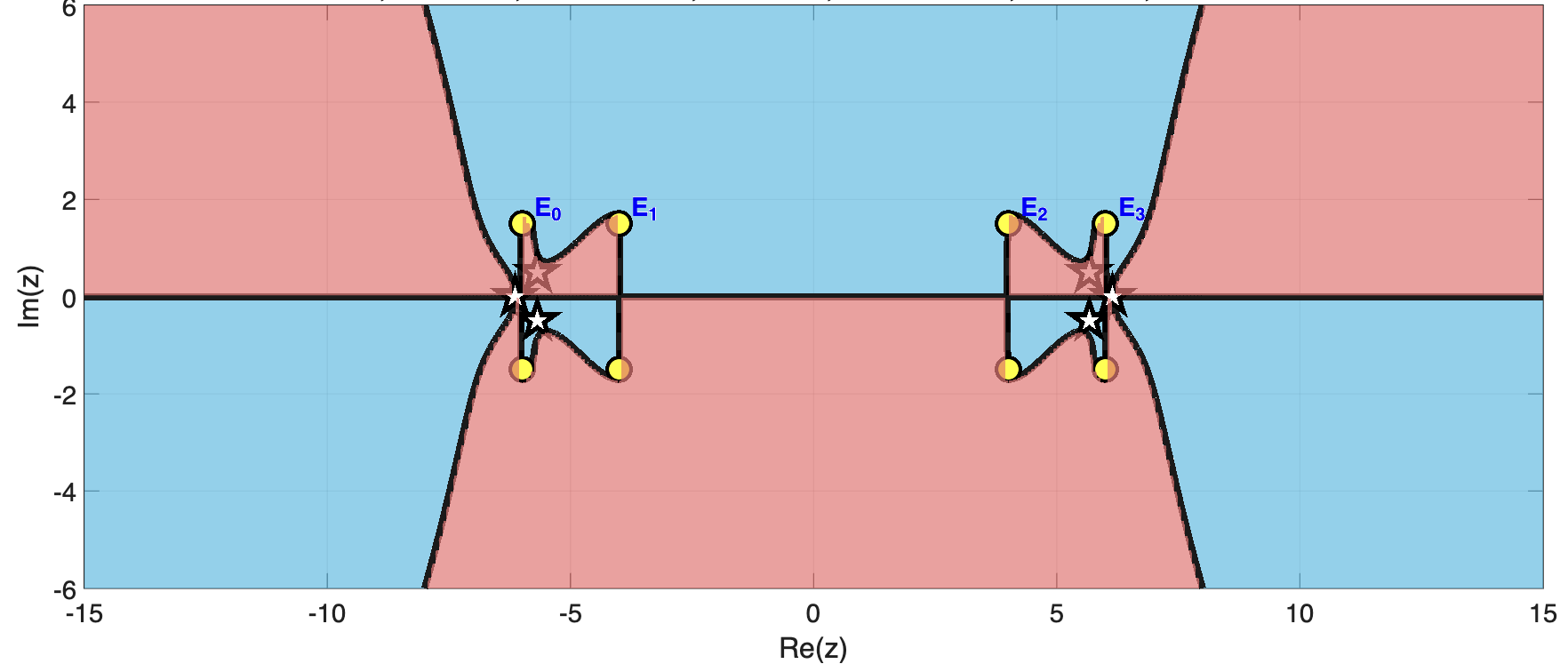}
		\end{minipage}}\\
		\subfloat[pre-collision with $\Gamma_0$]{\label{agf}
		\begin{minipage}[t]{0.45\linewidth}
			\centering
			\includegraphics[width=2.5in]{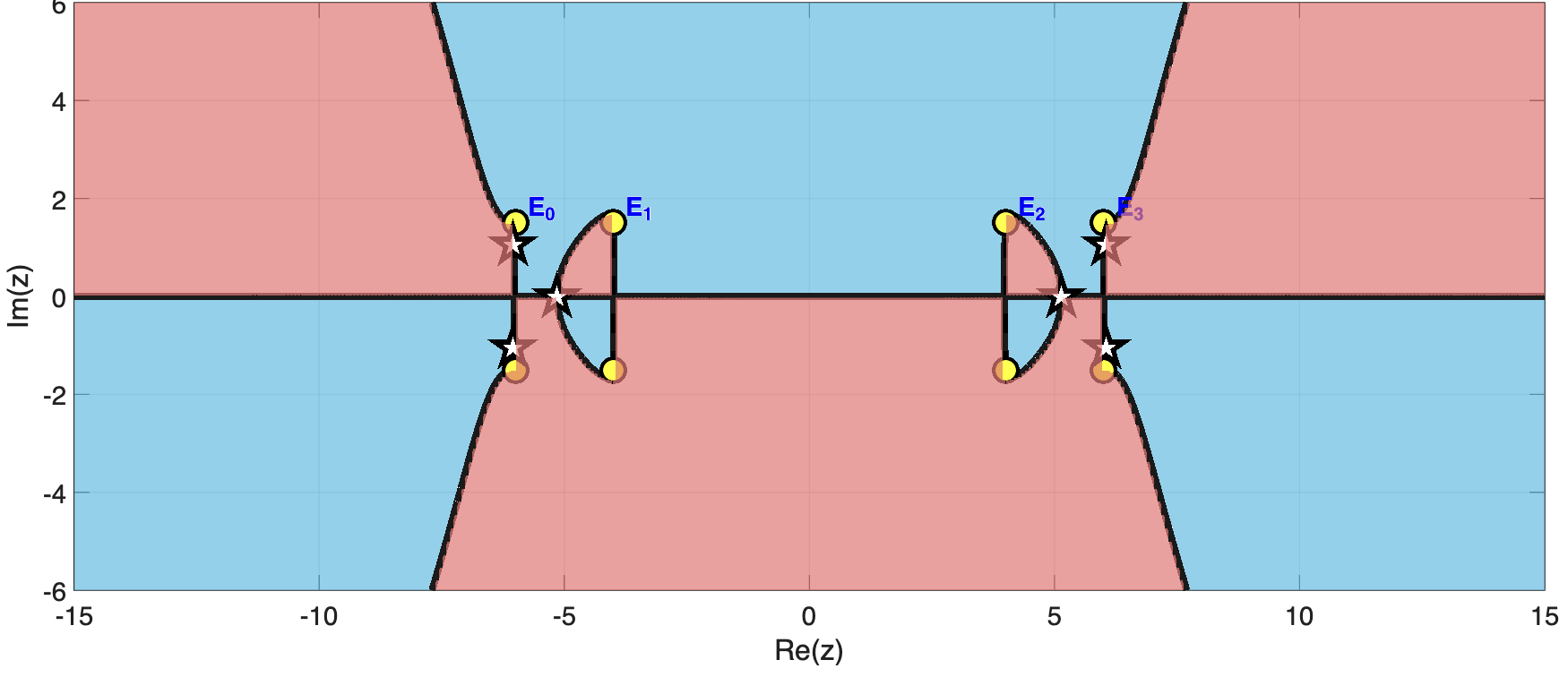}
		\end{minipage}
	}\quad
	\subfloat[coalescence with $\Gamma_0$]{\label{c}\begin{minipage}[t]{0.45\linewidth}
			\centering
			\includegraphics[width=2.5in]{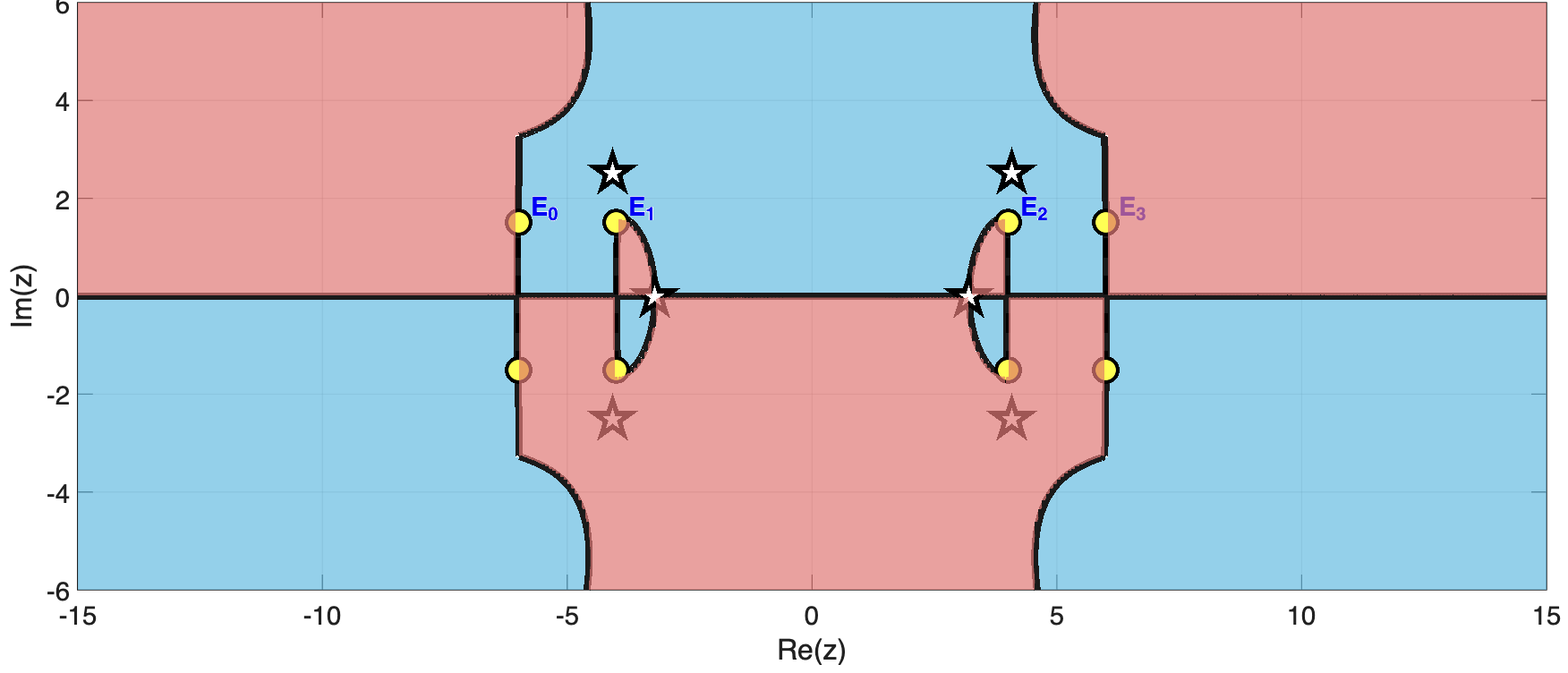}
	\end{minipage}}\\
	\subfloat[pre-collision with $\Gamma_0$]{\label{agf}
	\begin{minipage}[t]{0.45\linewidth}
		\centering
		\includegraphics[width=2.5in]{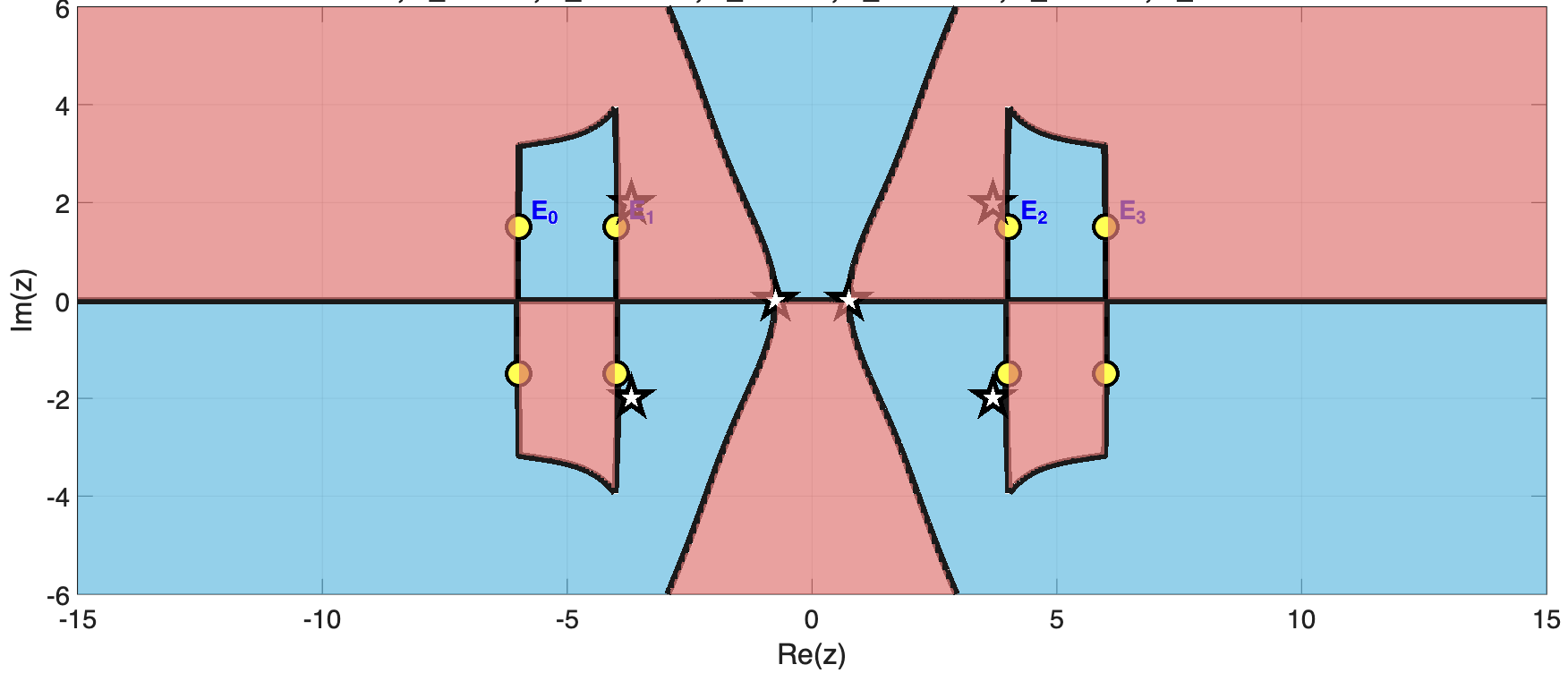}
	\end{minipage}
}\quad
\subfloat[coalescence with $\Gamma_0$]{\label{c}\begin{minipage}[t]{0.45\linewidth}
		\centering
		\includegraphics[width=2.5in]{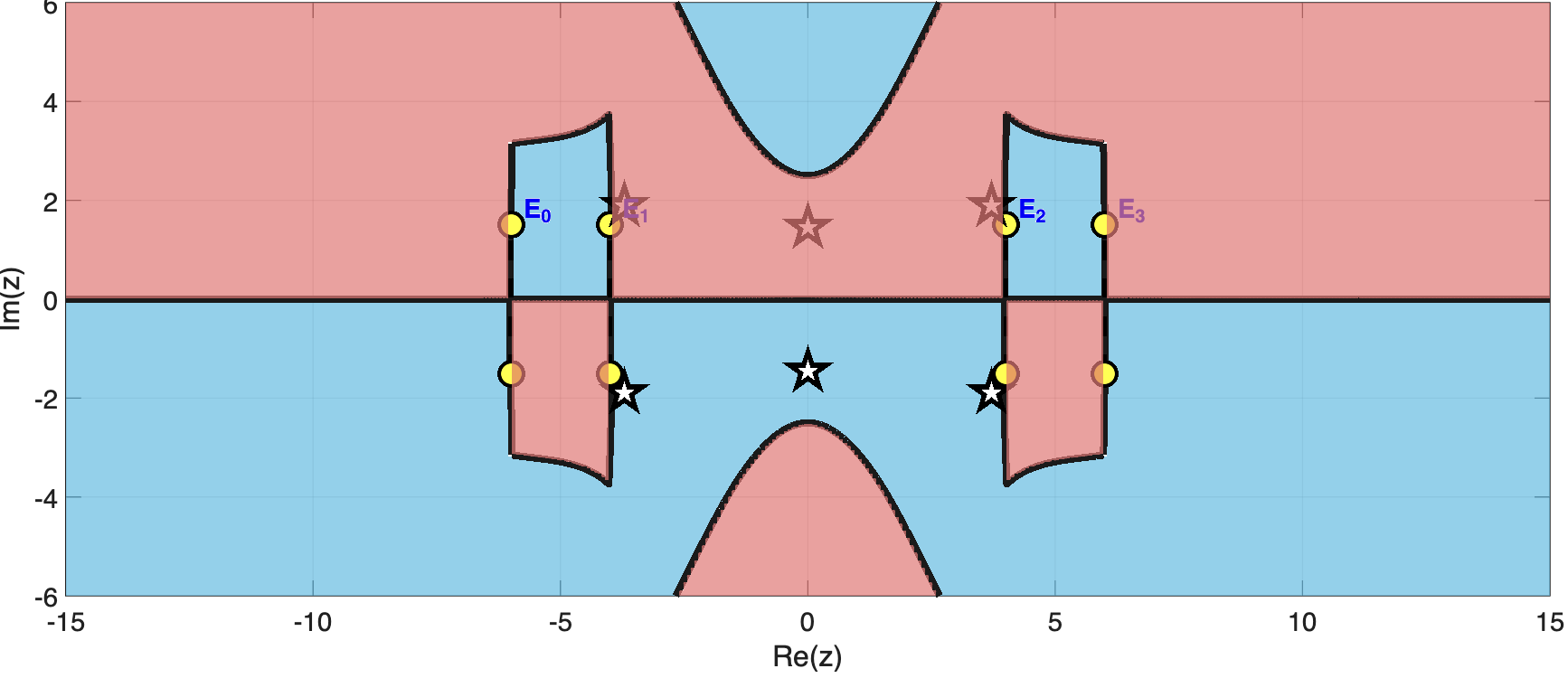}
\end{minipage}}
	\caption{Another possible scenario of collision between complex stationary phase points and the endpoint (odd genus 3)}
\label{ttttttt}
	\end{figure}	
\subsection{Conjugation}
The first step in our analysis is to introduce a transformation which 
renormalizes RH problem~\ref{RH2-2} so that it is 
well-conditioned as $t\to\infty$ with $\xi$ fixed. To arrive at a properly 
normalized problem, we define the scalar function
\begin{equation}\label{eq:delta}
	\delta(z)=\delta(z;\xi)=\exp\left\{\frac{w(z)}{2\pi i}\left(\int_{I}\frac{\log\bigl(1+|r(s)|^2\bigr)}{w(s)(s-z)}\,ds+\sum_{j=1}^n\delta_j\int_{\Gamma_j}\frac{i}{w_+(s)(s-z)}\,ds\right)\right\},
\end{equation}
where the constants $\delta_j$, $j=1,\dots,n$, are uniquely determined by 
the linear system
\begin{equation}\label{eq:delta_j}
	\int_{I}\frac{\log\bigl(1+|r(s)|^2\bigr)\,s^j}{w(s)}\,ds=\sum_{j=1}^n\delta_j\int_{\Gamma_j}\frac{i\,s^j}{w_+(s)}\,ds,
\end{equation}
Since $\bigl\{s^j\,ds/w(s)\bigr\}_{j=0}^{n-1}$ forms a basis of the holomorphic 
differentials on the Riemann surface $\mathcal{R}$ defined by 
$w^2=\prod_{j=0}^{n}(z-E_j)(z-\bar{E}_j)$, the coefficient matrix of the 
linear system~\eqref{eq:delta_j} is invertible; consequently, the constants 
$\delta_j$ exist uniquely and are real-valued.

\begin{proposition}\label{prop:delta}
	The function $\delta(z)$ defined by~\eqref{eq:delta} possesses the following 
	properties:
	\begin{enumerate}
		\item $\delta(z)$ is analytic in $\mathbb{C}\setminus(\Gamma\cup I)$ and 
		satisfies the symmetry relation 
		$\overline{\delta(\bar{z})}=1/\delta(z)$. Moreover, as $z\to\infty$,
		\begin{equation}\label{eq:delta-asymp}
			\delta(z)=\delta(\infty)+\frac{\delta^{(1)}}{z}+\mathcal{O}\bigl(z^{-2}\bigr),
		\end{equation}
		where
		\begin{equation}\label{eq:delta-infty}
			\begin{aligned}
			&	\delta(\infty)=\exp\left\{\frac{1}{2\pi i}\left[\int_{I}\frac{\log\bigl(1+|r(s)|^2\bigr)\,s^n}{w(s)}\,ds-\sum_{j=1}^n\delta_j\int_{\Gamma_j}\frac{i\,s^n}{w_+(s)}\,ds\right]\right\},\\[4pt]
			&	\delta^{(1)}=\delta(\infty)\sum_{j=0}^{n}\bigl(E_j+\bar{E}_j\bigr).
			\end{aligned}
		\end{equation}
		
		\item The boundary values $\delta_{\pm}(z)$ satisfy the jump conditions
		\begin{equation}\label{eq:delta-jump}
			\begin{cases}
				\delta_{+}(z)=\delta_{-}(z)\bigl(1+|r(z)|^2\bigr), & z\in I,\\[4pt]
				\delta_{+}(z)=\delta_{-}(z)\,e^{i\delta_j}, & z\in\Gamma_j,\quad j=1,\dots,n.
			\end{cases}
		\end{equation}
		
		\item For each endpoint $\beta\in\{E_j,\bar{E}_j\}_{j=0}^{n}$,
		\begin{equation}\label{eq:delta-endpoint}
			\delta(z)=\mathcal{O}\bigl((z-\beta)^{1/2}\bigr),\qquad z\to \beta.
		\end{equation}
	\end{enumerate}
\end{proposition}

\begin{proof}
	The representation \eqref{eq:delta} defines $\delta(z)$ as the exponential of a Cauchy-type integral with contours $\Gamma\cup I$; hence $\delta(z)$ is analytic in $\mathbb{C}\setminus(\Gamma\cup I)$ and admits continuous boundary values from either side. By the Schwarz symmetry $\overline{w(\bar{z})}=w(z)$, the reality of $\log(1+|r(s)|^2)$ on $I$, and the fact that the constants $\delta_j$ are real, a direct calculation yields $\overline{\delta(\bar{z})}=1/\delta(z)$.
	
	Expanding $(s-z)^{-1}$ in powers of $z^{-1}$ uniformly for $s$ on the compact contours $\Gamma\cup I$ and inserting this expansion into \eqref{eq:delta}, we obtain a series in $z^{-k}$ $(k=1,2,\dots)$. Owing to the linear system \eqref{eq:delta_j}, the coefficients of $z^{-1},z^{-2},\dots,z^{-n}$ vanish identically. The first non-vanishing contribution comes from the $z^{-(n+1)}$ term; multiplying by the asymptotic expansion $w(z)=z^{n+1}-\frac{1}{2}z^{n}\sum_{j=0}^{n}(E_j+\bar{E}_j)+\mathcal{O}(z^{n-1})$ and exponentiating, one arrives at \eqref{eq:delta-asymp} and \eqref{eq:delta-infty}.

	For $z\in I$, the second integral in \eqref{eq:delta} is analytic across $I$, whereas the first integral is a scalar Cauchy integral with density $\log(1+|r|^2)/w(s)$. By the Plemelj--Sokhotski formula and the fact that $w$ is continuous across $I$ (so that $w_+=w_-$ on $I$), we obtain
	\[
	\frac{\delta_+(z)}{\delta_-(z)}=\exp\left\{\frac{w(z)}{2\pi i}\cdot\frac{2\pi i\,\log(1+|r(z)|^2)}{w(z)}\right\}=1+|r(z)|^2,\qquad z\in I.
	\]
	For $z\in\Gamma_j$, the first integral is analytic across $\Gamma_j$, while the second integral has density $i/w_+(s)$. Since $w_+(s)=-w_-(s)$ on $\Gamma_j$, the Plemelj formula produces a jump of magnitude $i\delta_j$ in the exponent, whence
	\[
	\frac{\delta_+(z)}{\delta_-(z)}=e^{i\delta_j},\qquad z\in\Gamma_j,\quad j=1,\dots,n.
	\]
	
	Let $\beta\in\{E_j,\bar{E}_j\}$. In a neighborhood of $\beta$, the local coordinate on the Riemann surface gives $w(z)\sim c\,(z-\beta)^{1/2}$ for some non-zero constant $c$. The Cauchy integrals in \eqref{eq:delta} have at most logarithmic singularities at $\beta$ (their densities being integrable with $(s-\beta)^{-1/2}$ weights). Consequently, the exponent in \eqref{eq:delta} is $\mathcal{O}\bigl((z-\beta)^{1/2}\log(z-\beta)\bigr)$ as $z\to \beta$, and therefore tends to zero. Hence $\delta(z)$ remains bounded at $\beta$, and the square-root vanishing of the prefactor $w(z)$ yields the sharper estimate \eqref{eq:delta-endpoint}.
\end{proof}
We now perform contour deformation on RH problem~\ref{RH2-2}, 
following the standard procedure outlined in \cite{Borghese2018} 
in the presence of a discrete spectrum. For the real stationary phase point 
$z^{\mathrm{R}}_1$, the new contour is chosen to be
\begin{equation}\label{eq:contour-z1}
	\Sigma_{z^{\mathrm{R}}_1}=\Sigma_1^1\cup\Sigma_1^2\cup\Sigma_1^3\cup\Sigma_1^4,
\end{equation}
oriented with increasing real part. Denoting by $\Omega_1^k$ the four open 
sectors separated by $\mathbb{R}$ and the collection of rays 
$\Sigma_1^k$ ($k=1,2,3,4$), we then, by symmetry, define the contour 
$\Sigma_{z^{\mathrm{R}}_2}$ at the stationary phase point $z^{\mathrm{R}}_2$ 
through mirror reflection, as illustrated in Figure~\ref{fig:zs}.
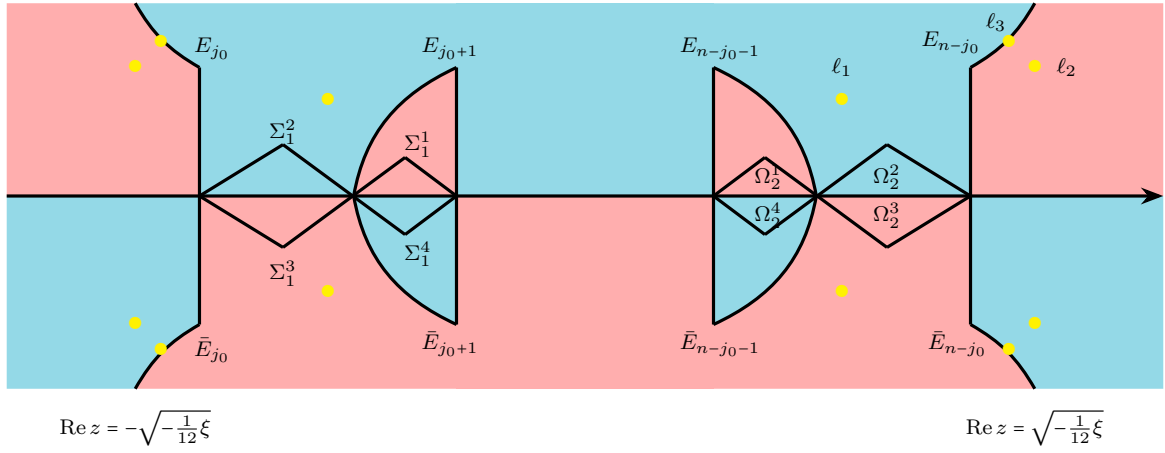
\begin{figure}[htp]
	\centering
	\begin{tikzpicture}[scale=0.85, >=Stealth, font=\footnotesize]
		% 定义颜色（若已在导言区定义，可删除以下两行）
		\definecolor{softred}{RGB}{255,120,120}
		\definecolor{softcyan}{RGB}{100,200,220}
		
		% =================== 填充区域（严格按照原颜色） ===================
		% 左上内瓣
		\fill[softred!60] (-2,0)--(-2,2) to[out=-155,in=80] (-3.6,0) --cycle;
		% 右上内瓣
		\fill[softred!60] (2,0)--(2,2) to[out=-25,in=100] (3.6,0) --cycle;
		% 左下内瓣
		\fill[softcyan!69] (-2,0)--(-2,-2) to[out=155,in=-80] (-3.6,0) --cycle;
		% 右下内瓣
		\fill[softcyan!69] (2,0)--(2,-2) to[out=25,in=-100] (3.6,0) --cycle;
		
		% 左上外瓣
		\fill[softred!60] (-9,0)--(-9,3)--(-7,3) to[out=-60,in=150] (-6,2)--(-6,0) --cycle;
		% 右上外瓣
		\fill[softred!60] (9,0)--(9,3)--(7,3) to[out=-120,in=30] (6,2)--(6,0) --cycle;
		% 左下外瓣
		\fill[softcyan!69] (-9,0)--(-9,-3)--(-7,-3) to[out=60,in=-150] (-6,-2)--(-6,0) --cycle;
		% 右下外瓣
		\fill[softcyan!69] (9,0)--(9,-3)--(7,-3) to[out=120,in=-30] (6,-2)--(6,0) --cycle;
		
		% 左上外过渡区
		\fill[softcyan!69] (-7,3)--(-2,3)--(-2,2) to[out=-155,in=80] (-3.6,0)--(-6,0)--(-6,2) to[out=150,in=-60] (-7,3);
		% 右上外过渡区
		\fill[softcyan!69] (7,3)--(2,3)--(2,2) to[out=-25,in=100] (3.6,0)--(6,0)--(6,2) to[out=30,in=-120] (7,3);
		% 左下外过渡区
		\fill[softred!60](-7,-3)--(-2,-3)--(-2,-2) to[out=155,in=-80] (-3.6,0)--(-6,0)--(-6,-2) to[out=-150,in=60] (-7,-3);
		% 右下外过渡区（原代码为 softred!60）
		\fill[softred!60] (7,-3)--(2,-3)--(2,-2) to[out=25,in=-100] (3.6,0)--(6,0)--(6,-2) to[out=-30,in=120] (7,-3);
		
		% 中间上条
		\fill[softcyan!69] (-2,3)--(2,3)--(2,0)--(-2,0);
		% 中间下条
		\fill[softred!60] (-2,-3)--(2,-3)--(2,0)--(-2,0);
		
		% =================== 边界线 ===================
		\draw[very thick] (-2,2) to[out=-155,in=80] (-3.6,0);
		\draw[very thick] (2,2) to[out=-25,in=100] (3.6,0);
		\draw[very thick] (-2,-2) to[out=155,in=-80] (-3.6,0);
		\draw[very thick] (2,-2) to[out=25,in=-100] (3.6,0);
		
		\draw[very thick] (-7,3) to[out=-60,in=150] (-6,2);
		\draw[very thick] (7,3) to[out=-120,in=30] (6,2);
		\draw[very thick] (-7,-3) to[out=60,in=-150] (-6,-2);
		\draw[very thick] (7,-3) to[out=120,in=-30] (6,-2);
		
		% 内部连接线
		\draw[very thick] (2.8,-0.6)--(4.7,0.8);
		\draw[very thick] (2.8,0.6)--(4.7,-0.8);
		\draw[very thick] (4.7,0.8)--(6,0);
		\draw[very thick] (4.7,-0.8)--(6,0);
		\draw[very thick] (-4.7,0.8)--(-6,0);
		\draw[very thick] (-4.7,-0.8)--(-6,0);
		\draw[very thick] (2.8,-0.6)--(2,0);
		\draw[very thick] (2.8,0.6)--(2,0);
		\draw[very thick] (-2.8,-0.6)--(-2,0);
		\draw[very thick] (-2.8,0.6)--(-2,0);
		\draw[very thick] (-2.8,-0.6)--(-4.7,0.8);
		\draw[very thick] (-2.8,0.6)--(-4.7,-0.8);
		
		% 竖直分界线
		\draw[very thick] (6,-2)--(6,2);
		\draw[very thick] (-6,-2)--(-6,2);
		\draw[very thick] (2,-2)--(2,2);
		\draw[very thick] (-2,-2)--(-2,2);
		
		% 水平实轴
		\draw[very thick, ->] (-9,0)--(9,0);
		
		% =================== 标记 ===================
		% 黄色圆点（孤立子/离散谱）
		\node[yellow, scale=1.5] at (7,2) {$\bullet$};
		\node[yellow, scale=1.5] at (7,-2) {$\bullet$};
		\node[yellow, scale=1.5] at (-7,2) {$\bullet$};
		\node[yellow, scale=1.5] at (-7,-2) {$\bullet$};
		\node[yellow, scale=1.5] at (4,1.5) {$\bullet$};
		\node[yellow, scale=1.5] at (4,-1.5) {$\bullet$};
		\node[yellow, scale=1.5] at (-4,1.5) {$\bullet$};
		\node[yellow, scale=1.5] at (-4,-1.5) {$\bullet$};
		\node[yellow, scale=1.5] at (6.6,2.4) {$\bullet$};
		\node[yellow, scale=1.5] at (6.6,-2.4) {$\bullet$};
		\node[yellow, scale=1.5] at (-6.6,2.4) {$\bullet$};
		\node[yellow, scale=1.5] at (-6.6,-2.4) {$\bullet$};
		
		% =================== 标签 ===================
		% \Sigma 与 \Omega 区域标签
		\node at (-2.6,0.8) {$\Sigma_1^1$};
		\node at (-2.6,-0.9) {$\Sigma_1^4$};
		\node at (-4.7,1) {$\Sigma_1^2$};
		\node at (-4.7,-1.2) {$\Sigma_1^3$};
		\node at (2.85,0.3) {$\Omega_2^1$};
		\node at (2.85,-0.3) {$\Omega_2^4$};
		\node at (4.7,0.3) {$\Omega_2^2$};
		\node at (4.7,-0.3) {$\Omega_2^3$};
		
		% 端点标签
		\node at (-2.1,2.3) {$E_{j_0+1}$};
		\node at (-2.1,-2.3) {$\bar E_{j_0+1}$};
		\node at (-5.8,2.3) {$E_{j_0}$};
		\node at (-5.8,-2.4) {$\bar E_{j_0}$};
		\node at (5.7,2.39) {$E_{n-j_0}$};
		\node at (5.8,-2.3) {$\bar E_{n-j_0}$};
		\node at (2.1,2.3) {$E_{n-j_0-1}$};
		\node at (2.1,-2.3) {$\bar E_{n-j_0-1}$};
		
		% \ell 标签
		\node at (6.4,2.7) {$\ell_3$};
		\node at (7.5,2) {$\ell_2$};
		\node at (4,2) {$\ell_1$};
		
		% 底部注释
		\node at (7,-3.6) {${\rm Re}\,z=\sqrt{-\frac{1}{12}\xi}$};
		\node at (-7,-3.6) {${\rm Re}\,z=-\sqrt{-\frac{1}{12}\xi}$};
	\end{tikzpicture}
	\caption{The jump contour $\Sigma^{(1)}$ and the decomposition regions for $N^{(1)}(z)$. 
}
	\label{fig:zs}
\end{figure}
	For the complex stationary phase point $z_{s_0}^{\mathrm{C}}$
	in the second quadrant that coalesces with the endpoint $E_{j_0}$, the symmetry relations \eqref{eq:symmetry-pq} imply that 
	the points $\bar{z}_{s_0}^{\mathrm{C}}$, $-\bar{z}_{s_0}^{\mathrm{C}}$, and 
	$-z_{s_0}^{\mathrm{C}}$ coalesce with $\bar{E}_{j_0}$, $E_{n-j_0}$, and 
	$\bar{E}_{n-j_0}$, respectively. We now choose the new local contours at 
	$E_{j_0}$ and $E_{n-j_0}$ as
	\begin{equation}\label{eq:contour-E}
		L_{j_0}=L_{j_0}^{+}\cup L_{j_0}^{-},\qquad 
		L_{n-j_0}=L_{n-j_0}^{+}\cup L_{n-j_0}^{-},
	\end{equation}
	and denote by $\Omega_{j_0}^{\pm}$ and $\Omega_{n-j_0}^{\pm}$ the corresponding 
	open sectors, respectively. By virtue of the symmetries 
	\eqref{eq:symmetry-pq}, the conjugate contours at $\bar{E}_{j_0}$ and 
	$\bar{E}_{n-j_0}$ are given by
	\begin{equation}\label{eq:contour-Ebar}
		\bar{L}_{j_0}=\bar{L}_{j_0}^{+}\cup\bar{L}_{j_0}^{-},\qquad 
		\bar{L}_{n-j_0}=\bar{L}_{n-j_0}^{+}\cup\bar{L}_{n-j_0}^{-},
	\end{equation}
	with associated sectors $\bar{\Omega}_{j_0}^{\pm}$ and 
	$\bar{\Omega}_{n-j_0}^{\pm}$, as shown in Figure~\ref{ttt6tttttttt}.
	
	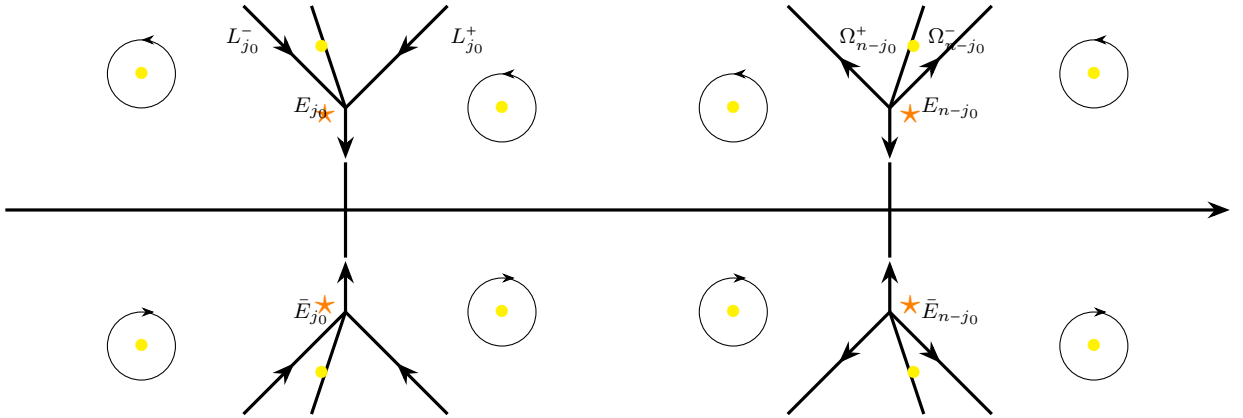
\begin{figure}[htp]
		\centering
		\begin{tikzpicture}[scale=0.9, >=Stealth, font=\footnotesize]
			% 水平实轴
			\draw[very thick, ->] (-9,0) -- (9,0);
			
			% 左端点 E_{j_0} 处 (-4,1.5)
			\draw[very thick] (-4,0) -- (-4,0.7);
			\draw[very thick, ->] (-4,1.5) -- (-4,0.75);
			\draw[very thick] (-4,1.5) -- (-4.5,3);
			\draw[very thick, postaction={decorate}, decoration={markings, mark=at position 0.5 with {\arrow{Stealth}}}] 
			(-5.5,3) -- (-4,1.5);
			\draw[very thick, postaction={decorate}, decoration={markings, mark=at position 0.5 with {\arrow{Stealth}}}] 
			(-2.5,3) -- (-4,1.5);
			
			% 左端点 \bar E_{j_0} 处 (-4,-1.5)
			\draw[very thick] (-4,0) -- (-4,-0.7);
			\draw[very thick, ->] (-4,-1.5) -- (-4,-0.75);
			\draw[very thick] (-4,-1.5) -- (-4.5,-3);
			\draw[very thick, postaction={decorate}, decoration={markings, mark=at position 0.5 with {\arrow{Stealth}}}] 
			(-5.5,-3) -- (-4,-1.5);
			\draw[very thick, postaction={decorate}, decoration={markings, mark=at position 0.5 with {\arrow{Stealth}}}] 
			(-2.5,-3) -- (-4,-1.5);
			
			% 右端点 E_{n-j_0} 处 (4,1.5)
			\draw[very thick] (4,0) -- (4,0.7);
			\draw[very thick, ->] (4,1.5) -- (4,0.75);
			\draw[very thick] (4,1.5) -- (4.5,3);
			\draw[very thick, postaction={decorate}, decoration={markings, mark=at position 0.5 with {\arrow{Stealth}}}] 
			(4,1.5) -- (5.5,3);
			\draw[very thick, postaction={decorate}, decoration={markings, mark=at position 0.5 with {\arrow{Stealth}}}] 
			(4,1.5) -- (2.5,3);
			
			% 右端点 \bar E_{n-j_0} 处 (4,-1.5)
			\draw[very thick] (4,0) -- (4,-0.7);
			\draw[very thick, ->] (4,-1.5) -- (4,-0.75);
			\draw[very thick] (4,-1.5) -- (4.5,-3);
			\draw[very thick, postaction={decorate}, decoration={markings, mark=at position 0.5 with {\arrow{Stealth}}}] 
			(4,-1.5) -- (2.5,-3);
			\draw[very thick, postaction={decorate}, decoration={markings, mark=at position 0.5 with {\arrow{Stealth}}}] 
			(4,-1.5) -- (5.5,-3);
			
			% 橙色星号（端点）
			\node[orange, scale=2.2] at (-4.3,1.4) {$\star$};
			\node[orange, scale=2.2] at (-4.3,-1.4) {$\star$};
			\node[orange, scale=2.2] at (4.3,1.4) {$\star$};
			\node[orange, scale=2.2] at (4.3,-1.4) {$\star$};
			
			% 黄色圆点（孤立子/离散谱）
			\node[yellow, scale=1.5] at (7,2) {$\bullet$};
			\node[yellow, scale=1.5] at (7,-2) {$\bullet$};
			\node[yellow, scale=1.5] at (-7,2) {$\bullet$};
			\node[yellow, scale=1.5] at (-7,-2) {$\bullet$};
			
			\node[yellow, scale=1.5] at (1.7,1.5) {$\bullet$};
			\node[yellow, scale=1.5] at (1.7,-1.5) {$\bullet$};
			\node[yellow, scale=1.5] at (-1.7,1.5) {$\bullet$};
			\node[yellow, scale=1.5] at (-1.7,-1.5) {$\bullet$};
			
			\node[yellow, scale=1.5] at (4.35,2.4) {$\bullet$};
			\node[yellow, scale=1.5] at (-4.35,2.4) {$\bullet$};
			\node[yellow, scale=1.5] at (-4.35,-2.4) {$\bullet$};
			\node[yellow, scale=1.5] at (4.35,-2.4) {$\bullet$};
			
			% 绕黄色圆点的小圆圈
			\draw[postaction={decorate}, decoration={markings, mark=at position 0.25 with {\arrow{Stealth}}}] 
			(1.7,1.5) circle (0.5cm);
			\draw[postaction={decorate}, decoration={markings, mark=at position 0.25 with {\arrow{Stealth}}}] 
			(-1.7,1.5) circle (0.5cm);
			\draw[postaction={decorate}, decoration={markings, mark=at position 0.25 with {\arrow{Stealth[reversed]}}}] 
			(1.7,-1.5) circle (0.5cm);
			\draw[postaction={decorate}, decoration={markings, mark=at position 0.25 with {\arrow{Stealth[reversed]}}}] 
			(-1.7,-1.5) circle (0.5cm);
			\draw[postaction={decorate}, decoration={markings, mark=at position 0.25 with {\arrow{Stealth}}}] 
			(7,2) circle (0.5cm);
			\draw[postaction={decorate}, decoration={markings, mark=at position 0.25 with {\arrow{Stealth}}}] 
			(-7,2) circle (0.5cm);
			\draw[postaction={decorate}, decoration={markings, mark=at position 0.25 with {\arrow{Stealth[reversed]}}}] 
			(-7,-2) circle (0.5cm);
			\draw[postaction={decorate}, decoration={markings, mark=at position 0.25 with {\arrow{Stealth[reversed]}}}] 
			(7,-2) circle (0.5cm);
			
			% 标签
			\node at (-4.5,1.5) {$E_{j_0}$};
			\node at (-4.5,-1.5) {$\bar E_{j_0}$};
			\node at (4.9,1.5) {$E_{n-j_0}$};
			\node at (4.9,-1.5) {$\bar E_{n-j_0}$};
			\node at (5,2.5) {$\Omega_{n-j_0}^-$};
			\node at (3.7,2.5) {$\Omega_{n-j_0}^+$};
			\node at (-2.2,2.5) {$L_{j_0}^+$};
			\node at (-5.5,2.5) {$L_{j_0}^-$};
		\end{tikzpicture}
	\caption{Contour deformation at the endpoints and soliton. ($\textcolor{yellow}{\bullet}$) stands for breather and ($\textcolor{orange}{\star}$) stands for complex stationary phase points.}
		\label{ttt6tttttttt}
	\end{figure}		
For each $j=1,2$, let $\ell_j, -\bar{\ell}_j\in\mathbb{C}^{+}$ 
be poles of the associated RH problem \ref{RH2-2}. Let $\Upsilon_j$ and 
$\bar{\Upsilon}_{-j}$ denote circles centered at $\ell_j$ and 
$-\bar{\ell}_j$, respectively, of sufficiently small radius so that the 
enclosed disks lie entirely in the open upper half-plane and are mutually 
disjoint. These circles are oriented counterclockwise. Likewise, for each $j=1,2$, let $\bar{\ell}_j, -\ell_j\in\mathbb{C}^{-}$ be the corresponding conjugate poles. Let 
$\bar{\Upsilon}_j$ and $\Upsilon_{-j}$ denote circles centered at 
$\bar{\ell}_j$ and $-\ell_j$, respectively, of sufficiently small radius 
so that the enclosed disks lie entirely in the open lower half-plane and 
are mutually disjoint. These circles are oriented clockwise.

Denoting by
\[
\Upsilon=\bigcup_{j=1}^{2}\bigl(\Upsilon_j\cup\bar{\Upsilon}_{-j}\cup
\bar{\Upsilon}_j\cup\Upsilon_{-j}\bigr),\quad \tilde \Upsilon=\Upsilon_3\cup\bar{\Upsilon}_{-3}\cup
\bar{\Upsilon}_3\cup\Upsilon_{-3}
\]
the union of all these circles, we replace the residue conditions 
\eqref{eq:res-condition} of RH problem~\ref{RH2-2} with 
equivalent jump conditions across the closed contours in $\Upsilon$. This 
procedure preserves the Schwarz symmetry of the problem. Let $\chi_{\mathcal{Z}}\in C_0^\infty(\mathbb{C},[0,1])$ be a smooth cutoff 
function supported near the discrete spectrum such that
\begin{equation}\label{eq:chi-Z}
	\chi_{\mathcal{Z}}(z)=
	\begin{cases}
		1, & \operatorname{dist}\bigl(z,\bigcup_{j=1}^{3}\mathcal{Z}_j\cup\Sigma^0\bigr)<\dfrac{2d}{3},\\
		0, & \operatorname{dist}\bigl(z,\bigcup_{j=1}^{3}\mathcal{Z}_j\cup\Sigma^0\bigr)>\dfrac{2d}{3}.
	\end{cases}
\end{equation}
and
\begin{equation}\label{eq:d-min}
	d=\frac{1}{2}\min_{\substack{\lambda,\mu\in\mathcal{P}\\\lambda\neq\mu}}|\lambda-\mu|.
\end{equation}On the support of $1-\chi_{\mathcal{Z}}$, the outer parametrix $N^{(out)}(z)$ 
is pole-free. 

We then define the 
deformed contour
\begin{equation}\label{eq:Sigma1}
	\Sigma^{(1)}=\Upsilon\cup\Sigma_{z^{\mathrm{R}}_1}\cup\Sigma_{z^{\mathrm{R}}_2}
	\cup L_{j_0}\cup L_{n-j_0}\cup\bar{L}_{j_0}\cup\bar{L}_{n-j_0}.
\end{equation}
The first transformation of RH problem~\ref{RH2-2} for $N(z)$ 
is defined by
\begin{equation}\label{eq:N1}
	N^{(1)}(z):=N^{(1)}(z;\xi,t)=\delta(\infty)^{-\sigma_3}N(z)G(z)\delta(z)^{\sigma_3},
\end{equation}
where the piecewise analytic matrix $G(z)$ is given by
\begin{equation}\label{eq:G}
	G(z)=
	\begin{cases}
	J_{\Upsilon_j}(z)=\begin{pmatrix}
			1 & e^{-2it\theta(z)}\dfrac{c_j}{z-\ell_j}\\[10pt]
			0 & 1
		\end{pmatrix},  z\in\operatorname{int}(\Upsilon_j),\quad \,\,\,	\widetilde{J}_L^{-1}(z),  z\in(1-\chi_{\mathcal{Z}})(\Omega_1^1\cup\Omega_2^1),\\
	J_{\bar\Upsilon_{-j}}(z)=	\begin{pmatrix}
			1 & -e^{2it\theta(z)}\dfrac{\bar{c}_j}{z+\bar{\ell}_j}\\[10pt]
			0 & 1
		\end{pmatrix},  z\in\operatorname{int}(\bar{\Upsilon}_{-j}),\,\,\,\,	J_U^{-1}(z),  z\in(1-\chi_{\mathcal{Z}})(\Omega_1^2\cup\Omega_2^2),\\
	J_{\bar\Upsilon_j}(z)=	\begin{pmatrix}
			1 & 0\\[6pt]
			-e^{2it\theta(z)}\dfrac{\bar{c}_j}{z-\bar{\ell}_j} & 1
		\end{pmatrix},  z\in\operatorname{int}(\bar{\Upsilon}_j),\quad	\,\,\,J_L(z),  z\in(1-\chi_{\mathcal{Z}})(\Omega_1^3\cup\Omega_2^3),\\
J_{\Upsilon_{-j}}(z)=		\begin{pmatrix}
			1 & 0\\[6pt]
			e^{2it\theta(z)}\dfrac{c_j}{z+\ell_j} & 1
		\end{pmatrix},  z\in\operatorname{int}(\Upsilon_{-j}),\quad \,	\widetilde{J}_U(z),  z\in(1-\chi_{\mathcal{Z}})(\Omega_1^4\cup\Omega_2^4)\,\,,j=1,2,\\
			J_U^{-1}(z),  z\in\Omega_{j_0}^+\cup\Omega_{n-j_0}^+\cup\bar{\Omega}_{j_0}^+\cup\bar{\Omega}_{n-j_0}^+,\qquad\quad\,\,\,\,\,
		J_L(z),  z\in\Omega_{j_0}^-\cup\Omega_{n-j_0}^-\cup\bar{\Omega}_{j_0}^-\cup\bar{\Omega}_{n-j_0}^-,\\
		I,  \text{otherwise}.
	\end{cases}
\end{equation}
where, for each closed contour $\Upsilon_\#\in\{\Upsilon_j,\bar{\Upsilon}_{-j},\bar{\Upsilon}_j,\Upsilon_{-j}\}_{j=1,2}$, $\operatorname{int}(\Upsilon_\#)$ denotes its interior. Then $N^{(1)}(z)$ satisfies the following RH problem.
	\begin{problem}\label{RH-1}
		Find a $2\times2$ matrix-valued function $N^{(1)}(z)$,  analytic in 
		$\mathbb{C}\setminus\bigl(\Sigma^{(1)}\cup\mathcal{Z}\bigr)$, such that:
		\begin{enumerate}
			\item $
				N^{(1)}(z)=I+\mathcal{O}(z^{-1}),\qquad |z|\to\infty.
	$
			
			\item For each $z\in\Sigma^{(1)}$, the boundary values 
			$N^{(1)}_\pm(z)$ satisfy
			\begin{equation}\label{eq:jump-N1}
				N^{(1)}_+(z)=N^{(1)}_-(z)J^{(1)}(z),
			\end{equation}
			where the jump matrix $J^{(1)}(z)$ is defined according to the component of 
			$\Sigma^{(1)}$ as follows.
			\begin{itemize}
				\item For $z\in\Gamma$,
				\begin{equation*}
					J^{(1)}(z)=\delta(z)^{-\sigma_3}J^{(\mathrm{alg})}(z)\delta(z)^{\sigma_3}.
				\end{equation*}
				
				\item For $z\in\Sigma_{z^{\mathrm{R}}_1}\cup\Sigma_{z^{\mathrm{R}}_2}\cup L_{j_0}\cup L_{n-j_0}\cup\bar{L}_{j_0}\cup\bar{L}_{n-j_0}$,
				\begin{equation}\label{eq:PC-jump}
					J^{(1)}(z)=
					\begin{cases}
						\delta(z)^{-\sigma_3}\widetilde{J}_L(z), & z\in\Sigma_1^1\cup\Sigma_2^1,\\[4pt]
						\delta(z)^{-\sigma_3}\widetilde{J}_U^{-1}(z), & z\in\Sigma_1^4\cup\Sigma_2^4,\\[4pt]
						\delta(z)^{-\sigma_3}J_U^{-1}(z), & z\in L_{n-j_0}^+\cup\bar{L}_{n-j_0}^+,\\[4pt]
						\delta(z)^{-\sigma_3}J_L(z), & z\in L_{j_0}^-\cup\bar{L}_{j_0}^-,\\[4pt]
							\delta(z)^{-\sigma_3}J_U(z), & z\in\Sigma_1^2\cup\Sigma_2^2\cup L_{j_0}^+\cup\bar{L}_{j_0}^+,\\[4pt]
						\delta(z)^{-\sigma_3}J_L^{-1}(z), & z\in\Sigma_1^3\cup\Sigma_2^3\cup L_{n-j_0}^-\cup\bar{L}_{n-j_0}^-.
					\end{cases}
				\end{equation}
				
				\item For $z\in\Upsilon$,
				\begin{equation}\label{eq:jump-soliton}
					J^{(1)}(z)=
					\begin{cases}
						\delta(z)^{-\sigma_3}	J_{\Upsilon_j}(z), \quad\,\, z\in\Upsilon_j,\quad\,\,\,
						\delta(z)^{-\sigma_3}	J_{\bar{\Upsilon}_j}(z), \quad\,\,\, z\in\bar{\Upsilon}_j,\\[4pt]
						\delta(z)^{-\sigma_3}	J_{\Upsilon_{-j}}(z), \quad z\in\Upsilon_{-j},\quad
						\delta(z)^{-\sigma_3}	J_{\bar{\Upsilon}_{-j}}(z), \quad z\in\bar{\Upsilon}_{-j}.
					\end{cases}
				\end{equation}
			\end{itemize}
			
			\item At each point $z\in\mathcal{Z}_3$, the residue 
			condition is given by the corresponding entry in \eqref{eq:res-condition}.
			
			\item For each $\beta\in\{E_j,\bar{E}_j\}_{j=0}^n$,
			\begin{equation*}
				N^{(1)}(z)=\mathcal{O}\bigl((z-\beta)^{-1/4}\bigr),\qquad z\to\beta.
			\end{equation*}
		\end{enumerate}
	\end{problem}
Specifically, let $U_{\alpha}$ denote the open neighborhood
\begin{equation}\label{eq:U-alpha}
	U_\alpha=\left\{z\in\mathbb{C}:|z-\alpha|<\frac{d}{2}\right\},\qquad 
	\alpha\in\mathcal{P}=\mathcal{P}_1\cup\mathcal{P}_2,
\end{equation}
where
\begin{equation}\label{eq:P1P2}
	\mathcal{P}_1=\{\kappa^{\mathrm{R}}_1,\kappa^{\mathrm{R}}_2\},\qquad
	\mathcal{P}_2=\{E_{j_0},\bar{E}_{j_0},E_{n-j_0},\bar{E}_{n-j_0}\},
\end{equation}
and the discrete sets $\mathcal{Z}_1,\mathcal{Z}_2$ are defined in \eqref{eq:zero-decomposition}--\eqref{eq:Zj-def}. % 按需调整引用
In what follows, we construct the solution $N^{(1)}(z)$ by seeking a function of the form
\begin{equation}\label{eq:N-1-decompose}
	N^{(1)}(z)=
	\begin{cases}
		\mathcal{E}(z)N^{(\mathrm{out})}(z), & z\in\mathbb{C}\setminus U,\\[6pt]
		\mathcal{E}(z)N^{(\mathrm{out})}(z)N^{(\mathrm{loc})}(z), & z\in U,
	\end{cases}
\end{equation}
where
\begin{equation}\label{eq:U-union}
	U=\bigcup_{\alpha\in\mathcal{P}}U_\alpha,
\end{equation}
and $N^{(\mathrm{out})}(z)$ and $N^{(\mathrm{loc})}(z)$ are the outer and local 
parametrices to be constructed in RH Problems~\ref{RH-out} and~\ref{rh:locc}, respectively. 
The error matrix $\mathcal{E}(z)$ solves a small-norm RH problem \ref{prob:small-norm}, and its existence together with 
uniform asymptotic bounds will be established in the sequel.
\subsection{The outer model: soliton and algebro-geometric background}
We solve RH problem~\ref{RH-1} for $N^{(1)}(z)$ by seeking a solution of the form \eqref{eq:N-1-decompose}, where the outer parametrix is further factorized as
\begin{equation}\label{eq:outer-decompose}
	N^{(\mathrm{out})}(z)=N^{(\mathrm{alg})}(z)N^{(\mathrm{sol})}(z),
\end{equation}
with the superscripts $(\mathrm{alg})$ and $(\mathrm{sol})$ indicating the 
algebro-geometric and soliton components, respectively.

Away from the critical point set $\mathcal{P}_1\cup\mathcal{P}_2$ defined in \eqref{eq:P1P2}, the jump matrix $J^{(1)}(z)$ is 
uniformly close to the identity. The matrix $N^{(1)}(z)$ is sectionally analytic 
in $\mathbb{C}\setminus\Sigma^{(1)}$, and its boundary values satisfy the jump 
relation \eqref{eq:PC-jump} on the contour 
$\Sigma^{(1)}$. Using the definition 
\eqref{eq:U-alpha} of the neighborhoods $U_\alpha$ together with the phase 
function \eqref{eq:theta}, we obtain the uniform estimate
\begin{equation}\label{eq:bound}
	\bigl\|J^{(1)}(z)-I\bigr\|_{L^\infty(\Sigma^{(1)}\setminus U)}
	=\mathcal{O}(e^{-ct}),\qquad c>0,
\end{equation}
which is exponentially small as $t\to\infty$ on $\Sigma^{(1)}\setminus U$. This exponential decay justifies the construction of an outer 
model solution valid outside $U$, which we now proceed to describe.
\begin{problem}\label{RH-out}
	Find a $2\times2$ matrix-valued function $N^{(\mathrm{out})}(z)$, analytic in 
	$\mathbb{C}\setminus\bigl(\Gamma\cup\Upsilon\bigr)$, such that:
	\begin{enumerate}
		\item $N^{(\mathrm{out})}(z)=I+\mathcal{O}(z^{-1}),\qquad |z|\to\infty$.
		
		\item For each $z\in\Gamma\cup\Upsilon$, the boundary values 
		$N^{(\mathrm{out})}_\pm(z)$ satisfy the jump relation
		\begin{equation*}
			N^{(\mathrm{out})}_+(z)=N^{(\mathrm{out})}_-(z)J^{(\mathrm{out})}(z),
		\end{equation*}
		where
		\begin{equation*}
			J^{(\mathrm{out})}(z)=J^{(1)}(z),\qquad z\in\Gamma\cup\Upsilon.
		\end{equation*}
					\item At each point $z\in\mathcal{Z}_3$, the residue 
		condition is given by the corresponding entry in \eqref{eq:res-condition}.
		\item For each endpoint $\beta\in\{E_j,\bar{E}_j\}_{j=0}^n$,
		\begin{equation*}
			N^{(\mathrm{out})}(z)=\mathcal{O}\bigl((z-\beta)^{-1/4}\bigr),\qquad z\to \beta.
		\end{equation*}
	\end{enumerate}
\end{problem}
By Assumption~\ref{ass:discrete-spectrum}, for the eigenvalue $\ell_1$ lying in 
the region $\operatorname{Im}\theta(\xi,\ell_1)<0$, the jump matrices 
$J^{\mathrm{sol}}(z)$ given by \eqref{eq:jump-soliton} are exponentially close 
to the identity as $t\to\infty$, and therefore do not contribute to the 
leading-order asymptotics of the RH problem~\ref{RH-out}. For 
the eigenvalue $\ell_2$ in the region $\operatorname{Im}\theta(\xi,\ell_2)>0$, the 
jumps $J^{\mathrm{sol}}(z)$ defined by \eqref{eq:jump-soliton} grow 
exponentially as $t\to\infty$. Nevertheless, as shown in 
\cite{DeiftKamvissisKriecherbauerZhou1996}, the jumps along $\Upsilon_2$ still 
do not contribute to the leading-order asymptotics. This can be seen by applying the following additional transformation to the 
problem:
\begin{equation}\label{tran-soliton}
	N_1^{(\mathrm{out})}(z)=
	\begin{cases}
		N^{(\mathrm{out})}(z)\,Y(z)^{\sigma_3}, & z\in\mathbb{C}\setminus\operatorname{int}\bigl(\bar{\Upsilon}_2\cup\Upsilon_2\cup\Upsilon_{-2}\cup\bar{\Upsilon}_{-2}\bigr),\\
		N^{(\mathrm{out})}(z)\,\widetilde{Y}(z)\,Y(z)^{\sigma_3}, & z\in\operatorname{int}\bigl(\bar{\Upsilon}_2\cup\Upsilon_2\cup\Upsilon_{-2}\cup\bar{\Upsilon}_{-2}\bigr),
	\end{cases}
\end{equation}
where $Y(z)$ is the piecewise-defined function
\begin{equation}\label{eq:Y}
	Y(z)=
	\begin{cases}
		\dfrac{z-\bar{\ell}_2}{z-\ell_2},  z\in\mathbb{C}\setminus\operatorname{int}\bigl(\bar{\Upsilon}_2\cup\Upsilon_2\cup\Upsilon_{-2}\cup\bar{\Upsilon}_{-2}\bigr),\quad
		z-\bar{\ell}_2,  z\in\operatorname{int}(\bar{\Upsilon}_2),\\
		\dfrac{1}{z-\ell_2},  z\in\operatorname{int}(\Upsilon_2),\quad
		z+\ell_2,  z\in\operatorname{int}(\Upsilon_{-2}),\quad
		\dfrac{1}{z+\bar{\ell}_2},  z\in\operatorname{int}(\bar{\Upsilon}_{-2}),
	\end{cases}
\end{equation}
and
\begin{equation}\label{eq:Y-tilde}
	\widetilde{Y}(z)=
	\begin{cases}
		\begin{pmatrix}
			\frac{1-\dfrac{n^2(\ell_2)}{n^2(z)}}{z-\ell_2} & -c_2\,\delta^{-2}(\xi,z)\,e^{-2it\theta(\xi,\ell_2)}\\
			\frac{n^2(\ell_2)\,\delta^2(\xi,z)\,e^{2it\theta(\xi,\ell_2)}}{c_2\,(z-\bar{\ell}_2)^2} & z-\ell_2
		\end{pmatrix}, & z\in\operatorname{int}(\Upsilon_2),\\
		\begin{pmatrix}
			z-\bar{\ell}_2 & \frac{e^{-2i\theta(\xi,\bar{\ell}_2)t}\,\delta^{-2}(z)}{n^2(\bar{\ell}_2)\,\bar{c}_2\,(z-\ell_2)^2}\\
			-\bar{c}_2\,\delta^2(\xi,z)\,e^{2i\theta(\xi,\bar{\ell}_2)t} & \frac{1-\dfrac{n^2(z)}{n^2(\bar{\ell}_2)}}{z-\bar{\ell}_2}
		\end{pmatrix}, & z\in\operatorname{int}(\bar{\Upsilon}_2),\\
		\begin{pmatrix}
			\frac{1-\dfrac{n^2(-\bar{\ell}_2)}{n^2(z)}}{z+\bar{\ell}_2} & \bar{c}_2\,\delta^{-2}(\xi,z)\,e^{-2it\theta(\xi,-\bar{\ell}_2)}\\[16pt]
			-\frac{n^2(-\bar{\ell}_2)\,\delta^2(\xi,z)\,e^{2it\theta(\xi,-\bar{\ell}_2)}}{\bar{c}_2\,(z+\ell_2)^2} & z+\bar{\ell}_2
		\end{pmatrix}, & z\in\operatorname{int}(\bar{\Upsilon}_{-2}),\\
		\begin{pmatrix}
			z+\ell_2 & -\frac{e^{-2i\theta(\xi,-\ell_2)t}\,\delta^{-2}(z)}{n^2(-\ell_2)\,c_2\,(z+\bar{\ell}_2)^2}\\[16pt]
			c_2\,\delta^2(\xi,z)\,e^{2i\theta(\xi,-\ell_2)t} & \frac{1-\dfrac{n^2(z)}{n^2(-\ell_2)}}{z+\ell_2}
		\end{pmatrix}, & z\in\operatorname{int}(\Upsilon_{-2}).
	\end{cases}
\end{equation}
Then $N_1^{(\mathrm{out})}(z)$ satisfies the following RH problem.
\begin{problem}\label{rh-jump-tran}
	Find a $2\times2$ matrix-valued function $N_1^{(\mathrm{out})}(z)$, analytic 
	in $\mathbb{C}\setminus\bigl(\Gamma\cup\Upsilon\bigr)$, such that:
	\begin{enumerate}
		\item $N_1^{(\mathrm{out})}(z)=I+\mathcal{O}(z^{-1}),\qquad |z|\to\infty$.
		
		\item For each $z\in\Gamma\cup\Upsilon$, the boundary values 
		$N_{1,\pm}^{(\mathrm{out})}(z)$ satisfy the jump relation
		\begin{equation*}
			N_{1,+}^{(\mathrm{out})}(z)=N_{1,-}^{(\mathrm{out})}(z)J_1^{(\mathrm{out})}(z),
		\end{equation*}
		where $z\in\Gamma\cup\Upsilon_1\cup\bar\Upsilon_1\cup\bar\Upsilon_{-1}\cup\Upsilon_{-1},$ we have
				\begin{equation*}
				J_1^{(\mathrm{out})}(z;x,t)=Y(z)^{-\sigma_3}\delta(z)^{-\sigma_3}J^{(1)}(z)\delta(z)^{\sigma_3}Y(z)^{\sigma_3},
			\end{equation*}
		and for $z\in\Upsilon_2\cup\bar\Upsilon_2\cup\bar\Upsilon_{-2}\cup\Upsilon_{-2}$, we obtain
		\begin{equation}\label{eq:J1-sol}
			J_1^{(\mathrm{out})}(z;x,t)=
			\begin{cases}
				\begin{pmatrix}
					1 & 0\\[4pt]
					\dfrac{n^2(\ell_2)\,\delta^2(\xi,z)\,e^{2it\theta(\xi,\ell_2)}}{c_2\,(z-\bar{\ell}_2)} & 1
				\end{pmatrix}, & z\in\operatorname{int}(\bar{\Upsilon}_2),\\
				\begin{pmatrix}
					1 & \dfrac{e^{-2i\theta(\xi,\bar{\ell}_2)t}\,\delta^{-2}(z)}{n^2(\bar{\ell}_2)\,\bar{c}_2\,(z-\ell_2)}\\[10pt]
					0 & 1
				\end{pmatrix}, & z\in\operatorname{int}(\Upsilon_2),\\
				\begin{pmatrix}
					1 & 0\\[4pt]
					-\dfrac{n^2(-\bar{\ell}_2)\,\delta^2(\xi,z)\,e^{2it\theta(\xi,-\bar{\ell}_2)}}{\bar{c}_2\,(z+\ell_2)} & 1
				\end{pmatrix}, & z\in\operatorname{int}(\Upsilon_{-2}),\\
				\begin{pmatrix}
					1 & -\dfrac{e^{-2i\theta(\xi,-\ell_2)t}\,\delta^{-2}(z)}{n^2(-\ell_2)\,c_2\,(z+\bar{\ell}_2)}\\[10pt]
					0 & 1
				\end{pmatrix}, & z\in\operatorname{int}(\bar{\Upsilon}_{-2}).
			\end{cases}
		\end{equation}
		
					\item At each point $z\in\mathcal{Z}_3$, the following residue conditions for $N_1^{(\mathrm{out})}(z)$ hold:
					\begin{subequations}\label{eq:res-conditionnn}
						\begin{align}
							\underset{z=\ell_j}{\operatorname{Res}}\,N(z)
							&=\lim_{z\to\ell_j}N(z)
							\begin{pmatrix}
								0 & -Y^{-2}(z)e^{-2it\theta(z)}c_j\\[2pt]
								0 & 0
							\end{pmatrix},\\[4pt]
							\underset{z=\bar{\ell}_j}{\operatorname{Res}}\,N(z)
							&=\lim_{z\to\bar{\ell}_j}N(z)
							\begin{pmatrix}
								0 & 0\\[2pt]
							Y^{2}(z)	e^{2it\theta(z)}\bar{c}_j & 0
							\end{pmatrix},\\[4pt]
							\underset{z=-\ell_j}{\operatorname{Res}}\,N(z)
							&=\lim_{z\to-\ell_j}N(z)
							\begin{pmatrix}
								0 & 0\\[2pt]
								-Y^{2}(z)e^{2it\theta(z)}c_j & 0
							\end{pmatrix},\\[4pt]
							\underset{z=-\bar{\ell}_j}{\operatorname{Res}}\,N(z)
							&=\lim_{z\to-\bar{\ell}_j}N(z)
							\begin{pmatrix}
								0 & Y^{-2}(z)e^{-2it\theta(z)}\bar{c}_j\\[2pt]
								0 & 0
							\end{pmatrix}.
						\end{align}
					\end{subequations}
		\item For each endpoint $\beta\in\{E_j,\bar{E}_j\}_{j=0}^n$,
		\begin{equation*}
			N_1^{(\mathrm{out})}(z)=\mathcal{O}\bigl((z-\beta)^{-1/4}\bigr),\qquad z\to\beta.
		\end{equation*}
	\end{enumerate}
\end{problem}
All the jumps of $N_1^{(\mathrm{out})}(z)$ with the exception of $z\in\Gamma$ 
tend to the identity exponentially fast as $t\to\infty$. Importantly, as a 
result of transformation \eqref{tran-soliton}, we anticipate that the leading-order 
contribution of the RH problem~\ref{rh-jump-tran} comes from the 
jump on $\Gamma$. As this jump depends on $z$ through the function 
$Y(z)$, prior to decomposing RH problem~\ref{rh-jump-tran} into dominant and 
error components we employ one further transformation that converts 
$J_1^{(\mathrm{out})}(z)$ into a constant jump. Specifically, we set
\begin{equation}\label{eq:N2-out}
	N_2^{(\mathrm{out})}(z)=e^{-ig(\infty)\sigma_3}\,N_1^{(\mathrm{out})}(z)\,
	e^{ig(z)\sigma_3},
\end{equation}
where the scalar function $g(z)$ is analytic in 
$\mathbb{C}\setminus(\Gamma\cup\Upsilon)$ and is defined as follows. Let
\begin{equation}\label{eq:mathcalH}
	\mathcal{H}(z)=
	\begin{cases}
		\dfrac{-i\ln\bigl(Y^{-2}(z)\delta^{-2}(z)\bigr)}{w_+(z)}, & 
		z\in \Upsilon_1\cup\bar{\Upsilon}_1\cup\Upsilon_2\cup\bar{\Upsilon}_2\cup\Gamma,\\[8pt]
		\dfrac{i\ln\bigl(Y^{-2}(z)\delta^{-2}(z)\bigr)}{w_+(z)}, & 
		z\in \Upsilon_{-1}\cup\bar{\Upsilon}_{-1}\cup\Upsilon_{-2}\cup\bar{\Upsilon}_{-2},
	\end{cases}
\end{equation}
and define $g(z)$ by the Cauchy integral
\begin{equation}\label{eq:g-def}
	g(z)=\frac{w(z)}{2\pi i}\int_{\Gamma\cup\Upsilon}\frac{\mathcal{H}(s)}{s-z}\,ds.
\end{equation}

\begin{lemma}\label{lem:g-properties}
	The function $g(z)$ defined by \eqref{eq:g-def} has the following 
	properties:
	\begin{enumerate}
		\item $g(z)$ obeys the Schwarz symmetry
		\begin{equation}\label{eq:g-sym}
			\overline{g(\bar{z})}=g(z),\qquad z\in\mathbb{C}\setminus(\Gamma\cup\Upsilon).
		\end{equation}
		
		\item As $z\to\infty$,
		\begin{equation}\label{eq:g-asymp}
			g(z)=g(\infty)+\frac{g^{(1)}}{z}+\mathcal{O}\bigl(z^{-2}\bigr),
		\end{equation}
		where $g(\infty)\equiv g(\xi,\infty)$ is a finite real constant given by
		\begin{equation}\label{eq:g-infty}
			g(\infty)=-\frac{1}{2\pi i}\int_{\Gamma\cup\Upsilon}\frac{\mathcal{H}(s)}{w(s)}\,ds,
		\end{equation}
		and
		\begin{equation}\label{eq:g1}
			g^{(1)}=-\frac{1}{2\pi i}\int_{\Gamma\cup\Upsilon}\frac{s\,\mathcal{H}(s)}{w(s)}\,ds.
		\end{equation}
		
		\item $e^{ig(z)\sigma_3}$ is a bounded and analytic function on 
		$\mathbb{C}\setminus(\Gamma\cup\Upsilon)$.
		
		\item $g(z)$ satisfies the jump conditions
		\begin{equation}\label{eq:g-jump}
			g_+(z)-g_-(z)=
			\begin{cases}
				-i\ln\bigl(Y^{-2}(z)\delta^{-2}(z)\bigr), & 
				z\in \Upsilon_1\cup\bar{\Upsilon}_1\cup\Upsilon_2\cup\bar{\Upsilon}_2\cup\Gamma,\\[4pt]
				i\ln\bigl(Y^{-2}(z)\delta^{-2}(z)\bigr), & 
				z\in\Upsilon_{-1}\cup\bar{\Upsilon}_{-1}\cup\Upsilon_{-2}\cup\bar{\Upsilon}_{-2},
			\end{cases}
		\end{equation}
	\end{enumerate}
\end{lemma}

\begin{proof}
The symmetry \eqref{eq:g-sym} and jump conditions \eqref{eq:g-jump} follow directly from the integral 
	representation \eqref{eq:g-def}. Expanding $(s-z)^{-1}$ in powers of $z^{-1}$ for $|z|>|s|$ and inserting 
	the expansion into \eqref{eq:g-def}, we obtain
	\[
	g(z)=\frac{w(z)}{2\pi i}\left(-\frac{1}{z}\int_{\Gamma\cup\Upsilon}\mathcal{H}(s)\,ds
	-\frac{1}{z^2}\int_{\Gamma\cup\Upsilon}s\,\mathcal{H}(s)\,ds
	+\mathcal{O}\bigl(z^{-3}\bigr)\right).
	\]
	Using $w(z)=z^{n+1}+\mathcal{O}(z^n)$ and the fact that the first $n$ moments 
	of $\mathcal{H}(z)$ vanish by virtue of the linear system \eqref{eq:delta_j} 
	(and its analogue for the $Y$-component), the leading term reduces to 
	\eqref{eq:g-asymp} with coefficients \eqref{eq:g-infty} and \eqref{eq:g1}. Boundedness and analyticity are immediate consequences of the integral 
	representation \eqref{eq:g-def}, since the density $\mathcal{H}(s)$ is 
	integrable on the compact contours $\Gamma\cup\Upsilon$ and the only 
	singularity of the integrand is a simple pole at $s=z$, which is 
	avoided by analyticity in the complement.
\end{proof}

Then $N^{(2)}(z)$ satisfies the following RH problem.
\begin{problem}\label{RH-2-out}
	Find a $2\times2$ matrix-valued function $N_2^{(\mathrm{out})}(z)$, analytic 
	in $\mathbb{C}\setminus\bigl(\Gamma\cup\Upsilon\bigr)$, such that:
	\begin{enumerate}
		\item $N_2^{(\mathrm{out})}(z)=I+\mathcal{O}(z^{-1}),\qquad |z|\to\infty$.
		
		\item For each $z\in\Gamma\cup\Upsilon$, the boundary values 
		$N_{2,\pm}^{(\mathrm{out})}(z)$ satisfy the jump relation
		\begin{equation*}
			N_{2,+}^{(\mathrm{out})}(z)=N_{2,-}^{(\mathrm{out})}(z)J_2^{(\mathrm{out})}(z),
		\end{equation*}
		where
		\begin{equation*}
			J_2^{(\mathrm{out})}(z)=J^{(\mathrm{alg})}(z),\qquad z\in\Gamma,
		\end{equation*}
		and for $z\in\Upsilon_2\cup\bar\Upsilon_2\cup\bar\Upsilon_{-2}\cup\Upsilon_{-2}$,
		\begin{equation}\label{eq:J2-sol}
			J_2^{(\mathrm{out})}(z)=
			\begin{cases}
				\begin{pmatrix}
					1 & 0\\[4pt]
					\dfrac{e^{2it\theta(\xi,\ell_2)}}{c_2\,(z-\bar{\ell}_2)} & 1
				\end{pmatrix},  z\in\operatorname{int}(\bar{\Upsilon}_2),\qquad
				\begin{pmatrix}
					1 & \dfrac{e^{-2it\theta(\xi,\bar{\ell}_2)}}{\bar{c}_2\,(z-\ell_2)}\\[10pt]
					0 & 1
				\end{pmatrix}, \quad z\in\operatorname{int}(\Upsilon_2),\\
				\begin{pmatrix}
					1 & 0\\[4pt]
					-\dfrac{e^{2it\theta(\xi,-\bar{\ell}_2)}}{\bar{c}_2\,(z+\ell_2)} & 1
				\end{pmatrix},  z\in\operatorname{int}(\Upsilon_{-2}),\quad
				\begin{pmatrix}
					1 & -\dfrac{e^{-2it\theta(\xi,-\ell_2)}}{c_2\,(z+\bar{\ell}_2)}\\
					0 & 1
				\end{pmatrix},  z\in\operatorname{int}(\bar{\Upsilon}_{-2}).
			\end{cases}
		\end{equation}
					\item At each point $z\in\mathcal{Z}_3$, the residue 
		condition is given by the corresponding entry in \eqref{eq:res-conditionnn}.
		
		\item For each endpoint $\beta\in\{E_j,\bar{E}_j\}_{j=0}^n$,
		\begin{equation*}
			N_2^{(\mathrm{out})}(z)=\mathcal{O}\bigl((z-\beta)^{-1/4}\bigr),\qquad z\to \beta.
		\end{equation*}
	\end{enumerate}
\end{problem}
We now analyze the contribution from the discrete spectrum 
$\mathcal{Z}_3=\{\ell_{3},\bar\ell_{3},-\ell_{3},-\bar\ell_{3}\}$. 
Recall that the jumps associated with the soliton component are given by \eqref{eq:jump-soliton}. 
When $\Im\theta(\ell_3)=0$, the time-dependent exponentials $e^{\pm 2it\theta(z)}$ 
appearing in these jumps are purely oscillatory, rather than exponentially 
decaying or growing. 
This is in sharp contrast to the generic case where $\Im\theta(\ell_2)>0$ 
and $\Im\theta(\ell_1)<0$. Since $\det N^{(\mathrm{alg})}(z)\equiv 1\neq 0$ 
for all $z\in\mathbb{C}$, we may rearrange the local factorization 
\eqref{eq:N-1-decompose} to isolate the discrete spectrum component:
\begin{equation}\label{eq:N-sol-inverse}
	N^{(\mathrm{sol})}(z)=N^{(\mathrm{out})}(z)\bigl(N^{\mathrm{(alg)}}(z)\bigr)^{-1}.
\end{equation}
It follows immediately from \eqref{eq:N-sol-inverse} and the jump structure 
of $N^{(\mathrm{out})}$ that $N^{(\mathrm{sol})}(z)$ is continuous across 
the contour $\Gamma$; consequently, $N^{(\mathrm{sol})}$ is meromorphic on 
$\mathbb{C}$ with simple poles located at the points of $\mathcal{Z}_3$.

Using \eqref{eq:N-sol-inverse} and the fact that $N^{(\mathrm{out})}$ is regular at the soliton points, the residue conditions for $N^{(\mathrm{sol})}$ are obtained by conjugating the residue matrices of $N^{(\mathrm{out})}(z)$ by $\bigl(N^{(\mathrm{alg})}\bigr)^{-1}(z)$. Explicitly,
\begin{subequations}\label{eq:res-trans}
	\begin{align}
		\underset{z=\ell_{3}}{\operatorname{Res}}\,N^{(\mathrm{sol})}(z)
		&=\lim_{z\to\ell_{3}}N^{(\mathrm{out})}(z)
		\begin{pmatrix} 0 & -Y^{-2}(\ell_3)e^{-2it\theta(\ell_{3})}c_3 \\[2pt] 0 & 0 \end{pmatrix}
		\bigl(N^{(\mathrm{alg})}(z)\bigr)^{-1}, \\
		\underset{z=\bar\ell_{3}}{\operatorname{Res}}\,N^{(\mathrm{sol})}(z)
		&=\lim_{z\to\bar\ell_{3}}N^{(\mathrm{out})}(z)
		\begin{pmatrix} 0 & 0 \\[2pt] Y^{2}(\bar\ell_3)e^{2it\theta(\bar\ell_{3})}\bar{c}_3 & 0 \end{pmatrix}
		\bigl(N^{(\mathrm{alg})}(z)\bigr)^{-1}, \\
		\underset{z=-\ell_{3}}{\operatorname{Res}}\,N^{(\mathrm{sol})}(z)
		&=\lim_{z\to-\ell_{3}}N^{(\mathrm{out})}(z)
		\begin{pmatrix} 0 & 0 \\[2pt] -Y^{2}(-\ell_3)e^{2it\theta(-\ell_{3})}c_3 & 0 \end{pmatrix}
		\bigl(N^{(\mathrm{alg})}(z)\bigr)^{-1}, \\
		\underset{z=-\bar\ell_{3}}{\operatorname{Res}}\,N^{(\mathrm{sol})}(z)
		&=\lim_{z\to-\bar\ell_{3}}N^{(\mathrm{out})}(z)
		\begin{pmatrix} 0 & Y^{-2}(-\ell_3)e^{-2it\theta(-\bar\ell_{3})}\bar{c}_3 \\[2pt] 0 & 0 \end{pmatrix}
		\bigl(N^{(\mathrm{alg})}(z)\bigr)^{-1}.
	\end{align}
\end{subequations}
Since $N^{(\mathrm{sol})}(z)$ is analytic in $\mathbb{C}\setminus\{\ell_{3},\bar\ell_{3},-\ell_{3},-\bar\ell_{3}\}$ and satisfies the normalization condition
\begin{equation}\label{eq:N-sol-asymp}
	N^{(\mathrm{sol})}(z)\to I,\qquad z\to\infty,
\end{equation}
Liouville's theorem implies that $N^{(\mathrm{sol})}(z)$ is the unique rational matrix function with simple poles at the indicated points and prescribed residues. Consequently,
\begin{equation}\label{eq:res-compute}
	N^{(\mathrm{sol})}(z)=I+\frac{\operatorname{Res}_{\ell_{3}}N^{(\mathrm{sol})}}{z-\ell_{3}}
	+\frac{\operatorname{Res}_{\bar\ell_{3}}N^{(\mathrm{sol})}}{z-\bar\ell_{3}}
	+\frac{\operatorname{Res}_{-\ell_{3}}N^{(\mathrm{sol})}}{z+\ell_{3}}
	+\frac{\operatorname{Res}_{-\bar\ell_{3}}N^{(\mathrm{sol})}}{z+\bar\ell_{3}}.
\end{equation}
It remains to determine the residues of $N^{(\mathrm{out})}(z)$ at these same points. 
Combining \eqref{eq:N-sol-inverse}, \eqref{eq:res-trans}, and \eqref{eq:res-compute}, 
we obtain the following algebraic system for the columns of $N^{(\mathrm{out})}(z)$:
\begin{subequations}\label{eq:res-out-cal}
	\begin{align}
		N^{(\mathrm{out})}_1(z) &= N^{(\mathrm{alg})}_1(z) 
		+ \mathcal{W}_1(z)\,N^{(\mathrm{alg})}_{11}(z) 
		+ \mathcal{W}_2(z)\,N^{(\mathrm{alg})}_{21}(z),\\
		N^{(\mathrm{out})}_2(z) &= N^{(\mathrm{alg})}_2(z) 
		+ \mathcal{W}_1(z)\,N^{(\mathrm{alg})}_{12}(z) 
		+ \mathcal{W}_2(z)\,N^{(\mathrm{alg})}_{22}(z),
	\end{align}
\end{subequations}
where the scalar rational functions $\mathcal{W}_1,\mathcal{W}_2$ are defined by
\begin{align}\label{eq:W-compact}
	\mathcal{W}_1(z) &= \frac{\mathfrak{a}_1\,\nu_{21}^{(1)}}{z-\ell_{3}} 
	+ \frac{\mathfrak{a}_2\,\nu_{22}^{(2)}}{z-\bar\ell_{3}} 
	- \frac{\mathfrak{a}_3\,\nu_{22}^{(3)}}{z+\ell_{3}} 
	- \frac{\mathfrak{a}_4\,\nu_{21}^{(4)}}{z+\bar\ell_{3}},\\[4pt]
	\mathcal{W}_2(z) &= -\frac{\mathfrak{a}_1\,\nu_{11}^{(1)}}{z-\ell_{3}} 
	- \frac{\mathfrak{a}_2\,\nu_{12}^{(2)}}{z-\bar\ell_{3}} 
	+ \frac{\mathfrak{a}_3\,\nu_{12}^{(3)}}{z+\ell_{3}} 
	+ \frac{\mathfrak{a}_4\,\nu_{11}^{(4)}}{z+\bar\ell_{3}}.
\end{align}
with coefficients
\begin{equation}\label{eq:V-coeffs}
	\begin{aligned}
		&	\mathfrak{a}_1 = c_3 Y^{-2}(\ell_3)e^{-2it\theta(\ell_3)} N_1^{(\mathrm{out})}(\ell_3), \quad\,\,\,
		\mathfrak{a}_2 = \bar c_3 Y^{2}(\bar\ell_3)e^{2it\theta(\bar\ell_3)} N_2^{(\mathrm{out})}(\bar\ell_3),\\
		&	\mathfrak{a}_3 = c_3Y^{2}(-\ell_3) e^{2it\theta(-\ell_3)} N_2^{(\mathrm{out})}(-\ell_3),  \quad
		\mathfrak{a}_4 = \bar c_3 Y^{-2}(-\bar\ell_3)e^{-2it\theta(-\bar\ell_3)} N_1^{(\mathrm{out})}(-\bar\ell_3),\\
		&	\nu_{ij}^{(1)} = N_{ij}^{(\mathrm{alg})}(\ell_3), \quad
		\nu_{ij}^{(2)} = N_{ij}^{(\mathrm{alg})}(\bar\ell_3),\quad	\nu_{ij}^{(3)} = N_{ij}^{(\mathrm{alg})}(-\ell_3),\quad 	\nu_{ij}^{(4)} = N_{ij}^{(\mathrm{alg})}(-\bar\ell_3).
	\end{aligned}
\end{equation}
Since $	N_2^{(\mathrm{out})}(z)$ is analytic at $\bar\ell_{3}$ and $-\ell_{3}$, 
we evaluate \eqref{eq:res-out-cal} to obtain
\begin{subequations}\label{eq:N-sol-cal}
	\begin{align}
	&	N_2^{(\mathrm{out})}(\bar\ell_{3}) = N_{12}^{(\mathrm{alg})}(\bar\ell_{3}) 
		+ \mathcal{V}_{11} N_{12}^{(\mathrm{alg})}(\bar\ell_{3}) 
		+ \mathcal{V}_{12} N_{22}^{(\mathrm{alg})}(\bar\ell_{3}),\\
	&	N_2^{(\mathrm{out})}(-\ell_{3}) = N_{22}^{(\mathrm{alg})}(-\ell_{3}) 
		+ \mathcal{V}_{21} N_{12}^{(\mathrm{alg})}(-\ell_{3}) 
		+ \mathcal{V}_{22} N_{22}^{(\mathrm{alg})}(-\ell_{3}),
	\end{align}
\end{subequations}
where the constants $\mathcal{V}_{jk}$ are given by
\begin{subequations}\label{eq:V-cal}
	\begin{align}
	&	\mathcal{V}_{11} = \frac{\mathfrak{a}_1\nu_{21}^{(1)}}{\bar\ell_3-\ell_3} 
		+ \mathfrak{a}_2(\nu_{22}^{(2)})' 
		- \frac{\mathfrak{a}_3\nu_{22}^{(3)}}{\bar\ell_3+\ell_3} 
		- \frac{\mathfrak{a}_4\nu_{21}^{(4)}}{2\bar\ell_3},\\
	&	\mathcal{V}_{12} = -\frac{\mathfrak{a}_1\nu_{11}^{(1)}}{\bar\ell_3-\ell_3} 
		- \mathfrak{a}_2(\nu_{12}^{(2)})' 
		+ \frac{\mathfrak{a}_3\nu_{12}^{(3)}}{\bar\ell_3+\ell_3} 
		+ \frac{\mathfrak{a}_4\nu_{11}^{(4)}}{2\bar\ell_3},\\
	&	\mathcal{V}_{21} = \frac{\mathfrak{a}_1\nu_{21}^{(1)}}{-2\ell_3} 
		+ \frac{\mathfrak{a}_2\nu_{22}^{(2)}}{-\ell_3-\bar\ell_3} 
		- \mathfrak{a}_3(\nu_{22}^{(3)})' 
		- \frac{\mathfrak{a}_4\nu_{21}^{(4)}}{\bar\ell_3-\ell_3},\\
	&	\mathcal{V}_{22} = -\frac{\mathfrak{a}_1\nu_{11}^{(1)}}{-2\ell_3} 
		- \frac{\mathfrak{a}_2\nu_{12}^{(2)}}{-\ell_3-\bar\ell_3} 
		+ \mathfrak{a}_3(\nu_{12}^{(3)})' 
		+ \frac{\mathfrak{a}_4\nu_{11}^{(4)}}{\bar\ell_3-\ell_3},
	\end{align}
\end{subequations}
Similarly, since $N_1^{(\mathrm{out})}(z)$ is analytic at $-\bar\ell_{3}$ and $\ell_{3}$, 
we evaluate \eqref{eq:res-out-cal} to obtain
\begin{subequations}\label{eq:N-sol-cal-tilde}
	\begin{align}
	&	N_1^{(\mathrm{out})}(-\bar\ell_{3}) = N_{11}^{(\mathrm{alg})}(-\bar\ell_{3}) 
		+ \widetilde{\mathcal{V}}_{11} N_{11}^{(\mathrm{alg})}(-\bar\ell_{3}) 
		+ \widetilde{\mathcal{V}}_{12} N_{21}^{(\mathrm{alg})}(-\bar\ell_{3}),\\
	&	N_1^{(\mathrm{out})}(\ell_{3}) = N_{21}^{(\mathrm{alg})}(\ell_{3}) 
		+ \widetilde{\mathcal{V}}_{21} N_{11}^{(\mathrm{alg})}(\ell_{3}) 
		+ \widetilde{\mathcal{V}}_{22} N_{21}^{(\mathrm{alg})}(\ell_{3}),
	\end{align}
\end{subequations}
where
\begin{subequations}\label{eq:V-cal-tilde}
	\begin{align}
	&	\widetilde{\mathcal{V}}_{11} = \frac{\mathfrak{a}_1\nu_{21}^{(1)}}{-\bar\ell_3-\ell_3} 
		+ \frac{\mathfrak{a}_2\nu_{22}^{(2)}}{-2\bar\ell_3} 
		- \frac{\mathfrak{a}_3\nu_{22}^{(3)}}{\ell_3-\bar\ell_3} 
		- \mathfrak{a}_4(\nu_{21}^{(4)})',\\
	&	\widetilde{\mathcal{V}}_{12} = -\frac{\mathfrak{a}_1\nu_{11}^{(1)}}{-\bar\ell_3-\ell_3} 
		- \frac{\mathfrak{a}_2\nu_{12}^{(2)}}{-2\bar\ell_3} 
		+ \frac{\mathfrak{a}_3\nu_{12}^{(3)}}{\ell_3-\bar\ell_3} 
		+ \mathfrak{a}_4(\nu_{11}^{(4)})',\\
	&	\widetilde{\mathcal{V}}_{21} = \mathfrak{a}_1(\nu_{21}^{(1)})' 
		+ \frac{\mathfrak{a}_2\nu_{22}^{(2)}}{\ell_3-\bar\ell_3} 
		- \frac{\mathfrak{a}_3\nu_{22}^{(3)}}{2\ell_3} 
		- \frac{\mathfrak{a}_4\nu_{21}^{(4)}}{\ell_3+\bar\ell_3},\\
	&	\widetilde{\mathcal{V}}_{22} = -\mathfrak{a}_1(\nu_{11}^{(1)})' 
		- \frac{\mathfrak{a}_2\nu_{12}^{(2)}}{\ell_3-\bar\ell_3} 
		+ \frac{\mathfrak{a}_3\nu_{12}^{(3)}}{2\ell_3} 
		+ \frac{\mathfrak{a}_4\nu_{11}^{(4)}}{\ell_3+\bar\ell_3}.
	\end{align}
\end{subequations}
We now establish the existence and uniqueness of the outer parametrix 
$N^{(\mathrm{out})}(z)$, which incorporates the discrete spectrum 
$\mathcal{Z}_3=\{\ell_3,\bar\ell_3,-\ell_3,-\bar\ell_3\}$ into the 
finite-genus background $N^{(\mathrm{alg})}(z)$.
\begin{proposition}\label{prop:outer-exist}
	The RH problem \ref{RH-2-out} admits a unique solution 
	$N^{(\mathrm{out})}(z)$ for all $(x,t)\in\mathbb{R}^2$. Moreover, the outer parametrix 
	$N^{(\mathrm{out})}(z)$ on the discrete spectrum $\mathcal{Z}_3$ is uniquely determined 
	by the linear algebraic system \eqref{eq:N-sol-cal} and \eqref{eq:N-sol-cal-tilde}. 
	The associated solution is recovered via
		\begin{align}\nonumber
			& 2i e^{2i(xp_0+tq_0)}
			\lim_{|z|\to\infty} z\bigl(\delta^{\sigma_3}(\infty)
			e^{ig(\infty)\sigma_3}N^{(\mathrm{out})}(z)
			e^{-ig(z)\sigma_3}\delta^{-\sigma_3}(z)Y(z)^{-\sigma_3}
			\bigr)_{12}\\\nonumber
			& =2i e^{2i(xp_0+tq_0)}\delta^{2}(\infty)e^{2ig(\infty)}
			\lim_{|z|\to\infty} z\bigl(
		N^{(\mathrm{alg})}(z)N^{(\mathrm{sol})}(z)
			Y(z)^{-\sigma_3}
			\bigr)_{12}\\\nonumber
			&= 2i e^{2i(xp_0+tq_0)} \delta^{2}(\infty)e^{2ig(\infty)}
			\lim_{|z|\to\infty} \bigl(z
		\bigl(I+\tfrac{N_1^{(\mathrm{alg})}}{z}\bigr)
			\bigl(I+\tfrac{N_1^{(\mathrm{sol})}}{z}\bigr)
		Y(z)^{-\sigma_3}
			\bigr)_{12}\\\nonumber
			&=2i e^{2i(xp_0+tq_0)} \delta^{2}(\infty)e^{2ig(\infty)}
			\bigl(N_1^{(\mathrm{alg})}+N_1^{(\mathrm{sol})}\bigr)\\\label{eq:breather-recon}
			&=	u^{\mathrm{(sol)}}(x,t;\ell_3)+2i e^{2i(xp_0+tq_0)} \delta^{2}(\infty)e^{2ig(\infty)}u^{\mathrm{(alg)}}(x,t),
		\end{align}
\end{proposition}
\begin{proof}
	Equations \eqref{eq:N-sol-cal} and \eqref{eq:N-sol-cal-tilde} constitute a closed 
	linear algebraic system for the boundary values 
	$(N_1^{(\mathrm{out})}(z_j),N_2^{(\mathrm{out})}(z_j))$, $z_j\in\mathcal{Z}_3$. 
	For generic scattering data, this system is nonsingular, hence the boundary values 
	are uniquely determined. The reconstruction formula \eqref{eq:breather-recon} then 
	follows from the standard dressing argument and the large-$z$ expansion of 
	$N^{(\mathrm{out})}(z)=N^{(\mathrm{alg})}(z)N^{(\mathrm{sol})}(z)$, noting that 
	the $z^{-1}$-coefficient of $N^{(\mathrm{out})}$ is 
	$N_1^{(\mathrm{alg})}(z)+N_1^{(\mathrm{sol})}(z)$.
\end{proof}
\subsection{The local models: parametrices near the stationary phase points}

The outer model $N^{(\mathrm{out})}(z)$ constructed in the previous section accurately approximates the solution of the RH problem~\ref{RH-1} everywhere in the complex plane except in small neighbourhoods of the branch points $E_{j_0}$ and the stationary phase points, where the jump matrices fail to be close to the identity. Near the real stationary phase points $\kappa^{\rm R}_1$ and $\kappa^{\rm R}_2$, which contribute at order $\mathcal{O}(t^{-1/2})$ to the long-time asymptotics of the Cauchy problem \eqref{eq:mkdv}--\eqref{eq:q_0}, the situation is different.  Consequently, a local parametrix $N^{(\mathrm{loc})}(z)$ must be constructed in a disk $U_{E_{j_0}}$ centred at $E_{j_0}$ such that it satisfies exactly the jump condition specified in the following RH problem.

\begin{problem}\label{rh:locc}
Find a 2$\times$2 matrix-valued function $N^{(\mathrm{loc})}(z)$, analytic in $\mathbb{C}\setminus\bigl(L_{j_0}\cup L_{j_0}\cup L_{n-j_0}\cup  L_{n-j_0}\bigr)$, such that:
	\begin{enumerate}
		\item $N^{(\mathrm{loc})}(z)=I+\mathcal{O}(z^{-1})$ as $|z|\to\infty$.
		\item For each $z\in L_{j_0}\cup L_{j_0}\cup L_{n-j_0}\cup L_{n-j_0}$, the boundary values $N^{(\mathrm{loc})}(z)$ satisfy the jump relation
		\[
		N_+^{(\mathrm{loc})}(z)=N_-^{(\mathrm{loc})}(z)J^{(\mathrm{loc})}(z)
		\]
		where
		\[
		J^{(\mathrm{loc})}(z)=J^{(1)}(z)\big|_{z\in (L_{j_0}\cup L_{j_0}\cup L_{n-j_0}\cup L_{n-j_0})}.
		\]
	\end{enumerate}
\end{problem}

To this end, we observe that near $z=E_j$ the phase function admits the following expansion.

\begin{lemma}\label{lem:endpoint-expansion}
	Let $E_j$ be an endpoint of the branch cut $\Gamma_j$ and assume the non-degeneracy condition
	\begin{equation}\label{eq:non-degeneracy}
		\theta'(E_j;\xi)\neq 0.
	\end{equation}
	Then there exists a sufficiently small neighbourhood 
	$U_{E_j}:=\{z\in\mathbb{C}:|z-E_j|<\epsilon\}$ such that, for $z\in U_{E_j}$,
	\begin{equation}\label{eq:theta-local}
		\theta(z;\xi)=\theta(E_j;\xi)+c_{j,1}(\xi,E_j)\,(z-E_j)^{1/2}
		+\frac{2}{3}c_{j,2}(\xi,E_j)\,(z-E_j)^{3/2}
		+O\bigl((z-E_j)^{5/2}\bigr),
	\end{equation}
	where the coefficients are given explicitly by
	\begin{subequations}\label{eq:coefficients}
		\begin{align}
			c_{j,1}(\xi,E_j)&=\frac{2\bigl(\xi P_{n+1}(E_j)+Q_{n+3}(E_j)\bigr)}
			{A_j}, \label{eq:c1} \\[4pt]
			c_{j,2}(\xi,E_j)&=\frac{\xi\bigl[P_{n+1}'(E_j)-B_j P_{n+1}(E_j)\bigr]
				+\tfrac{1}{2}\bigl[Q_{n+3}'(E_j)-B_j Q_{n+3}(E_j)\bigr]}
			{A_j}, \label{eq:c2}
		\end{align}
	\end{subequations}
	where $A_j$ and $B_j$ are constants determined by $\prod_{k\neq j}(z-E_k)^{1/2}$ near $z=E_j$, and $P_{n+1}(z)$ and $Q_{n+3}(z)$ are the polynomials appearing in the integrands of $p(z)$ and $q(z)$ in \eqref{eq:f(z)} and \eqref{eq:g(z)} respectively.
\end{lemma}

\begin{proof}
	Since $p(z)$ and $q(z)$ are Abelian integrals of the second kind, their 
	derivatives are rational functions on $\mathcal{R}$:
	\begin{equation}\label{eq:pq-derivatives}
		p'(z)=\frac{P_{n+1}(z)}{w(z)}, \qquad 
		q'(z)=\frac{Q_{n+3}(z)}{w(z)},
	\end{equation}
  Near the endpoint 
	$E_j$, writing $w(z)=(z-E_j)^{1/2}\sqrt{\prod_{k\neq j}(z-E_k)}$ and 
	using \eqref{eq:theta}, one obtains
	\begin{equation}\label{eq:theta-prime-local}
		\theta'(z;\xi)=\frac{\xi P_{n+1}(z)+Q_{n+3}(z)}
		{(z-E_j)^{1/2}\prod_{k\neq j}(z-E_k)^{1/2}}
		= \frac{c_{j,1}(\xi)}{2}(z-E_j)^{-1/2}
		+ \frac{3c_{j,2}(\xi)}{2}(z-E_j)^{1/2}
		+ O\bigl((z-E_j)^{3/2}\bigr).
	\end{equation}
	Integrating term by term along a path in $U_{E_j}$ avoiding the branch 
	cut yields \eqref{eq:theta-local}--\eqref{eq:coefficients}.
\end{proof}

In the generic situation $\theta'(E_j;\xi)\neq 0$, the leading-order behaviour is linear in the local coordinate, and the jump matrices can be factorised in such a way that the local model reduces to a standard RH problem whose solution is explicitly given in terms of Bessel functions. This contributes $\mathcal{O}(t^{-1})$ as $t\to\infty$, and the estimate holds uniformly on the boundary of the local disk $z\in\partial U_{E_j}$.

When the complex stationary phase point $\kappa^C_{s_0}$ coalesces into $E_{j_0}$,
we define $\zeta\colon U_{E_{j_0}}\to\mathbb{C}$ by
\begin{equation}\label{eq:omega-explicit}
	\zeta(z):=
	\left(
	\frac{3it}{2}\bigl(\theta(z;\xi)-\theta(E_{j_0};\xi)\bigr)
	-\frac{3it}{2}\,c_{j,1}(\xi,E_{j_0})\,(z-E_{j_0})^{1/2}
	\right)^{\!2/3},
	\qquad z\in U_{E_{j_0}},
\end{equation}
where the square root $(z-E_{j_0})^{1/2}$ is taken with its branch cut along 
$\Gamma_j$, and the $2/3$-power is fixed so that $\zeta(z)$ is 
real-valued and increasing for $z\in\Gamma_j\cap U_{E_{j_0}}$ oriented 
away from $E_{j_0}$.
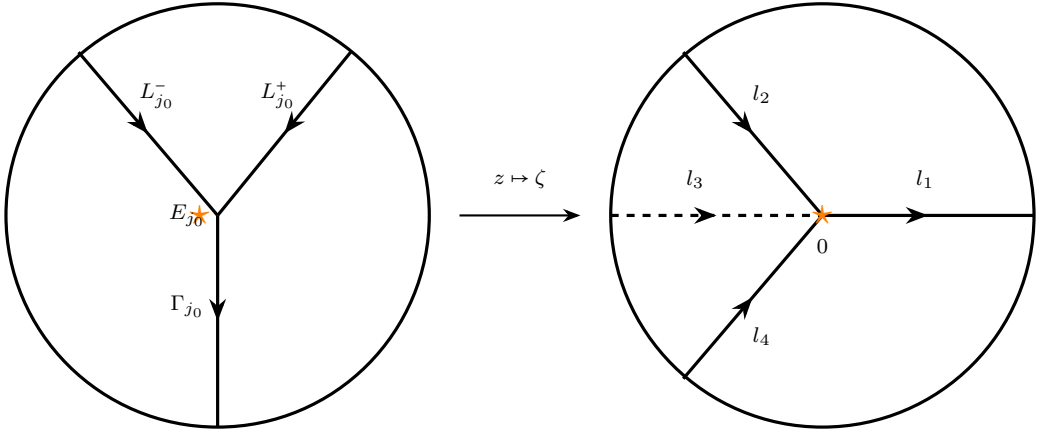
\begin{figure}[htp]
	\centering
	\begin{tikzpicture}[scale=0.8, >=Stealth, font=\footnotesize]
		% 左右两个圆盘边界
		\draw[very thick] (-5,0) circle (3.5cm);
		\draw[very thick] (5,0) circle (3.5cm);
		
		% 左圆：割线与跳跃轮廓
		\draw[very thick, postaction={decorate}, decoration={markings, mark=at position 0.5 with {\arrow{Stealth}}}] 
		(-5,0) -- (-5,-3.5);
		\draw[very thick, postaction={decorate}, decoration={markings, mark=at position 0.5 with {\arrow{Stealth}}}] 
		(-7.3,2.7) -- (-5,0);
		\draw[very thick, postaction={decorate}, decoration={markings, mark=at position 0.5 with {\arrow{Stealth}}}] 
		(-2.8,2.7) -- (-5,0);
		
		% 右圆：跳跃轮廓
		\draw[very thick, postaction={decorate}, decoration={markings, mark=at position 0.5 with {\arrow{Stealth}}}] 
		(5,0) -- (8.5,0);
		\draw[very thick, postaction={decorate}, decoration={markings, mark=at position 0.5 with {\arrow{Stealth}}}] 
		(2.7,2.7) -- (5,0);
		\draw[very thick, postaction={decorate}, decoration={markings, mark=at position 0.5 with {\arrow{Stealth}}}] 
		(2.7,-2.7) -- (5,0);
		\draw[very thick, dashed, postaction={decorate}, decoration={markings, mark=at position 0.5 with {\arrow{Stealth}}}] 
		(1.5,0) -- (5,0);
		
		% 中间映射箭头
		\draw[thick, ->] (-1,0) -- (1,0);
		
		% 中心点（橙色星号，放大）
		\node[orange, scale=2.2] at (-5.3,0) {$\star$};
		\node[orange, scale=2.2] at (5,0) {$\star$};
		
		% 标签
		\node at (-5.5,0) {$E_{j_0}$};
		\node at (-4,2) {$L^+_{j_0}$};
		\node at (-6,2) {$L^-_{j_0}$};
		\node at (-5.5,-1.5) {$\Gamma_{j_0}$};
		\node at (0,0.6) {$z\mapsto\zeta$};
		\node at (5,-0.5) {$0$};
		\node at (6.7,0.6) {$l_1$};
		\node at (4,2) {$l_2$};
		\node at (4,-2) {$l_4$};
		\node at (2.9,0.6) {$l_3$};
	\end{tikzpicture}
	\caption{Contour deformation under the map $z\mapsto\zeta$, where the (\textcolor{orange}{$\star$}) mark the complex stationary phase points.}
	\label{ttt6ttt66ttttt}
\end{figure}
\begin{lemma}\label{lem:airy-id}
	For $z\in U_{E_{j_0}}$, there holds
	\begin{equation}\label{eq:theta-omega}
		it\bigl(\theta(z;\xi)-\theta(E_{j_0};\xi)\bigr)
		=\frac{2}{3}\,\zeta(z)^{3/2}+\omega(z)\,\zeta(z)^{1/2},
	\end{equation}
	where $\zeta(z)^{1/2}$ denotes the analytic branch of the square root 
	in $U_{E_{j_0}}\setminus\Gamma_j$ that is positive on 
	$\Gamma_j\cap U_{E_{j_0}}$, and the analytic function 
	$\omega\colon U_{E_{j_0}}\setminus\Gamma_j\to\mathbb{C}$ is given by
	\begin{equation}\label{eq:omega-def}
		\omega(z)\coloneqq
		it\,c_{j,1}(\xi,E_{j_0})\,
		\frac{(z-E_{j_0})^{1/2}}{\zeta(z)^{1/2}}.
	\end{equation}
	Moreover, the function $\omega(z)$ extends continuously to $E_{j_0}$ and
	\begin{equation}\label{eq:omega-limit}
		\omega(E_{j_0})
		\coloneqq\lim_{\substack{z\to E_{j_0}\\z\notin\Gamma_j}}\omega(z)
		=\frac{c_{j,1}(\xi,E_{j_0})}{|c_{j,2}(\xi,E_{j_0})|^{1/3}}\,t^{2/3}
		\in\mathbb{R}.
	\end{equation}
\end{lemma}

\begin{proof}
	From \eqref{eq:omega-explicit} we have
	\[
	\zeta(z)^{3/2}
	=\frac{3it}{2}\bigl(\theta(z)-\theta(E_{j_0})\bigr)
	-\frac{3it}{2}c_{j,1}(\xi,E_{j_0})\,(z-E_{j_0})^{1/2}.
	\]
	Multiplying by $2/3$ and rearranging gives
	\[
	it\bigl(\theta(z)-\theta(E_{j_0})\bigr)
	=\frac{2}{3}\,\zeta(z)^{3/2}
	+it\,c_{j,1}(\xi,E_{j_0})\,(z-E_{j_0})^{1/2}.
	\]
	Since $\zeta(z)$ does not vanish identically in 
	$U_{E_{j_0}}\setminus\{E_{j_0}\}$, we factor the remainder as 
	$\omega(z)\,\zeta(z)^{1/2}$ with $\omega(z)$ defined by 
	\eqref{eq:omega-def}, yielding \eqref{eq:theta-omega}.
	
	To establish \eqref{eq:omega-limit}, we first derive the local 
	expansion of $\zeta(z)$ near $E_{j_0}$.  Substituting 
	\eqref{eq:theta-local} into the bracket of \eqref{eq:omega-explicit} 
	and using the cancellation of the half-power terms, we obtain
	\begin{equation}\label{eq:zeta-conformal}
		\zeta(z)=(it\,c_{j,2}(\xi,E_{j_0}))^{2/3}(z-E_{j_0})
		+O\bigl((z-E_{j_0})^2\bigr),\qquad z\to E_{j_0}.
	\end{equation}
	It follows from \eqref{eq:zeta-conformal} that
	\[
	\lim_{z\to E_{j_0}}\frac{(z-E_{j_0})^{1/2}}{\zeta(z)^{1/2}}
	=\frac{1}{(it\,c_{j,2}(\xi,E_{j_0}))^{1/3}}.
	\]
	Substituting this into \eqref{eq:omega-def} yields
	\begin{equation}\label{eq:pmeggaaa}
	\omega(E_{j_0})
	=it\,c_{j,1}(\xi,E_{j_0})\cdot\frac{1}{(it\,c_{j,2}(\xi,E_{j_0}))^{1/3}}
	=c_{j,1}(\xi,E_{j_0})\,\frac{(it)^{2/3}}{c_{j,2}(\xi,E_{j_0})^{1/3}}.
	\end{equation}
	The choice of the $2/3$-branch in \eqref{eq:omega-explicit}, dictated 
	by the signature-table geometry and the deformation of the jump contour 
	onto the positive real $\zeta$-axis, fixes the argument of 
	$(it\,c_{j,2})^{2/3}$ so that the ratio $(it)^{2/3}/c_{j,2}^{1/3}$ is 
	real and positive.  Writing this real value as 
	$t^{2/3}/|c_{j,2}(\xi,E_{j_0})|^{1/3}$ yields \eqref{eq:omega-limit}.
\end{proof}
		
Construct a meromorphic function $N_1^{(\mathrm{loc})}(z)$ such that
\[
N_1^{(\mathrm{loc})}(z)=N^{(\mathrm{loc})}(z)\,e^{it\theta(p_0)\sigma_3},
\quad z\in U_0=\cup_{p_0\in\mathcal{P}_2}U_{p_0},
\]
where then we have the following RH problem.

\begin{problem}\label{rh:loc}
	Find a $2\times2$ matrix-valued function $N^{(\mathrm{loc})}(z)$, analytic in 
	$\mathbb{C}\setminus\bigcup_{p_0\in\mathcal{P}_2}L_{p_0}$, such that:
	\begin{enumerate}
		\item $N^{(\mathrm{loc})}(z)=I+\mathcal{O}(z^{-1})$,\quad $|z|\to\infty$.
		\item For $z\in\bigcup_{p_0\in\mathcal{P}_2}L_{p_0}$, the boundary values satisfy
		\[
		N_+^{(\mathrm{loc})}(z)=N_-^{(\mathrm{loc})}(z)J^{(\mathrm{loc})}(z),
		\]
		where, for each $p_0\in\mathcal{P}_2$, we set
		\begin{equation}\label{eq:loc-params}
			\rho(p_0)=\delta^2(p_0)e^{i\phi_{p_0}},\qquad 
			\theta_{p_0}(\zeta)=\tfrac{2}{3}\zeta^{3/2}+\omega(p_0)\zeta^{1/2},
		\end{equation}
		and define the sign functions
		\[
		\varepsilon(p_0)=
		\begin{cases}
			-1,&p_0\in\{E_{j_0},E_{n-j_0}\},\\[2pt]
			+1,&p_0\in\{\bar E_{j_0},\bar E_{n-j_0}\},
		\end{cases}
		\qquad
		\kappa(p_0)=
		\begin{cases}
			-1,&p_0\in\{E_{j_0},\bar E_{j_0}\},\\[2pt]
			+1,&p_0\in\{E_{n-j_0},\bar E_{n-j_0}\}.
		\end{cases}
		\]
		Then the jump on $L_{p_0}=L_{p_0}^-\cup(\Gamma_{p_0}\cap U_{p_0})\cup L_{p_0}^+$ is given by
		\begin{equation}\label{eq:jump-unified}
			J^{(\mathrm{loc})}(z)=
			\begin{cases}
				\varepsilon(p_0)\,i\sigma_1\,e^{\alpha(p_0)\sigma_3},
				& z\in\Gamma_{p_0}\cap U_{p_0},\\[10pt]
				\begin{pmatrix}
					1 & 0 \\[2pt]
					\kappa(p_0)\,\rho(p_0)\,e^{\,2\theta_{p_0}(\zeta)} & 1
				\end{pmatrix},
				& z\in L_{p_0}^-,\\[14pt]
				\begin{pmatrix}
					1 & \kappa(p_0)\,\rho(p_0)^{-1}\,e^{-\,2\theta_{p_0}(\zeta)} \\[2pt]
					0 & 1
				\end{pmatrix},
				& z\in L_{p_0}^+,
			\end{cases}
		\end{equation}
		with $\alpha(p_0)=2\ln\delta(p_0)+i\phi_{p_0}$.
	\end{enumerate}
\end{problem}

More precisely, for $p_0\in\mathcal{P}_2$, we define
\begin{align}\label{eq:gg}
	&	N_{p_0}^{\mathrm{loc}}(\zeta)\xlongequal{\text{$p_0=E_{j_0},E_{n-j_0}$}}
	\begin{pmatrix}
		1 & 0 \\
		ia(\omega(p_0)) & 1
	\end{pmatrix}
	N^{P_{(34)}}\!\left(\zeta;s,\tfrac{1}{4},0\right)
	e^{-\left(\frac{2}{3}\zeta^{3/2}+\omega(p_0)\zeta^{1/2}\right)\sigma_3}
	G_{p_0}(\zeta),\\[4pt]\label{eq:ggg}
	&	N_{p_0}^{\mathrm{loc}}(\zeta)\xlongequal{\text{$p_0=\bar E_{j_0},\bar E_{n-j_0}$}}
	\begin{pmatrix}
		1 & 0 \\
		ia(\omega(p_0)) & 1
	\end{pmatrix}
	N^{P_{(34)}}\!\left(\zeta;s,-\tfrac{1}{4},0\right)
	e^{-\left(\frac{2}{3}\zeta^{3/2}+\omega(p_0)\zeta^{1/2}\right)\sigma_3}
	G_{p_0}(\zeta),
\end{align}
where, for $p_0\in\mathcal{P}_2$, we define
\begin{equation}\label{eq:G-unified}
	G_{p_0}(\zeta)=
	\begin{cases}
		\sigma_2^+\,e^{-i\widetilde\Theta(p_0)\sigma_3}\,e^{\lambda(p_0)\pi i\sigma_3/4},
		& \zeta\in\mathbb{C}^{+}, \\[12pt]
		e^{-i\widetilde\Theta(p_0)\sigma_3}\,e^{\lambda(p_0)\pi i\sigma_3/4},
		& \zeta\in\mathbb{C}^{-},
	\end{cases}
\end{equation}
with
\begin{equation}\label{eq:Theta-p0}
	\widetilde \Theta(p_0)=t\theta(p_0)+i\ln\delta(p_0)+\frac{\phi_{\iota(p_0)}}{2}-\frac{\mu(p_0)\pi}{2},
\end{equation}
and the index functions
\begin{equation*}\label{eq:index-funcs}
	\iota(p_0)=
	\begin{cases}
		j_0,&p_0\in\{E_{j_0},\bar E_{j_0}\},\\[2pt]
		n-j_0,&p_0\in\{E_{n-j_0},\bar E_{n-j_0}\},
	\end{cases}
	\,\,
	\mu(p_0)=
	\begin{cases}
		1,&p_0\in\{E_{j_0},\bar E_{j_0}\},\\[2pt]
		0,&p_0\in\{E_{n-j_0},\bar E_{n-j_0}\},
	\end{cases}
	\,\,
	\lambda(p_0)=
	\begin{cases}
		-1,&p_0\in\{E_{j_0},E_{n-j_0}\},\\[2pt]
		+1,&p_0\in\{\bar E_{j_0},\bar E_{n-j_0}\},
	\end{cases}
\end{equation*}
and $\sigma_2^+=\begin{pmatrix}0&-1\\1&0\end{pmatrix}$.
For $z\in U_{p_0}$, using \eqref{eq:gg}, \eqref{eq:ggg}, \eqref{eq:G-unified}, we obtain
\begin{equation}\label{eq:asymp-Ej0}
	N_{p_0}^{(\mathrm{loc})}(z)=
	H_{p_0}(z)\,t^{\sigma_3/6}
	\left(I+\frac{N_1^{p_0}(\zeta)}{\zeta}+\mathcal{O}\bigl(\zeta^{-2}\bigr)\right)
	\frac{\zeta^{-\sigma_3/4}}{\sqrt{2}}
	\bigl(I+i\sigma_1\bigr)\,G_{p_0}(\zeta),
\end{equation}
where $N_1^{p_0}(\zeta)$ is the coefficient in the large-$\zeta$ expansion of 
$N^{\mathrm{P}_{34}}(\zeta)$. 
Here the analytic prefactors are defined by
\begin{equation}\label{sub--hhhh-abcd}
	H_{p_0}(z)=
	G_{p_0}(\zeta(z))^{-1}\,\frac{1}{\sqrt{2}}\bigl(I-i\sigma_1\bigr)
	\Bigl(\bigl(z-p_0\bigr)\,|c_{j_0,2}(\xi,p_0)|^{2/3}\Bigr)^{\sigma_3/4}.
\end{equation}
From Proposition~\ref{prop:outer-exist}, it is immediate that 
$N_{p_0}^{\mathrm{loc}}(z)$ is uniquely solvable for large positive $t$, where
\begin{align}\label{eq:residue-sum}
	&\bigl(N_1^{E_{j_0}}\bigr)_{12}
	+\bigl(N_1^{\bar E_{j_0}}\bigr)_{12}
	+\bigl(N_1^{E_{n-j_0}}\bigr)_{12}
	+\bigl(N_1^{\bar E_{n-j_0}}\bigr)_{12} \notag\\[4pt]
	&\qquad
	=ia(\omega(E_{j_0}))+ia(\omega(\bar E_{j_0}))
	+ia(\omega(E_{n-j_0}))+ia(\omega(\bar E_{n-j_0}))
	+\mathcal{O}\bigl(t^{-1/3}\bigr),
\end{align}
with $a(\omega)$ and $\omega$ given in \eqref{P-34-a} and \eqref{eq:pmeggaaa}, respectively.

\subsection{The error problem: small-norm Riemann--Hilbert problem for $\mathcal{E}(z)$}

Using the outer parametrix $N^{(\mathrm{out})}(z)$ and the local parametrix 
$N^{(\mathrm{loc})}(z)$ constructed in RH problem~\ref{prop:outer-exist} and 
Problem~\ref{rh:loc}, respectively, we define the error matrix 
$\mathcal{E}(z)$ by
\begin{equation}\label{eq:E-def}
	\mathcal{E}(z)\coloneqq
	\begin{cases}
		N^{(1)}(z)\,N^{(\mathrm{out})}(z)^{-1}, & z\in\mathbb{C}\setminus U,\\[4pt]
		N^{(1)}(z)\,N^{(\mathrm{loc})}(z)^{-1}N^{(\mathrm{out})}(z)^{-1}, & z\in U,
	\end{cases}
\end{equation}
where the boundary 
$\partial U$ is oriented clockwise.  It is straightforward to verify that 
$\mathcal{E}(z)$ is analytic in $\mathbb{C}\setminus\Sigma^{\mathcal{E}}$ with
\begin{equation}\label{eq:Sigma-E}
	\Sigma^{\mathcal{E}}\coloneqq\partial U\cup\bigl(\Sigma^{(1)}\setminus U\bigr).
\end{equation}

\begin{problem}\label{prob:small-norm}
	Find a $2\times2$ matrix-valued function 
	$\mathcal{E}(z)$, analytic in $\mathbb{C}\setminus\Sigma^{\mathcal{E}}\to\mathbb{C}^{2\times2}$,
	such that:
	\begin{enumerate}
		\item $\mathcal{E}(z)=I+\mathcal{O}(z^{-1})$ as $|z|\to\infty$.
		\item For each $z\in\Sigma^{\mathcal{E}}$, the boundary values satisfy 
		$\mathcal{E}_+(z)=\mathcal{E}_-(z)J_{\mathcal{E}}(z)$, where
		\begin{equation}\label{eq:J-E}
			J_{\mathcal{E}}(z)=
			\begin{cases}
				N^{(\mathrm{out})}(z)\,J^{(1)}(z)\,N^{(\mathrm{out})}(z)^{-1}, 
				& z\in\Sigma^{(1)}\setminus U,\\
			N^{(\mathrm{out})}(z)		N^{(\mathrm{loc})}(z)\,N^{(\mathrm{out})}(z)^{-1}, 
				& z\in\partial U.
			\end{cases}
		\end{equation}
	\end{enumerate}
\end{problem}

Starting from \eqref{eq:J-E} and using the exponential decay of $J^{(1)}-I$ 
on $\Sigma^{(1)}\setminus U$ (cf.~\eqref{eq:bound}) together with the 
boundedness of $N^{(\mathrm{out})}$, one finds that
\begin{equation}\label{eq:J-E-estimate}
	\bigl|J_{\mathcal{E}}(z)-I\bigr|=
	\begin{cases}
		\mathcal{O}(e^{-ct}), & z\in\Sigma^{(1)}\setminus\bigl(U\cup\Upsilon_2\bigr),\\[4pt]
		\mathcal{O}(t^{-1/3}), & z\in\partial U,
	\end{cases}
\end{equation}
for some constant $c>0$.  Consequently,
\begin{equation}\label{eq:Lp-estimate}
	\bigl\|\langle\cdot\rangle^k\bigl(J_{\mathcal{E}}-I\bigr)\bigr\|_{L^p(\Sigma^{\mathcal{E}})}
	=\mathcal{O}(t^{-1/3}),\qquad p\in[1,\infty],\;k\ge0.
\end{equation}
The uniformly vanishing bound \eqref{eq:Lp-estimate} establishes 
Problem~\ref{prob:small-norm} as a small-norm RH problem, 
for which the existence and uniqueness of the solution is guaranteed by 
the standard theory \cite{DeiftZhou2003,DeiftZhou1994,Zhou1989}.  In fact, the 
solution admits the Cauchy integral representation
\begin{equation}\label{eq:E-Cauchy}
	\mathcal{E}(z)=I+\frac{1}{2\pi i}\int_{\Sigma^{\mathcal{E}}}
	\frac{(I+\eta(s))\bigl(J_{\mathcal{E}}(s)-I\bigr)}{s-z}\,ds,
\end{equation}
where $\eta\in L^2(\Sigma^{\mathcal{E}})$ is the unique solution of the 
singular integral equation
\begin{equation}\label{eq:singular-int}
	\bigl(1-C_{\Sigma^{\mathcal{E}}}\bigr)\eta=C_{J_{\mathcal{E}}}I,
\end{equation}
and $C_{\Sigma^{\mathcal{E}}}$ denotes the associated Cauchy operator.

In order to reconstruct the solution $u(x,t)$ of \eqref{eq:mkdv} we need 
the large-$z$ expansion of $\mathcal{E}(z)$.  Geometrically expanding 
$(s-z)^{-1}$ for large $z$ in \eqref{eq:E-Cauchy} yields
\begin{equation}\label{eq:E-large-z}
	\mathcal{E}(z)=I+z^{-1}\mathcal{E}_1+\mathcal{O}\bigl(z^{-2}\bigr),\qquad z\to\infty,
\end{equation}
where
\begin{equation}\label{eq:E1-integral}
	\mathcal{E}_1=-\frac{1}{2\pi i}\int_{\Sigma^{\mathcal{E}}}
	\bigl(I+\eta(s)\bigr)\bigl(J_{\mathcal{E}}(s)-I\bigr)\,ds.
\end{equation}
Since $J_{\mathcal{E}}-I$ is exponentially small on 
$\Sigma^{(1)}\setminus U$, the dominant contribution to 
\eqref{eq:E1-integral} comes from the clockwise-oriented boundaries 
$\partial U_{p_0}$.  Using the asymptotic formulae 
\eqref{eq:asymp-Ej0}--\eqref{sub--hhhh-abcd} together with the 
local relation $\zeta(z)\sim (it\,c_{j,2}(\xi,p_0))^{2/3}(z-p_0)$ on 
$\partial U_{p_0}$, one computes the residue contributions explicitly.  
The diagonal conjugation by $t^{\sigma_3/6}$ implies that the 
$(1,2)$-entry of $J_{\mathcal{E}}-I$ on $\partial U_{p_0}$ is precisely 
of order $\mathcal{O}(t^{-1/3})$, while all other entries are 
$\mathcal{O}(t^{-2/3})$ or smaller.  Therefore, to leading order,
\begin{equation}\label{eq:E1-leading}
	\mathcal{E}_1(x,t)=
	\frac{1}{t^{1/3}}		N^{(\mathrm{out})}(z)	\sum_{p_0\in\mathcal{P}_2}H_{p_0}(p_0)
	\begin{pmatrix}
		0 & i\displaystyle
		\frac{a(\omega(p_0))}{|c_{j_0,2}(\xi,p_0)|^{2/3}}
		\\[12pt]
		0 & 0
	\end{pmatrix}H^{-1}_{p_0}(p_0)		N^{(\mathrm{out})}(z)^{-1}
	+\mathcal{O}\bigl(t^{-2/3}\bigr),
\end{equation}
where the sum runs over the four endpoints 
$p_0\in\{E_{j_0},\bar E_{j_0},E_{n-j_0},\bar E_{n-j_0}\}$, and we have 
used \eqref{eq:residue-sum} to identify the Painlev\'e residue 
coefficients $a(\omega(p_0))$.

We are now ready to prove Theorem \ref{thm:main}. 
\begin{proof}[Proof of Theorem~\ref{thm:main}]
	Inverting the sequence of transformations \eqref{eq:N1}, \eqref{eq:N-1-decompose}, \eqref{eq:N2-out}, the solution of RH problem~\ref{RH2-2} is given by
	\begin{equation}\label{eq:N-recover}
		N(z)=\delta^{\sigma_3}(\infty)\mathcal{E}(z)e^{ig(\infty)\sigma_3}N_2^{(\mathrm{out})}(z)
		e^{-ig(z)\sigma_3}Y(z)^{-\sigma_3}\delta^{-\sigma_3}(\infty)G^{-1}(z).
	\end{equation}
	Taking $z\to\infty$, we have $G^{-1}(z)\to I$. The solution of \eqref{eq:mkdv}--\eqref{eq:q_0} can now be recovered using \eqref{eq:2.28}. Since
	\begin{equation}\label{eq:Y-asymp-proof}
		Y(z)^{-\sigma_3}=I+\frac{Y_1\sigma_3}{z}+\mathcal{O}\bigl(z^{-2}\bigr),\qquad 
		Y_1=2i\operatorname{Im}(\ell_2),
	\end{equation}
	we obtain
	\begin{equation}\label{eq:N-out-expansion}
		N(z)=\delta^{\sigma_3}(\infty)e^{ig(\infty)\sigma_3}\left(I+\frac{\mathcal{E}_1}{z}\right)\left(I+\frac{N_2^{(\mathrm{out})}}{z}\right)\left(I+\frac{ig^{(1)}\sigma_3}{z}\right)\left(I+\frac{Y_1\sigma_3}{z}\right)+\mathcal{O}\bigl(z^{-2}\bigr),
	\end{equation}
	and consequently the coefficient of $z^{-1}$ in the Laurent expansion of $N(z)$ is given by
	\begin{equation}\label{eq:N1-out}
		N_1=\delta^{\sigma_3}(\infty)e^{ig(\infty)\sigma_3}\bigl(\mathcal{E}_1+N_2^{(\mathrm{out})}+ig^{(1)}\sigma_3+Y_1\sigma_3\bigr)e^{-ig(\infty)\sigma_3}\delta^{-\sigma_3}(\infty).
	\end{equation}
	Using the reconstruction formula \eqref{eq:2.28} and Proposition~\ref{prop:outer-exist}, we have
	\begin{equation}\label{eq:u-recon}
		u(x,t)=2i e^{2i(xp_0+tq_0)}\left[\delta^{\sigma_3}(\infty)e^{ig(\infty)\sigma_3}\bigl(\mathcal{E}_1+N_2^{(\mathrm{out})}\bigr)e^{-ig(\infty)\sigma_3}\delta^{-\sigma_3}(\infty)\right]_{12}+\mathcal{O}\bigl(t^{-1/2}\bigr).
	\end{equation}
	Since $\delta^{\sigma_3}(\infty)$ and $e^{ig(\infty)\sigma_3}$ are diagonal, the $(1,2)$ entry is multiplied by $\delta^{2}(\infty)e^{2ig(\infty)}$, 
	applying Proposition~\ref{prop:outer-exist} to the first term and using \eqref{sub--hhhh-abcd}--\eqref{eq:E1-leading} to evaluate the second term, we obtain
	\begin{equation}\label{eq:u-final}
		\begin{aligned}
			u(x,t)&=2i e^{2i(xp_0+tq_0)}\delta^{2}(\infty)e^{2ig(\infty)}\bigl((\mathcal{E}_1)_{12}+(N_2^{(\mathrm{alg})})_{12}\bigr)+\mathcal{O}\bigl(t^{-1/2}\bigr)\\
			&=u^{(\mathrm{sol})}(x,t;\ell_3)+2i e^{2i(xp_0+tq_0)}\delta^{2}(\infty)e^{2ig(\infty)}\\
			&\qquad\times\Biggl(u^{(\mathrm{alg})}(x,t)+\frac{1}{t^{1/3}}\sum_{p_0\in\mathcal{P}_2}\frac{i\widetilde H_{p_0}(p_0)\,a(\omega(p_0))}{|c_{j_0,2}(\xi,p_0)|^{2/3}}\Biggr)+\mathcal{O}\bigl(t^{-1/2}\bigr),
		\end{aligned}
	\end{equation}
	where
	\begin{equation}\label{eq:tildehhhh}
		\widetilde H_{p_0}(z)=
		\begin{dcases}
			\begin{aligned}
				&\frac{|c_{j_0,2}(\xi,p_0)|^{1/3}}{2}\,(z-p_0)^{1/2}\\
				&\quad\times\Bigl(\mathcal{A}(p_0)^2 N^{(\mathrm{alg})}_{11}(z)^2-2 i N^{(\mathrm{alg})}_{11}(z) N^{(\mathrm{alg})}_{12}(z)-\mathcal{A}(p_0)^{-2} N^{(\mathrm{alg})}_{12}(z)^2\Bigr),
			\end{aligned}
			& p_0\in\{\bar E_{j_0},\bar E_{n-j_0}\},\\[16pt]
			\begin{aligned}
				&-\frac{|c_{j_0,2}(\xi,p_0)|^{1/3}}{2}\,(z-p_0)^{1/2}\\
				&\quad\times\Bigl(\mathcal{A}(p_0)^2 N^{(\mathrm{alg})}_{11}(z)^2-2 i N^{(\mathrm{alg})}_{11}(z) N^{(\mathrm{alg})}_{12}(z)-\mathcal{A}(p_0)^{-2} N^{(\mathrm{alg})}_{12}(z)^2\Bigr),
			\end{aligned}
			& p_0\in\{E_{j_0},E_{n-j_0}\},
		\end{dcases}
	\end{equation}
	with
	\begin{equation}\label{eq:A-p0}
		\mathcal{A}(p_0)=e^{i\widetilde\Theta(p_0)-i\lambda(p_0)\pi/4},
	\end{equation}
	and where the entries of $N^{(\mathrm{alg})}(z)$ are given explicitly by
	\begin{equation}\label{eq:Nalg-entries}
		N^{(\mathrm{alg})}_{11}(z)=\frac{\bigl(\nu(z)+\nu(z)^{-1}\bigr)\,\Lambda_{11}(z)}{2\,\Lambda_{11}(\infty)},
		\qquad
		N^{(\mathrm{alg})}_{12}(z)=\frac{\bigl(\nu(z)-\nu(z)^{-1}\bigr)\,\Lambda_{12}(z)}{2\,\Lambda_{11}(\infty)}.
	\end{equation}
	This completes the proof of Theorem~\ref{thm:main}.
\end{proof}
\begin{appendices}
	\section{The Painleve XXXIV $(P_{34})$ Parameterix}
The $P_{34}$ parametrix $N^{(P_{34})}(\zeta)=N^{(P_{34})}(\zeta;\omega,\gamma,\tau)$ is a $2\times2$ matrix-valued function depending on the parameters $\omega$, $\gamma$, and $\tau$. It satisfies the following RH problem.
\begin{problem}\label{rhp:P34}
	The $P_{34}$ parametrix $N^{(P_{34})}(\zeta)=N^{(P_{34})}(\zeta;\omega,\gamma,\tau)$ is a $2\times2$ matrix-valued function depending on the parameters $\omega$, $\gamma$, and $\tau$. It satisfies the following RH problem.
	\begin{enumerate}
		\item $N^{(P_{34})}(\zeta)$ is analytic for $\zeta\in\mathbb{C}\setminus\left\{\bigcup_{j=1}^4L_j\cup\{0\}\right\}$, where
		\[
		L_1=\mathbb{R}^{+},\quad L_2=e^{\frac{2\pi i}{3}}\mathbb{R}^{+},\quad L_3=e^{\pi i}\mathbb{R}^{+},\quad L_4=e^{-\frac{2\pi i}{3}}\mathbb{R}^{+},
		\]
		with the orientations as shown in Figure~\ref{t7tt}.
		\begin{figure}[htp]
			\centering
			\begin{tikzpicture}[scale=1.0, >=Stealth, font=\small]
				% 轮廓线：保持原始方向，箭头均位于线段中点
				\draw[very thick, postaction={decorate}, decoration={markings, mark=at position 0.5 with {\arrow{>}}}] (-4.5,0) -- (0,0);
				\draw[very thick, postaction={decorate}, decoration={markings, mark=at position 0.5 with {\arrow{>}}}] (0,0) -- (4.5,0);
				\draw[very thick, postaction={decorate}, decoration={markings, mark=at position 0.5 with {\arrow{>}}}] (-2.5,2.5) -- (0,0);
				\draw[very thick, postaction={decorate}, decoration={markings, mark=at position 0.5 with {\arrow{>}}}] (-2.5,-2.5) -- (0,0);
				
				% 原点标记
				\filldraw[black] (0,0) circle (1.2pt);
				\node[below right, inner sep=2pt] at (0,0) {$0$};
				
				% 区域标签
				\node at (1.5,1.0) {$D_1$};
				\node at (-1.5,1.0) {$D_2$};
				\node at (-1.5,-1.0) {$D_3$};
				\node at (1.5,-1.0) {$D_4$};
				
				% 轮廓标签（置于轮廓外侧，避免与线条重叠）
				\node[below, inner sep=4pt] at (2.25,0) {$l_1$};
				\node[left, inner sep=4pt] at (-1.6,1.5) {$l_2$};
				\node[above, inner sep=4pt] at (-2.25,0) {$l_3$};
				\node[left, inner sep=4pt] at (-1.6,-1.5) {$l_4$};
			\end{tikzpicture}
			\caption{Jump contours $l_j$ and regions $D_j$ in the RH problem for $N^{P_{34}}(z)$.}
			\label{t7tt}
		\end{figure}
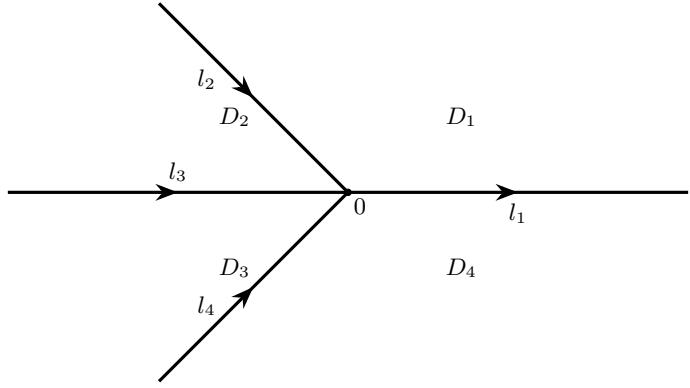
		\item $N^{(P_{34})}$ satisfies the jump condition $N_{+}^{(P_{34})}(\zeta)=N_{-}^{(P_{34})}(\zeta)J^{(P_{34})}(\zeta)$, with
		\[
		J^{(P_{34})}(\zeta)=
		\begin{cases}
		\begin{pmatrix}
				1 & \tau\\
				0 & 1
		\end{pmatrix},  \,\,\,\zeta\in l_1, \quad\,\,\
		\begin{pmatrix}
				1 & 0 \\
				e^{2\gamma\pi i} & 1
		\end{pmatrix}, \quad \zeta\in l_2, \\[12pt]
		\begin{pmatrix}
				0 & 1 \\
				-1 & 0
		\end{pmatrix},  \zeta\in l_3, \quad\,\,\,\,
		\begin{pmatrix}
				1 & 0 \\
				e^{-2\gamma\pi i} & 1
			\end{pmatrix}, \,\,\, \zeta\in l_4.
		\end{cases}
		\]
		
		\item As $\zeta\to\infty$, there exists a function $a(\omega)=a(\omega;\gamma,\tau)$ such that
		\[
		\begin{aligned}
			N^{(P_{34})}(\zeta)&=\left(\begin{array}{cc}
				1 & 0 \\
				-ia(\omega) & 1
			\end{array}\right)\left(I+\frac{N_1^{(P_{34})}(\omega)}{\zeta}+\mathcal{O}\bigl(\zeta^{-2}\bigr)\right) \times\frac{\zeta^{-\frac{1}{4}\sigma_3}}{\sqrt{2}}\begin{pmatrix}
				1 & i \\
				i & 1
		\end{pmatrix}e^{-\left(\frac{2}{3}\zeta^{3/2}+\omega\zeta^{1/2}\right)\sigma_3},
		\end{aligned}
		\]
		where we take the principal branch for the fractional powers and
		\begin{equation}\label{P-34-a}
		\left(N_1^{(P_{34})}\right)_{12}(\omega)=ia(\omega).
		\end{equation}
		\item[(d)] As $\zeta\to0$, we have, if $-1/2<\gamma<0$,
		\[
		N^{(P_{34})}(\zeta)=\mathcal{O}\bigl(\zeta^\gamma\bigr),
		\]
		and if $\gamma\geq 0$,
		\[
		N^{(P_{34})}(\zeta)=
		\begin{cases}
			\left(\begin{array}{ll}
				\mathcal{O}\bigl(\zeta^\gamma\bigr) & \mathcal{O}\bigl(\zeta^{-\gamma}\bigr) \\
				\mathcal{O}\bigl(\zeta^\gamma\bigr) & \mathcal{O}\bigl(\zeta^{-\gamma}\bigr)
			\end{array}\right), & \zeta\in\Omega_1\cup\Omega_4, \\[10pt]
			\mathcal{O}\bigl(\zeta^{-\gamma}\bigr), & \zeta\in\Omega_2\cup\Omega_3,
		\end{cases}
		\]
		where the regions $D_j$, $j=1,2,3,4$, are shown in Figure~\ref{t7tt}.
	\end{enumerate}
\end{problem}
By \cite{ItsKuijlaarsOstensson2008,ItsKuijlaarsOstensson2009,XuZhao2011}, the above RH problem is uniquely solvable for $\gamma>-1/2$, $\tau\in\mathbb{C}\setminus(-\infty,0)$, and $\omega\in\mathbb{R}$. Moreover, let $a(\omega)$ be given in \eqref{P-34-a} and define the function
\begin{equation}\label{eq:u-def}
	u(\omega):=u(\omega;\gamma,\tau)=a'(\omega;\gamma,\tau)-\frac{\omega}{2}.
\end{equation}
Then $u(\omega)$ satisfies the Painlev\'{e} XXXIV equation, namely,
\begin{equation}\label{eq:P34-ode}
	u''(\omega)=4u(\omega)^2+2\omega u(\omega)+\frac{u'(\omega)^2-(2\gamma)^2}{2u(\omega)},
\end{equation}
and is pole-free on the real axis. In particular, the following asymptotics hold:
\begin{equation}\label{eq:u-asymp}
	u(\omega;\gamma,0)=
	\begin{cases}
		\gamma/\sqrt{\omega}+\mathcal{O}\bigl(\omega^{-2}\bigr), & \omega\to+\infty, \\[6pt]
		-\omega/2+\mathcal{O}\bigl(\omega^{-2}\bigr), & \omega\to-\infty.
	\end{cases}
\end{equation}
This, together with the fact that $a(\omega;\gamma,0)\to0$ as $\omega\to-\infty$ ~\cite{DaiXuZhang2020}, implies that
\begin{equation}\label{eq:a-integral}
	a(\omega;\gamma,0)=\int_{-\infty}^{\omega}\left(u(t;\gamma,0)+\frac{t}{2}\right)\mathrm{d}t.
\end{equation}
		\end{appendices}
	
\subsection*{Acknowledgments}
The authors are deeply grateful to the anonymous reviewers for their meticulous review of the manuscript and for providing valuable comments and constructive suggestions that have significantly improved the quality of this work. This research was supported by the National Natural Science Foundation of China (Grant No.~12271104).

\subsection*{Ethical Statement}
This manuscript presents original research that has not been published previously and is not under consideration for publication in any other journal. The study has not been fragmented into separate submissions to increase the number of publications, and no part of this work has been submitted to multiple journals either simultaneously or in sequence.

\section*{Data Availability Statements}
All data generated or analyzed during this study are included in this published article.
\medskip
\let\em=\it
	\bibliographystyle{alpha}
	\bibliography{ref-BA}
\end{document}